\documentclass[twocolumn, 10pt]{article}

\newcounter{ag}
\addtocounter{ag}{1}
\newcommand{\ag}[1]{\textcolor{red}{AG-\arabic{ag}}\blankfootnote{\textcolor{red}{AG-\arabic{ag}: #1}}\addtocounter{ag}{1}}

\newcounter{ab}
\addtocounter{ab}{1}

\newcounter{sc}
\addtocounter{sc}{1}

\newcounter{ks}
\addtocounter{ks}{1}
\newcommand{\ks}[1]{\textcolor{teal}{KS-\arabic{ks}}\blankfootnote{\textcolor{teal}{KS-\arabic{ks}: #1}}\addtocounter{ks}{1}}

\newcounter{zk}
\addtocounter{zk}{1}

\newcommand{\data}{\ensuremath{\textbf{x}}\xspace}
\newcommand{\Data}{\ensuremath{\textbf{X}}\xspace}
\newcommand{\lap}{\ensuremath{\textsf{Lap}}\xspace}

\newcommand{\MW}{\ensuremath{\sf{MW_{\textup{stat}}}}\xspace}
\newcommand{\privMW}{\ensuremath{\sf{\widetilde{MW}}_{\textup{stat}}}\xspace}
\newcommand{\privMWP}{\ensuremath{\sf{\widetilde{MW}_{\textup{p}}}}\xspace}

\newcommand{\statabs}{\ensuremath{h_{abs}\xspace}}
\newcommand{\statsqu}{\ensuremath{h\xspace}}
\newcommand{\privsqu}{\ensuremath{\widetilde{h}\xspace}}
\newcommand{\KW}{\ensuremath{\sf{KW_{\textup{stat}}}}\xspace}
\newcommand{\KWA}{\ensuremath{\sf{KWabs_{\textup{stat}}}}\xspace}

\newcommand{\privKW}{\ensuremath{\sf{\widetilde{KW}_{\textup{stat}}}}\xspace}
\newcommand{\privKWA}{\ensuremath{\sf{\widetilde{KWabs}_{\textup{stat}}}}\xspace}
\newcommand{\privKWP}{\ensuremath{\sf{\widetilde{KW}_{\textup{p}}}}\xspace}
\newcommand{\privKWPA}{\ensuremath{\sf{\widetilde{KWabs}_{\textup{p}}}}\xspace}

\newcommand{\W}{\ensuremath{\sf{W_{\textup{stat}}}}\xspace}
\newcommand{\privW}{\ensuremath{\sf{\widetilde{W}_{\textup{stat}}}}\xspace}
\newcommand{\WP}{\ensuremath{\sf{WP}_{\textup{stat}}}\xspace}
\newcommand{\privWP}{\ensuremath{\sf{\widetilde{WP}_{\textup{stat}}}} \xspace}
\newcommand{\privWPP}{\ensuremath{\sf{\widetilde{WP}_{\textup{p}}}}\xspace}

\newcommand{\T}{\ensuremath{\sf{T_{\textup{stat}}}}\xspace}
\newcommand{\privT}{\ensuremath{\sf{\widetilde{T}_{\textup{stat}}}}\xspace}
\newcommand{\privTP}{\ensuremath{\sf{\widetilde{T}_{\textup{p}}}}\xspace}

\newcommand{\algrule}[1][.8pt]{\par\vskip.1\baselineskip\hrule height #1\par\vskip.1\baselineskip}
 \let\svthefootnote\thefootnote
\newcommand\blankfootnote[1]{%
  \let\thefootnote\relax\footnotetext{#1}%
  \let\thefootnote\svthefootnote%
}

\usepackage{amsmath}
\usepackage{hyperref}
\usepackage{color, enumitem}
\usepackage{xcolor}
\usepackage{caption}
\usepackage{graphicx}
\usepackage{verbatim}
\usepackage{float}
\usepackage{blindtext}
\usepackage[vlined]{algorithm2e}
\usepackage{comment}
\usepackage{placeins}
\usepackage{mathtools}
\usepackage{amsthm}
\usepackage{cite}
\usepackage{amsfonts}
\usepackage[margin=0.8in]{geometry}
\usepackage[utf8]{inputenc}
\usepackage{titlesec}

\newtheorem{theorem}{Theorem}[section]
\newtheorem{lemma}[theorem]{Lemma}
\theoremstyle{definition}
\newtheorem{definition}[theorem]{Definition}

\newtheorem{corollary}[theorem]{Corollary}
\newtheorem*{theorem*}{Theorem}

\DeclareMathSizes{10}{9}{7}{5}
\titleformat*{\section}{\LARGE\bfseries}
\titleformat*{\subsection}{\large\bfseries}
\titleformat*{\subsubsection}{\normalsize\bfseries}
\titleformat*{\paragraph}{\normalsize\bfseries}
\titleformat*{\subparagraph}{\normalsize\bfseries}

\begin{document}

\title{\LARGE \textbf{Differentially Private Nonparametric Hypothesis Testing}}

\author{\begin{tabular}{c c c c c}
Simon Couch & & Zeki Kazan & & Kaiyan Shi \\
\textit{\small couchs@reed.edu} & & \textit{\small kazanz@reed.edu} & & \textit{\small kaishi@reed.edu} \\
& Andrew Bray &  & Adam Groce & \\
& \textit{\small abray@reed.edu} & & \textit{\small agroce@reed.edu} &  \\[10pt]
\multicolumn{5}{c}{Reed College Mathematics Department} \\
\multicolumn{5}{c}{Portland, OR, USA} \\
\end{tabular}}

\date{}

\maketitle

\section*{Abstract}
Hypothesis tests are a crucial statistical tool for data mining and are the workhorse of scientific research in many fields.  Here we study differentially private tests of independence between a categorical and a continuous variable.  We take as our starting point traditional nonparametric tests, which require no distributional assumption (e.g., normality) about the data distribution.  We present private analogues of the Kruskal-Wallis, Mann-Whitney, and Wilcoxon signed-rank tests, as well as the parametric one-sample t-test.  These tests use novel test statistics developed specifically for the private setting.  We compare our tests to prior work, both on parametric and nonparametric tests.  We find that in all cases our new nonparametric tests achieve large improvements in statistical power, even when the assumptions of parametric tests are met.  

\section{Introduction}\label{sec:intro}

In 2011, researchers in Switzerland began an investigation into the link between methylation levels of a given gene and the occurance of schizophrenia and bipolar disorder\cite{carrard2011}. They recruited patients that suffered from psychosis as well as healthy controls and measured the level of methylation of the gene in each individual. They found that levels in the groups suffering from psychosis were higher than those in the healthy group. In order to rule out the possibility that this difference was due to sampling variability, they relied upon a suite of nonparameteric hypothesis tests to establish that this link likely exists in the general population.

While the results of these tests were published in an academic journal, the data itself is unavailable to preserve the privacy of the patients. This is a necessary consideration when working with sensitive data, but it hampers scientific reproducibility and the extension of this work by other researchers.

Our goal in this paper is to provide nonparametric hypothesis tests that satisfy differential privacy. The difficulty of developing a private test comes not just from the need to privately approximate a test statistic, but also from the need for an accurate reference distribution that will produce valid p-values. It is not sufficient to treat the approximate test statistic as equivalent to its non-private counterpart.

In this paper, we present several new hypothesis tests.  In all cases we are considering data sets with a categorical explanatory variable (e.g., membership in the schizophrenic, bipolar, or control group) and a continuous dependent variable (e.g., methylation level). Our goal when designing a hypothesis test is to maximize the \textit{statistical power} of the test, or equivalently to minimize the amount of data needed to detect a particular effect.

In the traditional public setting, there are two families of tests for these scenarios.  The more commonly used are \textit{parametric} tests that assume that within each group, the continuous variable follows a particular distribution (usually Gaussian).  An alternative to these tests are \textit{nonparametric} tests, which make no distributional assumption but, in exchange, have slightly lower power. Nonparametric tests generally rely upon substituting in, for each data point, the rank of the continuous variable relative to the rest of the sample.  The test statistic is then a function of these ranks rather than the original values.

The private hypothesis tests we propose all use rank-based test statistics. Our overarching argument in this paper, beyond the individual value of each of the tests we introduce, is that in the private setting these rank-based test statistics are \emph{more powerful} than the traditional parametric alternatives.  This is contrary to the public setting, where the parametric tests (when their assumptions are met) perform better.

Our second, broader point is that as a research community we need to support the development of hypothesis tests specifically tailored to the private setting.  Our private test statistics are not simply approximations of traditional test statistics from the public setting and as a result, we find that they can require an order of magnitude less data.  Current tests used by statisticians in the public setting have been refined through decades of incremental improvement, and the same sort of development needs to happen in the private setting.

\subsection{Our contributions}

We introduce several new private hypothesis tests that mirror the three most commonly used rank-based tests. In one setting, this is a private analog of the traditional public test statistic, but for the remaining two settings it is a new statistic developed specifically for the private setting.  The privacy of these statistics generally follows from non-trivial but reasonably straightforward applications of the Laplace mechanism.  Our main contribution is not the method of achieving privacy but the construction of novel private hypothesis tests with high statistical power.

There are two components to the construction of each hypothesis test. The first is the creation of a test statistic to capture the effect of interest while remaining provably private. The second is a method to learn the distribution of the statistic under the null hypothesis in order to compute p-values, which are the object of primary interest to researchers.

In particular, we develop tests for the following cases:

\paragraph{Three or more groups}  When the categorical variable divides the sample into three or more groups, the traditional public test is the one-way analysis of variance (ANOVA) in the parametric case and Kruskal-Wallace in the nonparametric case.  ANOVA has been previously studied by Campbell et al.~\cite{ANOVA} and then by Swanberg et al.~\cite{new_anova}, who improved the power by an order of magnitude.  We give the first private nonparametric test by modifying the rank-based Kruskal-Wallace test statistic for the private setting.  We provide experimental evidence that this statistic has dramatically higher power than a simple privatized version of the public Kruskal-Wallace statistic.  Moreover, we provide evidence that even in the parametric setting (i.e., when the data is normally distributed) the private rank-based statistic outperforms the private ANOVA test, in one representative case requiring only 23\% as much data to reach the same power.

\paragraph{Two groups}  In the two group setting, the most common public nonparametric test is the Mann-Whitney test.  We provide an algorithm to release an approximation of this statistic under differential privacy and a second algorithm to conduct the test and release a valid p-value.

In the public setting, the analogous parametric test is the two sample $t$-test which is equivalent to an ANOVA test when there are two groups.  We compare therefore to the private ANOVA test and show experimental evidence that our Mann-Whitney analog has significantly higher power.


\paragraph{Paired data}  We also consider the case where the two sets of data are in correspondence with each other (e.g., before-and-after measurements).  The nonparametric test in this case is the Wilcoxon signed-rank test, and it is the only nonparametric test that has previously been studied in the private setting\cite{DPWilcoxon}. We provide two improvements to the prior work.  First, we change the underlying statistic to the less well-known Pratt variant, which we find conforms more easily to the addition of noise.  Second, we show that our simulation method for computing reference distributions is more precise than the upper bounds used in the prior work (which contained an error which we identify and correct). The result is a significant improvement in the power of the test over the the prior work.

In parallel with the previous two scenarios, we then compare our private nonparametric test with a private version of the analogous parametric test, the paired $t$-test.  A direct private implementation of this test does not exist in the literature, so we propose one here.\footnote{In simultaneous work Gaboardi et al.~\cite{gaboardi2018locally} propose such a test, but our test is still higher power.  See Section \ref{sec:ttestexp} for more details.} In alignment with the previous results, we show experimental evidence that the rank-based test has superior power.

\medskip \noindent
For all our tests we give not just private test statistics but also precise methods of computing a reference distribution and a p-value, the final output practitioners actually need.  We also experimentally verify that the probability of Type 1 errors (incorrectly rejecting the null hypothesis) is acceptably low. We give careful power analyses and use these to compare tests to each other.  All our tests and experiments are implemented with publicly available code.\footnote{Our source code is available at: \textit{github.com/simonpcouch/non-pm-dpht}}

We find that rank-based statistics are very amenable to the private setting.  We also repeatedly find that what is optimal in the public setting is no longer optimal in the private setting.  We hope that this work contributes to the development of a standard set of powerful hypothesis tests that can be used by scientists to enable inferential analysis while protecting privacy.

\section{Background} \label{bkg}

In this section, we begin by discussing hypothesis testing in general and outlining the formalities of differential privacy. We then discuss the difficulties of hypothesis testing within the privacy framework and the previous work done in this area.  Each of our main results requires a more detailed discussion of prior work on that particular test or use case; we leave those discussions for Sections \ref{sec:kw}-\ref{sec:wc}.

\subsection{Hypothesis Testing}

The key inferential leap that is made in hypothesis testing is the claim that not only is the \emph{sample} of data incompatible with a particular scientific theory, but that it is evidence that this incompatible holds in a broader \emph{population}. In the study on psychotic disease, the researchers used this technique to generalize from their 165 subjects to the population of causasian-descended Swiss. The scientific theory that they refuted, that there is no link between methylation levels and psychosis, is called the \textit{null hypothesis} ($H_0$). 

To test whether or not the data is consistent with $H_0$, a researcher computes a \textit{test statistic}. The choice of a function $f$ to compute the test statistic largely determines the hypothesis test being used. For a random database $\Data$ drawn according to $H_0$, the distribution of the statistic $T = f(\Data)$ can be determined either analytically or through simulation. The researcher then computes a \textit{p-value}, the probability that the observed test statistic or more extreme would occur under $H_0$.

\begin{definition}
For a given test statistic $t=f(\data)$ and null hypothesis $H_0$, the \emph{p-value} is defined as
$$\Pr[T \geq t \mid T = f(\Data) \text{ and } \Data \leftarrow H_0] = p.$$
\end{definition}

If the function $f$ is well-chosen, then the more the underlying distribution of $\Data$ differs from the distribution under $H_0$, the more likely a low p-value becomes. Typically a threshold value $\alpha$ is chosen, such that we reject $H_0$ as a plausible explanation of the data when $p < \alpha$. The choice of $\alpha$ determines the \textit{type I error rate}, the probability of incorrectly rejecting a true null hypothesis. 

We define the \textit{critical value} $t^*$ to be the value of the test statistic $t$ where $p = \alpha$. We use this to define the \textit{statistical power}, a measure of how likely a hypothesis test is to pick up on a given effect (i.e. the chance of rejecting when the null hypothesis is false). The power is a function of how much the underlying distribution of $\Data$ differs from the distribution under $H_0$ as well as the size of the database.

\begin{definition}
For a given alternate data distribution $H_A$, the \emph{statistical power} of a hypothesis test is 
$$\Pr[T \geq t^* \mid T = f(\Data) \text{ and } \Data \leftarrow H_A].$$
\end{definition}

Statistical power is the accepted metric by which the statistics community judges the usefulness of a hypothesis test. It provides a common scale upon which to evaluate different tests for the same use case.

\subsection{Differential Privacy}

To persuade people to allow their personal data to be collected, data owners must protect information about specific individuals. Historically, ad-hoc database anonymization techniques have been used (i.e. changing names to numeric IDs, rounding geospatial coordinates to the nearest block, etc.), but these methods have repeatedly been shown to be ineffective \cite{anonyfail1, anonyfail2, anonyfail3}.

Differential privacy, proposed by Dwork et al. in 2006 \cite{DPDefinitions}, is a mathematically robust definition of privacy preservation, which guarantees that a query does not reveal anything about an individual as a consequence of their presence in the database. When the condition of differential privacy is satisfied, there is not \textit{much} difference between the output obtained from the original database, and that obtained from a database that differs by only one individual's data. Here we present $(\epsilon, \delta)$-differential privacy \cite{DPDefDelta}, which allows the closeness of output distributions to be measured with both a multiplicative and an additive factor.  However, for most of our hypothesis tests $\delta=0$.

\begin{definition}[Differential Privacy] \label{def:dp}
 A randomized algorithm $\tilde{f}$ on databases is $(\epsilon,\delta)$-differentially private if for all $\mathcal{S} \subseteq \textup{Range}(\tilde{f})$ and
for databases $\data, \data'$ that only differ in one row:
$$\Pr[\tilde{f}(\data) \in \mathcal{S}] \leq e^\epsilon \Pr[\tilde{f}(\data') \in \mathcal{S}]+\delta.$$
\end{definition}

We call databases $\data$ and $\data'$ \textit{neighboring} if they differ only in that a single row is altered (but not added or deleted), and we will use this notation in the following sections.\footnote{This is one of two roughly equivalent variants of differential privacy.  The key difference is that under this definition the size of the database is public knowledge.}



Differential privacy, like any acceptable privacy definition, is resistant to post processing. That is, if an algorithm is differentially private, an adversary with no access to the database will be unable to violate such privacy through further analysis (e.g., attempted deanonymization) of the query output \cite{DPDefinitions}.

\begin{theorem}[Post Processing] \label{thm:pp}
 Let $\tilde{f}$ be an $(\epsilon,\delta)$-differentially private randomized
algorithm. Let $g$ be an
arbitrary randomized algorithm. Then $g \circ \tilde{f}$ is $(\epsilon,\delta)$-
differentially private.
\end{theorem}

Theorem \ref{thm:pp} has another useful consequence. It allows us to develop private algorithms by first computing some private output, and then carrying out further computation on that output without accessing the database.  The additional computation need not be analyzed carefully---the final output of the additional analysis is automatically known to retain privacy.

Differential privacy requires the introduction of some randomness to any query output. A frequently used method is the \textit{Laplace mechanism}, introduced by Dwork et al.~\cite{DPDefinitions}. When given an arbitrary algorithm $f$ with real-valued output, this mechanism will add some noise drawn from the Laplace distribution to the output of the algorithm and release a noisy output. 

\begin{definition}[Laplace Distribution]\label{def:ld}
 The Laplace Distribution
centered at 0 with scale $b$ is the distribution with probability density function:$$ \lap(x|b)=\frac{1}{2b}\textup{exp}\Big(-\frac{|x|}{b}\Big).$$ We write $\lap(b)$ to denote the Laplace distribution with scale $b$.
\end{definition} 

The scale of the Laplace Distribution used to produce the noisy output depends on the global sensitivity of the given algorithm $f$, which is the maximum change on the output of $f$ that could result from the alteration of a single row.

\begin{definition}[Global sensitivity] \label{def:gs}
The global sensitivity of a function $f$  is:
$$GS_f = \underset{\data,\data'}{ \textup{max}}\ |f(\data)-f(\data')|,$$ where $\data$ and $\data'$ are neighbouring databases. 
\end{definition}

With computed sensitivity $GS_f$ and privacy parameter $\epsilon$, the Laplace mechanism applied to $f$ ensures $(\epsilon,0)$-differential privacy \cite{DPDefinitions}.

\begin{definition}[Laplace Mechanism]  \label{def:lm}
Given any function $f$, the Laplace mechanism is defined as:$$\tilde{f}(\data)= f(\data)+Y,$$ where $Y$ is drawn from $\lap(GS_f/\epsilon)$, and $GS_f$ is the global sensitivity of $f$.
\end{definition}

\begin{theorem} \label{thm:lm}
(Laplace Mechanism) The Laplace mechanism preserves $(\epsilon,0)$ differential privacy.\end{theorem}

Global sensitivity is the maximum effect that can be caused by changing a single row of any database.  Sometimes it is helpful to talk about local sensitivity \textit{for a given database $\data$} \cite{LocalSensitivityDef}.  This is the maximum effect that can be caused by changing a row of that particular database.

\begin{definition}[Local Sensitivity] \label{def:ls}
 The sensitivity of a function $f$ at a particular database $\data$ is:
$$LS_f(\data) = \underset{\data'}{ \textup{max}}\ |f(\data)-f(\data')|,$$ where $\data'$ is a neighboring database. 
\end{definition}

Note that $GS_f = \max_\data LS_f(\data)$. Local sensitivity cannot simply be used in the Laplace mechanism in place of global sensitivity, because local sensitivity itself is a function of the database and therefore cannot be released.  But private upper bounds on local sensitivity can be used to create similar mechanisms that do preserve privacy, and one of our algorithms uses just such a technique.

Choosing $\epsilon$ is an important consideration when using differential privacy. We consider several values of $\epsilon$ throughout our power analyses.  The lowest, .01 is an extremely conservative privacy parameter and allows for safe composition with many other queries of comparable $\epsilon$ value. We also use $\epsilon$s of $.1$ and $1$, which, while higher, still provide very meaningful privacy protection.  Ultimately, the choice of $\epsilon$ is a question of policy and depends on the relative importance with which privacy and utility are regarded.  We also measure, for comparison, the power of the public versions of each test (equivalent to an $\epsilon$ of $\infty$). As one might expect, the amount of data needed to detect a given effect often scales roughly with the inverse of $\epsilon$.

\subsection{Differentially Private Hypothesis Testing}

Performing hypothesis tests within the framework of differential privacy introduces new complexity. A function to compute a private test statistic (be it a private version of a standard test statistic or an entirely new test statistic) is not useful on its own. We need a p-value or other understandable output, and that means understanding the reference distribution (i.e., the distribution of the statistic given $H_0$).

In classical statistics, test statistics are computed with deterministic functions.  The randomness added to the test statistic in order to privatize it introduces new complexity.  Most importantly, it causes the reference distribution to change. One \textit{cannot} simply compare the private test statistic to the usual reference distribution, as the addition of noise can inflate the type I error well above acceptable levels \cite{ANOVA}.

Because of this, a complete differentially private hypothesis test requires not only a function for computing a private test statistic, but also a method for determining its null distribution. Often the exact reference distribution cannot be determined, so worst-case reference distributions or upper bounds on the resulting critical value must be used, and the precision of this reference distribution can have a large effect on the resulting power.

The goal of differentially private hypothesis test design is to develop a test with power as close as possible to the public test.

\subsection{Related Work}

There is a substantial and growing literature on differentially private hypothesis testing. One area of research is the study of the rate of convergence of private statistics to the distributions of their public analogues \cite{smith2008efficient, smith2011privacy, wasserman2010statistical}. These papers do not offer practical, implementable tests and discussion of reference distributions when the noise is not yet negligible is often limited or entirely absent. Further, the results are often entirely asymptotic, without regard for constants that may prove to be problematic.

The chi-squared test, which tests the independence of two categorical variables,\footnote{This is the chi-squared test of independence. There are several related tests that use the same statistic, the chi-squared.} has been the subject of much study, resulting in the development of many private variants. One of these works, that of Vu and Slavkovic \cite{vu2009differential}, provides methods for calculation of accurate p-values adjusted for the addition of Laplace noise for differentially private single proportion and chi-squared tests specifically for clinical trial data. Several other papers, though they make asymptotic arguments on the uniformity of their p-values, have developed frameworks for private chi-squared tests specifically for the intent of genome-wide association study (GWAS) data \cite{fienberg2011privacy, johnson2013privacy, uhlerop2013privacy}. For these same tests, Monte Carlo simulation has been shown to offer more precise analysis in some cases \cite{gaboardi2016differentially, wang2015revisiting}. There has also been work, like that of Rogers and Kifer \cite{rogers2017new}, that proposes entirely new test statistics with asymptotic distributions more similar to their public counterparts.

While the development of private test statistics has achieved much attention, careful evaluations of statistical power of these new test statistics is not always demonstrated. This is unfortunate, as the cost of privacy (utility loss) must be accurately quantified in order for the widespread adoption or implementation of any of these methods. Fortunately, rigorous power analysis seems to be more common in recent work. Awan and Slavkovic recently presented a test for simple binomial data \cite{awan2018differentially}.  While the setting is the simplest possible, their paper gives what we believe is the first private test to come with a proof of optimality, something normally very difficult to achieve even in the public setting.

The body of work on numerical (rather than categorical) methods is less extensive but has been growing quickly in recent years. In 2017, Nguyen and Hui proposed algorithms for survival analysis methods \cite{nguyen2017differentially}. There have been frameworks developed for testing the difference in means of normal distributions \cite{ding2018comparing, d2015differential}, and for testing whether a sample is consistent with a normal distribution with a particular mean \cite{solea2014differentially}. Differentially private versions of linear regressions, a class of inference that is extremely common in many fields both within and outside of academia, have received a notable level of attention, but the treatment of regression coefficients as test statistics has come about only recently \cite{barrientos2017differentially, sheffet2015differentially}. Two works have studied differentially private versions of one-way analysis of variance (ANOVA) \cite{ANOVA, new_anova}. The only prior work done on nonparametric hypothesis tests, as far as we are aware, is on the Wilcoxon signed-rank test by Task and Clifton in 2016 \cite{DPWilcoxon}.  Prior work specifically relevant to the tests we are proposing will be discussed in more detail in the relevant section.

\section{Many groups}\label{sec:kw}

We first consider the most general case, where we wish to distinguish whether many groups share the same distribution on a continuous variable.  The standard parametric test in the public setting is the one-way analysis of variance (ANOVA), which tests the equality of means across many groups.  Private ANOVA has been studied previously first by Campbell et al.~\cite{ANOVA} and then by Swanberg et al.~\cite{new_anova}, who improved the power by an order of magnitude.  The standard nonparametric test in the public setting is the Kruskal-Wallis test, which was used by the pschosis research group to determine that subjects in the schizophrenia, bipolar, and control groups had different methylation levels at a particular gene site\cite{carrard2011}. As is standard for nonparametric statistics in the public setting, it sacrifices some power compared to ANOVA but no longer assumes normally distributed data. \cite{ANOVAvsKW}

In this section we present two tests.  The first is a straightforward privatization of the standard Kruskal-Wallis test statistic.  The second modifies the statistic, essentially by linearizing the implied distance metric.  We find first that our modified statistic has much higher power.  We then further show that our modified statistic has much higher power than the ANOVA test of Swanberg et al.~\textit{even when the data is normally distributed}.  

\subsection{The Kruskal-Wallis test}\label{sec:kw_bg}

The Kruskal-Wallis test, proposed by William Kruskal and W. Allen Wallis in 1952 \cite{kruskal1952use}, is used to determine if several groups share the same distribution in a continuous variable. The only assumptions are that the data are drawn randomly and independently from a distribution with at least an ordinal scale.

Take a database $\data$ with $g$ groups\footnote{Throughout the paper we assume $g$ is public and independent of the data, so we don't list it as a separate input.  Because $g$ is the number of \textit{valid} groups, one or more of the $g$ groups might not contain any observations.  Allowing many valid groups that have no actual observations artificially increases the critical value, so it can reduce the power of our tests but does not affect the validity or privacy of the output.} and $n$ rows. Let $n_i$ be the size of each group and $r_{ij}$ be the rank of the $j^\text{th}$ element of group $i$. (If values are equal for several elements, all are given a rank equal to the average rank for that set.) We define $\bar{r}_{i}$ to be $\frac{1}{n_i}\sum_{j=1}^{n_i}r_{ij}$, the mean rank of group $i$, and $\bar{r}$ to be  $\frac{n+1}{2}$, the average of all the ranks. Then, the Kruskal-Wallis $\statsqu$-statistic is defined to be
$$\statsqu = (n - 1) \frac{\sum_{i=1}^g n_i(\bar{r}_i - \bar{r})^2}{\sum_{i=1}^g \sum_{j=1}^{n_i} (r_{ij} - \bar{r})^2}.$$
If there are no ties in the database, the denominator is constant and the formula can be simplified to
$$\statsqu = \frac{12}{n(n+1)}\sum_{i=1}^g n_i\bar{r}_i^2 - 3(n + 1).$$

For clarity and consistency with later sections, we present this calculation as an algorithm.  In general we use a subscript ``stat'' to label the algorithm computing a test statistic and a subscript ``p'' to denote the fully hypothesis test that outputs a p-value.  We use tildes to indicate private algorithms.

\begin{algorithm}
\DontPrintSemicolon
\algrule
\textbf{Algorithm} $\KW$ \textbf{:} Kruskal-Wallis Test Statistic \;
\algrule
\KwIn{$\data$}
    \For{group $i$ of $\data$}{
        $\bar{r}_{i} \longleftarrow \Big(\sum_{j=1}^{n_i}r_{ij}\Big)/n_i$
    }
    $\statsqu \longleftarrow \frac{12}{n(n+1)}\sum_{i=1}^{g}n_i\bar{r}_i -3(n+1) $\;
\KwOut{$\statsqu$}
\algrule
\label{alg:kw}
\end{algorithm}

\subsection{Privatized Kruskal-Wallis}

In this section, we bound the sensitivity of \KW, allowing us to create a private version. We then present a complete algorithm for calculating a p-value and prove that it too is differentially private.  We begin with the following sensitivity claim, proved in  Appendix \ref{sec:kw_sq_sens}.

\begin{theorem} \label{thm:squaresensi}
The sensitivity of $\KW$ is bounded by 87.
\end{theorem}

We are using the simplified formula that assumes there are no ties in the data, so our algorithm begins by adding a small amount of random noise to each data point to randomly order any ties. We may then compute the $\statsqu$-statistic as in the public setting and add noise proportional to the sensitivity.

\begin{algorithm}[!htb]
\DontPrintSemicolon
\algrule
\textbf{Algorithm} $\privKW$ \textbf{:} Private Kruskal-Wallis Test Statistic\;
\algrule
\KwIn{$\data$, $\epsilon$}
    Rank all data points, randomly breaking ties\;
    $\statsqu \longleftarrow \KW(\data)$\;
    $\widetilde{\statsqu} \longleftarrow \statsqu + \lap \Big (87/\epsilon \Big )$\;
\KwOut{$\privsqu$}
\algrule 
\end{algorithm}

\begin{theorem} \label{thm:pKW}
    Algorithm $\privKW$ is $\epsilon$-differentially private.
\end{theorem}

See Appendix \ref{sec:kw_privacy_pf} for the proof.

\begin{algorithm} [!htb]
\DontPrintSemicolon
\algrule
\textbf{Algorithm} $\privKWP$ \textbf{:} Complete Kruskal-Wallis Test\;
\algrule
\KwIn{\data, $\epsilon$, $z$}
	$\widetilde{\statsqu} \longleftarrow \privKW(\data, \epsilon)$\;
    \For{$k = 1$ \textup{to} $z$}{
        $\data^* \longleftarrow $ a database with independent uniform values from [0,1], divided almost equally into $g$ groups\;
        $h_k \longleftarrow \privKW(\data^*)$;
    }
    $p \longleftarrow$ fraction of $h_k$ values greater than $\privsqu$ \;
\KwOut{$\widetilde{\statsqu}, p$}
\algrule
\label{alg:complete_kw}
\end{algorithm}

Algorithm $\privKWP$ is our complete algorithm to find a p-value given a database $\data$, privacy parameter $\epsilon$. First a private test statistic $\widetilde{\statsqu}$ is computed.  Then the reference distribution is approximated by simulating $z$ databases under $H_0$ and computing the test statistic for each.\footnote{When we use the traditional Kruskal-Wallis test, the distribution of $h$-statistics asymptotically converges to the $\chi^2$ distribution. Thus, for efficiency purposes, we sample $h_k$ from ${\sf \chi^2}(g-1) + \lap(\Delta \statsqu/\epsilon)$}  (The distribution of the test statistic is independent of the distribution of data between groups and the distribution of the i.i.d.~data points, so our choice of equal-sized groups and uniform data from $[0,1]$ is arbitrary.)  The p-value is the percent of $h_k$ more extreme than $\widetilde{\statsqu}$. 


\begin{theorem} \label{thm:pcomp}
Algorithm $\privKWP$ is $\epsilon$-differentially private.
\end{theorem}

\begin{proof}
By Theorem \ref{thm:pKW}, the computation of $\privsqu$ is $\epsilon$-differentially private. All of the following steps (generating the reference distribution and calculating $p$-value) do not need to access to the database $\data$, and therefore by Theorem \ref{thm:pp} (post processing), Algorithm $\privKWP$ is $\epsilon$-differentially private.
\end{proof}

\subsection{A New Test: Absolute Value Kruskal-Wallis}

We now introduce our own new test, specifically designed for the private setting.  Inspired by Swanberg et al.~\cite{new_anova}, we alter the Kruskal-Wallis statistic, measuring distance with the absolute value instead of the square of the differences.  This statistic is now
$$\statabs = (n - 1) \frac{\sum_{i=1}^g n_i | \bar{r}_i - \bar{r} | }{\sum_{i=1}^g \sum_{j=1}^{n_i} | r_{ij} - \bar{r} | }.$$

As before, when there are no ties in the data, the statistic can be simplified.  (See Appendix \ref{sec:kw_abs_deriv} for the calculation.) In this case, the form depends on the parity of $n$.

\[ 
\statabs = 
\begin{dcases}
\frac{4(n-1)}{n^2}\sum_{i=1}^{g}n_i \left|\bar{r}_i-\frac{n+1}{2} \right|, &\mbox{if } n \mbox{ is even} \\
\frac{4}{n+1}\sum_{i=1}^{g}n_i \left|\bar{r}_i-\frac{n+1}{2} \right|, &\mbox{if } n \mbox{ is odd} \\
\end{dcases}  
\]

We call the algorithm to calculate the $\statabs$ test statistic $\KWA$. This statistic is preferable for two reasons.  First, it has lower sensitivity.  The following theorem is prooved in Appendix \ref{sec:kw_abs_sens}.

\begin{theorem} \label{thm:abssensi}
The sensitivity of $\KWA$ is bounded by 8.
\end{theorem}

Second, the actual values for $\statabs$ are significantly higher than they are for $h$, so any given amount of noise is less likely to overwhelm the value.

Because of space constraints, we don't give pseudocode for this hypothesis test, but it follows exactly that of the previous test.  The privatized statistic is computed by \privKWA, which adds Laplace noise in the same way as for \privKW, but scaled to the lower sensitivity.  The full hypothesis test, \privKWPA, computes the $p$ value in the same way as was done for \privKWP.\footnote{Unlike before, a ${\sf \chi^2}$ approximation cannot be used.}


\begin{theorem}
Algorithm $\privKWA$ and $\privKWPA$ are $\epsilon$-differentially private.
\end{theorem}

\begin{proof}
The proof is identical to the proofs for Algorithms $\privKW$ and $\privKWP$ (Theorems \ref{thm:pKW} and \ref{thm:pcomp}).
\end{proof}

\paragraph{Unequal group sizes}  The traditional $\statsqu$ statistic (and therefore the noisy private analogue) has a reference distribution that is independent of the allocation of observations between groups.  This is unfortunately not true for our new $\statabs$ statistic.  Fortunately, it seems that the worst-case distribution (i.e., the one resulting in the highest critical value) occurs when all groups are of equal size.  (We present both theoretical and experimental evidence for this in Appendix \ref{sec:kw_abs_t1}.)  As a result, it is safe to always equal-sized groups when simulating a reference distribution, though for very unequal group sizes, there will be a significant loss in power compared to a hypothetical where group sizes were known.  (If approximate group sizes are known publicly or released through other queries, those could be used instead when simulating the reference distribution.)

\subsection{Experimental Results}

\paragraph{Power analysis} We now assess the power of our \privKWPA test on synthetic data (See Appendix \ref{sec:kw_appendix} for an application to real-world data.) We generate many databases of data distributed with specified parameters and then run $\privKWPA$ on each. The power of the test for a given set of parameters is the proportion of times $\privKWPA$ returns a significant result (i.e. a p-value less than the significance level $\alpha$, generally set at 0.05). We use three groups of normally-distributed data, separated by steps of one standard deviation (so the highest and lowest groups differ by two standard deviations).  In our captions we denote the mean of group $i$ with $\mu_i$.

\begin{figure} [!htb] 
    \centering
    \includegraphics[width=\linewidth, keepaspectratio]{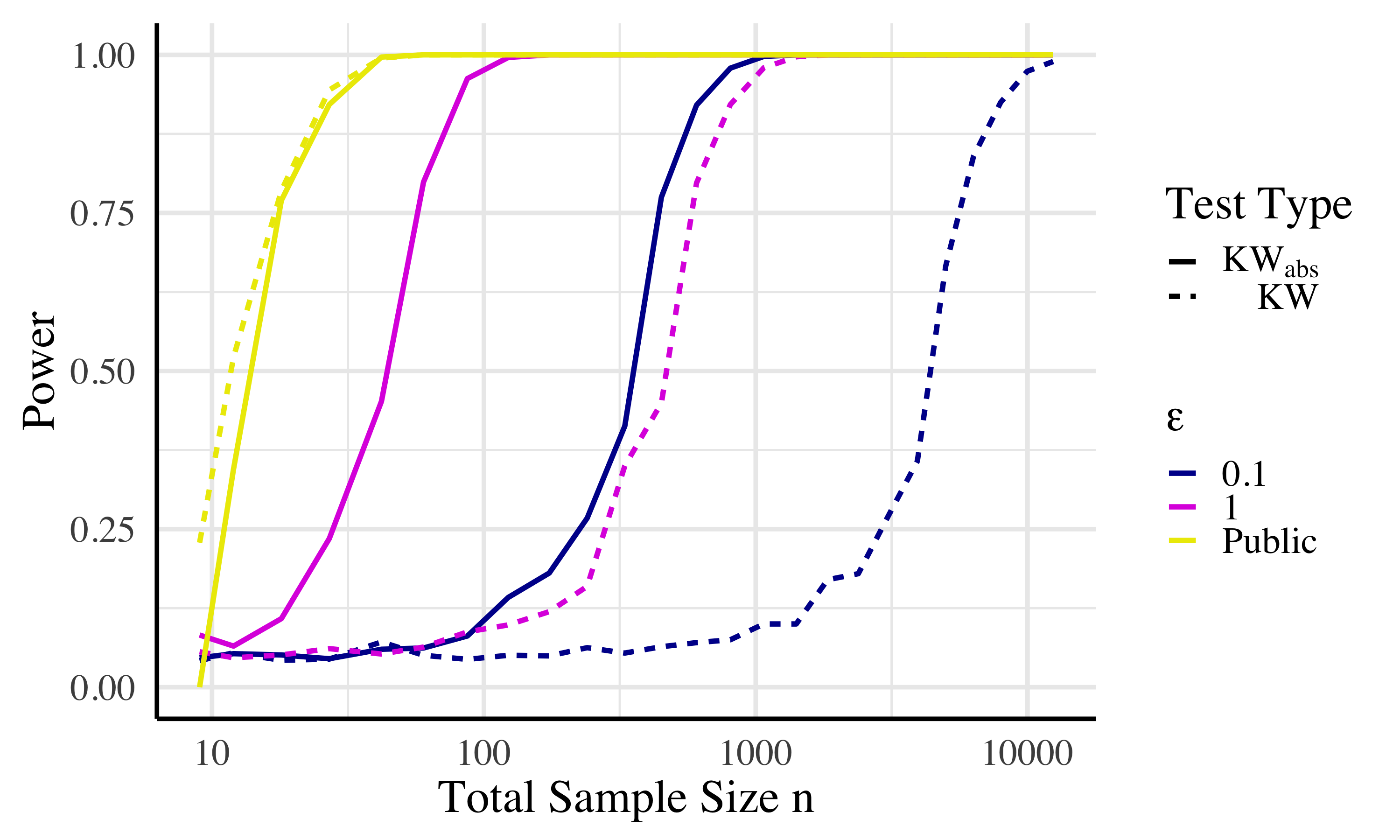}
    \caption{Power of $\privKWPA$ at various values of $\epsilon$ and total sample size $n$. (Effect size: $max_i(\mu_{i}) - min_i(\{\mu_{i}) = 2 \sigma$; $g = 3$; $\alpha=.05$; equal group sizes; normally distributed sample data)}\label{fig:no_ties}
\end{figure}

As shown in Figure \ref{fig:no_ties}, our private absolute value test variant requires significantly less data points than the original private test to reach the same power. Thus, from here on, we only evaluate the power of the absolute value variant. Figure \ref{fig:no_ties} also shows that, at an $\epsilon$ of 1, our private absolute value test only requires a database around a factor of 3 larger than the public test needs. 

\paragraph{Uniformity of p-values}  If a test is correctly designed, the probability of type 1 error (i.e., rejecting the null hypothesis when it is correct) should be less than or equal to $\alpha$ for any chosen value of $\alpha$. Comparing the fit of a large number of simulated p-values generated from null distributions to the uniform distribution on the unit interval allows one to evaluate the uniformity of p-values for a given hypothesis test. A common tool to carry this procedure out, the quantile-quantile (or Q-Q) plot, plots the quantiles of the uniform, theoretical distribution against the quantiles of the p-values. The theoretical and emperical quantiles will be nearly equal at all quantiles when the p-values follow the theoretical distribution, resulting in a linear trend on the Q-Q plot. A convex Q-Q plot indicates an increase in the type II error rate (i.e. the test not rejecting the null hypothesis when it is indeed not true, causing a decrease in power) which is acceptable but undesirable, while a concave Q-Q plot indicates an exceedingly high type I error rate (i.e. the test rejecting the null hypothesis when it is true, causing undue increases in power) which is not acceptable. Figure \ref{fig:qqplot_kw} demonstrates the p-value uniformity of \privKWPA. See Appendix \ref{sec:kw_appendix} for a discussion of uniformity of p-values with unequal group sizes.

\begin{figure} [!htb] 
    \centering
    \includegraphics[width=\linewidth, keepaspectratio]{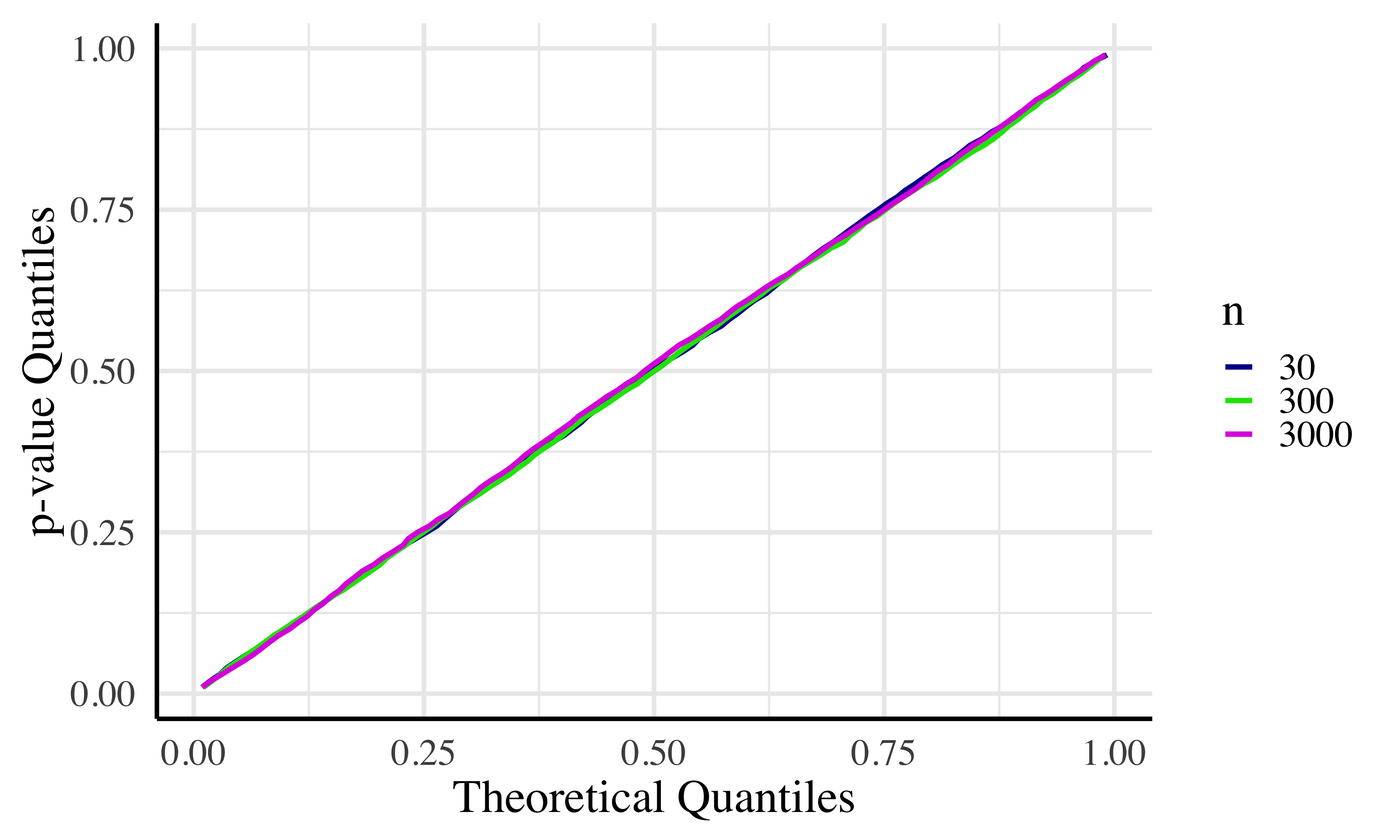}
    \caption{A quantile-quantile plot of \privKWPA comparing the distribution of simulated p-values to the uniform distribution at varying $n$, all with equal group sizes. ($g = 3$; $\epsilon = 1$)}\label{fig:qqplot_kw}
\end{figure}


\paragraph{Comparison to previous work} The only prior work on hypothesis testing for independence of two variables, one continuous and one categorical, is that on ANOVA.  The best private ANOVA analogue is that of Swanberg et al.~\cite{new_anova}.  In Figure \ref{fig:anova_comp} we compare \privKWPA to their test and we find its power to be much greater.  To get 80\% power with this effect size, our test requires only 23\% as much data as the private ANOVA test.  (The effect size used is the same as in \cite{new_anova}.)  We stress that this means our test is significantly higher-power, \textit{in addition} to being usable for non-normal data.  The test of Swanberg et al.~also requires that the analyst issuing the query can accurately bound the range of the data---a bound that is too tight or too loose will reduce the power of the test.  Our test works for data with unknown range.


\begin{figure} [!htb] 
    \centering
    \includegraphics[width=\linewidth, keepaspectratio]{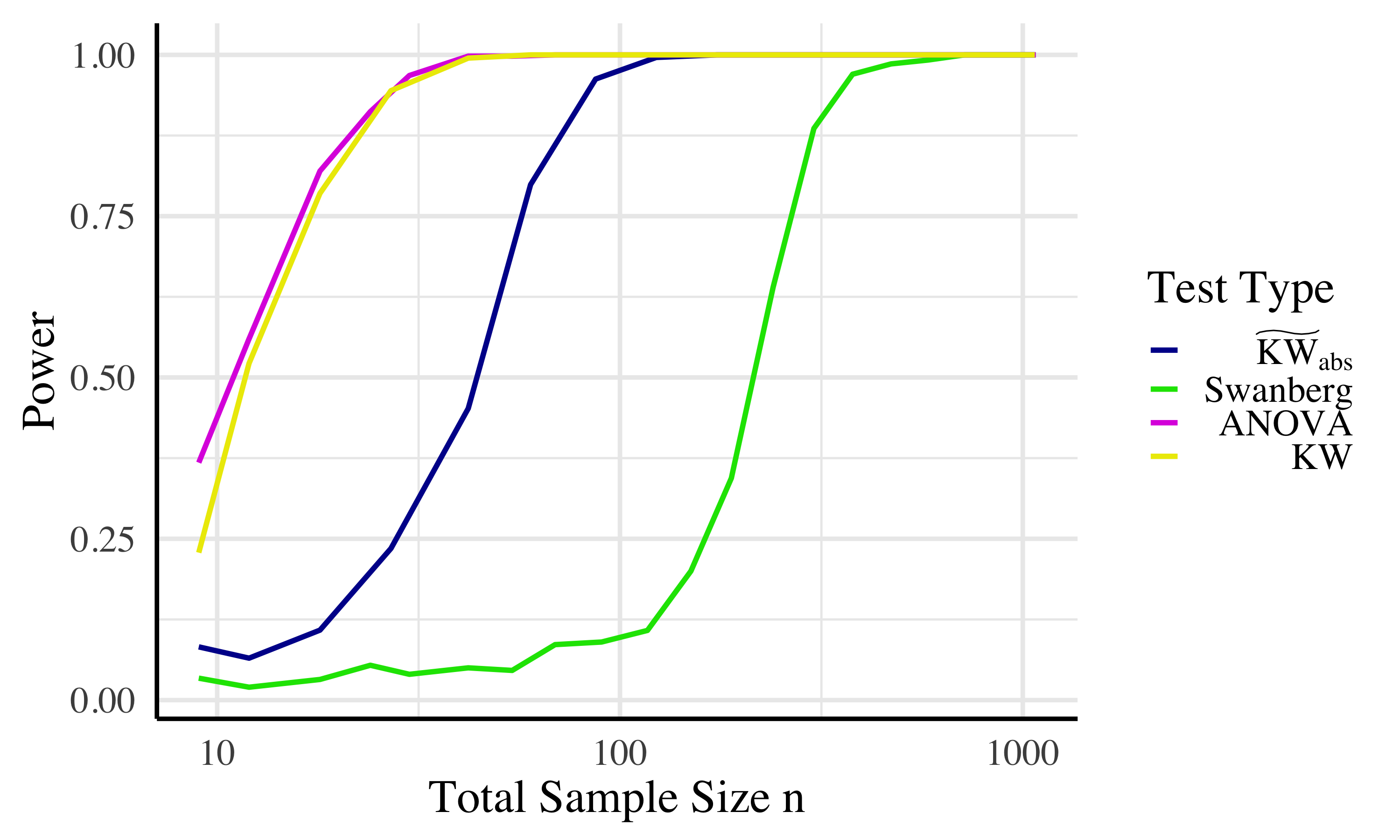}
    \caption{Power of \privKWPA, Swanberg et. al.'s test \cite{new_anova}, and the public tests at various $n$. (Effect size: $max(\mu_{n}) - min(\mu_{n}) = 2 \sigma$; $\epsilon = 1$; $g = 3$; $\alpha=.05$; equal group sizes; continuous sample data)}\label{fig:anova_comp}
\end{figure}

\paragraph{Robustness of results}  Though it is unusual, it is possible that the relative power of different hypothesis tests could change when different effect sizes are considered.  Therefore we repeat the analysis shown in Figure \ref{fig:no_ties} with a variety of different effect sizes, group sizes, and number of groups.  We also vary the frequency of tied values in the data, since the random ordering of tied values adds additional noise for our statistic.  Finally, we run the comparison on real data comparing income and age.  The results of these experiments are shown in Appendix \ref{sec:kw_appendix}. We find that the results discussed above are consistent across these variations.

\section{Two Groups}\label{sec:mw}


We now consider the case of data with only two groups (e.g., restricting our comparison to the methylation levels of the bipolar subjects versus the healthy controls.)  In the public nonparametric setting, one could simply use Kruskal-Wallis with $g=2$, but one can also use the Mann-Whitney $U$-test (also called the Wilcoxon \textit{rank-sum} test), proposed in 1945 by Frank Wilcoxon \cite{WC} and formalized in 1947 by Henry Mann and Donald Whitney \cite{MW}.  In this section we construct a private version of the Mann-Whitney test and compare it to simply using \privKWPA with $g=2$.

The standard parametric test in the public setting is the two-sample t-test.  We know of three prior works that can, in some sense, be seen as providing an analogue of the two-sample t-test for the standard private setting.  The only one for which this is an explicit goal is that of D'Orazio et al.~\cite{d2015differential}.  This test releases private estimates of the difference in means and of the within-group variance and produces a confidence interval rather than a p-value. (The difference in means is done with simple Laplace noise, while the variance estimate uses a subsample-and-aggregate algorithm.) Most importantly, they assume that the size of the two groups is public knowledge, where we treat the categorical value of a data point (ex., schizophrenic or not) to be private data.

There are two other works we know of that provide a private analogue of the two-sample t-test as a result of a slightly different goal.  The first is Ding et al.~\cite{ding2018comparing}, who give a test under the more restrictive \textit{local} differential privacy definition.  This test is of course also private under the standard differential privacy definition.  The other work is that of Swanberg et al.~\cite{new_anova}, who give a private analogue of the ANOVA test, as discussed previously.  In the public setting, ANOVA with $g=2$ is equivalent to the two-sample t-test.  Based on (somewhat incomparable) experiments in their respective papers, it appears that the Swanberg et al.~test is much higher power, which is unsurprising given that it was developed for the centralized database model of privacy.  We therefore compare our work to this. 

To our knowledge, there is no prior work specifically on a private version of the Mann-Whitney test.  As before, we find that our rank-based nonparametric tests outperform the private parametric equivalent even when the data is normally distributed.  We also find that, unlike in the public setting, the more generic Kruskal-Wallace analogue (used with $g=2$) outperforms the more purpose-built test.

\subsection{The Mann-Whitney test}

The function used to calculate the Mann-Whitney $U$ statistic is formalized in Algorithm \MW. As before, $\data$ is a database of size $n$, with $r_{ij}$ being the rank of the $j^{\text{th}}$ data point in group $i$. A statistic is first calculated for each group by summing the rankings in that group and subtracting a term depending on the group size. We then take the minimum of the two statistics to get $U$.  Compared to the other statistics we are considering, the directionality of $U$ is reversed --- low values are inconsistent with the null hypothesis and cause rejection, rather than high values.

\begin{algorithm} [h]
\DontPrintSemicolon
\algrule
\textbf{Algorithm} \MW \textbf{:} Mann-Whitney Test Statistic\;
\algrule
\KwIn{\data}
    \For{$i \in \{1, 2\}$}{
        $R_i \longleftarrow \sum_{j} r_{ij}$\;
        $U_i \longleftarrow R_i - \frac{n_i(n_i + 1)}{2}$ \;
    }
    $U \longleftarrow \min \{ U_1, U_2 \}$ \;
\KwOut{$U$}
\algrule
\end{algorithm}









\subsection{A Differentially Private Algorithm}

The global sensitivity of \MW is $n$, but the local sensitivity is lower.  We prove the following in Appendix \ref{sec:mw_sens}:

\begin{theorem} [Sensitivity of \MW] \label{thm:mwsensi} 
The local sensitivity is given by $LS_{\MW}(\data) = \max \{ n_1, n_2 \}$, where $n_1$ and $n_2$ are the sizes of the two groups in $\data$.
\end{theorem}

We can now define our private test statistic, $\privMW$.  This algorithm first uses a portion of its privacy budget ($\epsilon_m$) to estimate the size of the smallest group.  This value is then reduced to $m^*$, such that with probability $1-\delta$ we have $n-m^* >  LS_{\MW}(\data)$.  This means that we can then release $U$ using noise proportional to $n-m^*$ (using the remaining privacy budget, $\epsilon_U$.  See Appendix \ref{sec:mw_pf} for proof that \privMW is $(\epsilon_m + \epsilon_U, \delta)$-differentially private.

\begin{algorithm} [!htb]
\DontPrintSemicolon
\algrule
\textbf{Algorithm} $\privMW$ \textbf{:} Private Mann-Whitney Test Statistic\;
\algrule
\KwIn{\data, $\epsilon_{m}$, $\epsilon_{U}$, $\delta$}
    $m \longleftarrow  \min \{n_1,n_2\}$\;
    $\widetilde{m} \longleftarrow m + \lap \Big( \frac{1}{\epsilon_{m}} \Big) $\;
    $c \longleftarrow -\frac{\ln (2\delta)}{\epsilon_{m}}$\;
    $m^* \longleftarrow \max(\left \lceil{\widetilde{m}-c} \right \rceil, 0)$\;
    $\widetilde{U} \longleftarrow \MW (\data) + \lap \Big( \frac{n-m^*}{\epsilon_{U}} \Big)$\; 
\KwOut{$\widetilde{m}$, $\widetilde{U}$}
\algrule
\end{algorithm} 


As before, $\privMW$ is not meaningful on its own; we want an applicable reference distribution with which to calculate a corresponding $p$-value. This is shown below in algorithm \privMWP.  It works similarly to the analogous algorithms \privKWP and \privKWPA.  The key difference is that the reference distribution now depends on the group size estimate $\widetilde{m}$.\footnote{The algorithm given simulates full databases to compute the reference distribution.  This is not particularly slow, but in Appendix \ref{sec:mw_t1} we show that one can also sample from a normal distribution with certain parameters to get an acceptable reference distribution more quickly.}

\begin{algorithm} [!htb]
\DontPrintSemicolon
\algrule
\textbf{Algorithm} $\privMWP$ \textbf{:} Complete Mann-Whitney Test\;
\algrule
\KwIn{\data, $\epsilon_{m}$,$\epsilon_{U}$,$\delta$, $z$}
	$(\widetilde{m}, \ \widetilde{U}) \longleftarrow \privMW(\data, \epsilon_{m},\epsilon_{U},\delta)$\\
	$\widetilde{m} \longleftarrow \lceil \max(0, \widetilde{m})\rceil$ \\
    \For{$k := 1$ \textup{to} $z$}{
    	$\data^* \longleftarrow $ a database with $n$ independent uniform values from [0,1] divided into $2$ groups of size $\widetilde{m}$ and $n-\widetilde{m}$  \;
    	$U_k \longleftarrow  \privMW(\data^*, \epsilon_{m},\epsilon_{U},\delta)$
    }
    $p \longleftarrow$ fraction of $U_k$  less than $\widetilde{U}$ \;
\KwOut{$\widetilde{U}, p$}
\algrule
\end{algorithm}

\paragraph{A note on design}  In the case of \privKWPA we found that the highest possible critical value came from a reference distribution with equal-size groups.  For this test that is not the case, so we cannot use equal-size groups when generating the reference distribution without unacceptable type 1 error.  As a result, we need an estimate of group size. If we didn't need this estimate for the reference distribution, it is possible that \privMW would be more accurate by simply using the global sensitivity bound on \MW.  (It would be a slightly higher sensitivity, but no privacy budget would need to be expended on estimating $m$.)  This is a good example of a point made in Section \ref{sec:intro}: simply acheiving an accurate of estimate of a test statistic is not enough. The ultimate goal of a hypothesis test is a p-value, which also requires an accurate reference distribution and high power in order to minimize decision error.

\paragraph{Type 1 error}  The reference distribution in the \privMWP algorithm depends on $m$, which is only estimated by $\widetilde{m}$, so we need to experimentally verify that type 1 error never exceeds $\alpha$. See Appendix \ref{sec:mw_app} for evidence that our estimate appears to be sufficiently accurate and for additional discussion.

%

\begin{theorem}
Algorithm $\privMWP$ is $(\epsilon_{m} + \epsilon_{U}, \delta)$-differentially private.
\end{theorem}

\begin{proof}
By Theorem \ref{thm:dpprrof}, the computation of $(\widetilde{m}, \ \widetilde{U})$ is $(\epsilon_{m} + \epsilon_{U}, \delta)$-differentially private. All of the steps following this computation do not require access to the database and are, thus, post processing. By Theorem \ref{thm:pp}, it follows that the complete algorithm is also $(\epsilon_{m} + \epsilon_{U}, \delta)$-differentially private.
\end{proof}

\subsection{Experimental Results}\label{sec:pwr}

\paragraph{Power analysis} We first assessed the power of our test on synthetic data.\footnote{For application of our test to real-world data, see Appendix \ref{sec:mw_app}} We run \privMWP on many simulated databases and report the percentage of the time that a significant result was obtained.  For our first effect size, we have the two groups consist of normally distributed data with means one standard deviation apart.  In all experiments we set $\delta = 10^{-6}$.


Our first step was to determine the optimal proportion of the total privacy budget, $\epsilon_{tot}$, to allot to estimating $m$ and the test statistic $\privMW$. We found that the optimal proportion of $\epsilon$ to allot to estimating $m$ is roughly $.65$, experimentally confirmed at several choices of $\epsilon_{tot}$, effect size, total sample size $n$, group size ratios $n_1/n$, and underlying distribution. (See Appendix \ref{sec:mw_app} for more details.) We then fix the proportion of $\epsilon_{tot}$ allotted to $\epsilon_{m}$ as $.65$ and vary $\epsilon_{tot}$ and total sample size $n$ to measure the power of our test.

\begin{figure} [!htb] 
    \centering
    \includegraphics[width=\linewidth, keepaspectratio]{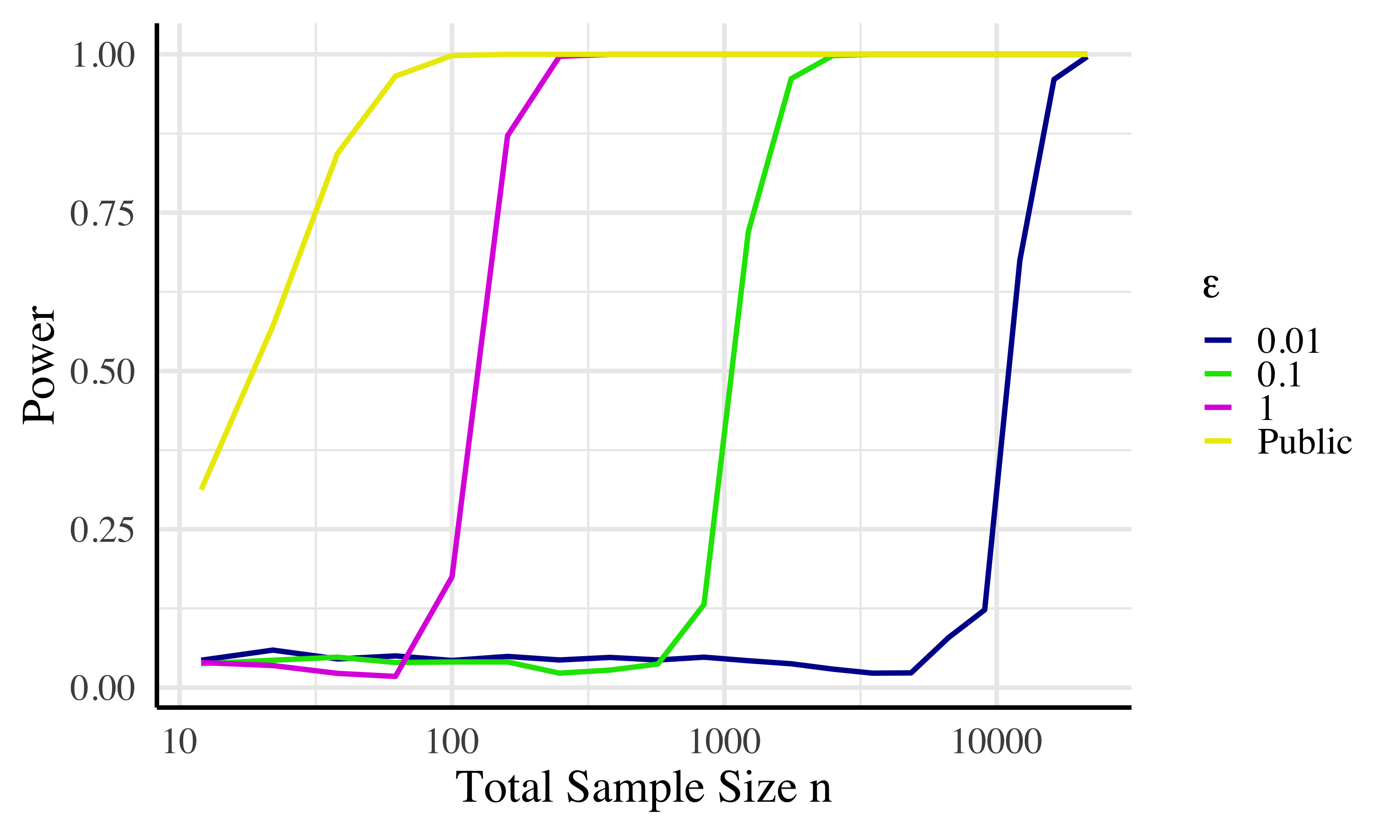}
    \caption{Power of $\privMWP$ at various values of $\epsilon_{tot}$ and total sample size $n$. (Effect size: $\mu_{1} - \mu_{2} = 1 \sigma$; proportion of $\epsilon_{tot}$ to $\epsilon_{m} = .65$; $\alpha=.05$; $m$:($n-m$) $= 1$)}\label{fig:var_N_tot_eps}
\end{figure}

As shown in Figure \ref{fig:var_N_tot_eps}, the power loss due to privacy is not unreasonably large. At an $\epsilon_{tot}$ of $1$, the test only requires a database that is approximately a factor of $3$ larger than that needed for the public test to reach a power of $1$. As one might expect, the database size needed to detect a given effect has a roughly inverse relationship with $\epsilon_{tot}$. In the appendix we perform a similar power analysis, varying effect size rather than sample size.

\paragraph{Uniformity of p-values} Algorithm $\privMWP$ uses the privatized group sample sizes $m^*, (n - m^*)$ in place of the true group sizes $n_1$, $n_2$ in order to simulate the reference distribution. Naturally, then, one may wonder \textit{how} conservative our critical values are as a result of ensuring that the type 1 error rate does not exceed $\alpha$. As shown in Figure \ref{fig:qqplot_mw}, the type I error rate of our test does not exceed $\alpha$ when group sample sizes are equal. As total sample size $n$ increases, the p-value quantiles asymptotically approach that of the theoretical distribution. In Appendix \ref{sec:mw_app}, we also examine uniformity of p-values of \privMWP with unequal group sizes and a variation of \privMWP that assumes equal group sizes.

\begin{figure} [!htb]
    \centering
    \includegraphics[width=\linewidth, keepaspectratio]{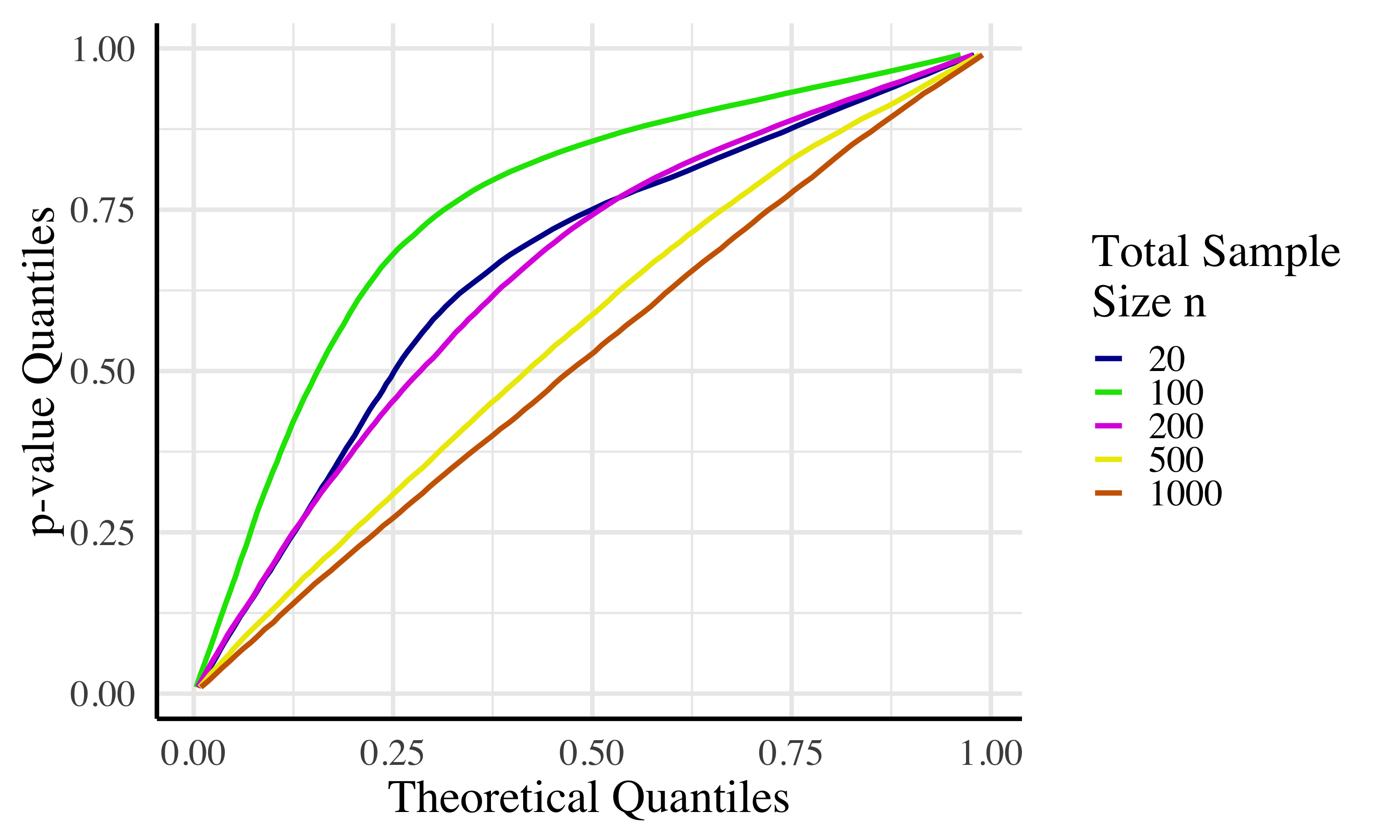}
    \caption{A quantile-quantile plot of \privMWP varying $n$. ($\epsilon_{tot} = 1$; proportion of $\epsilon_{tot}$ to $\epsilon_{m} = .65$; $m$:($n-m$)$ = 1$; normally distributed sample data)}\label{fig:qqplot_mw}
\end{figure}

\paragraph{Comparison to previous work} The best existing test applicable in the same use case is that of Swanberg et al.~\cite{new_anova}.
Their differentially private ANOVA test can be used in the 2-group case to compare to our Mann-Whitney test. The results of this comparison, using the same paramater settings chosen for optimal power in their test, can be seen in Figure \ref{fig:mw_vs_anova}, where our test offers a substantial power increase.

\begin{figure} [!htb] 
    \centering
    \includegraphics[width=\linewidth, keepaspectratio]{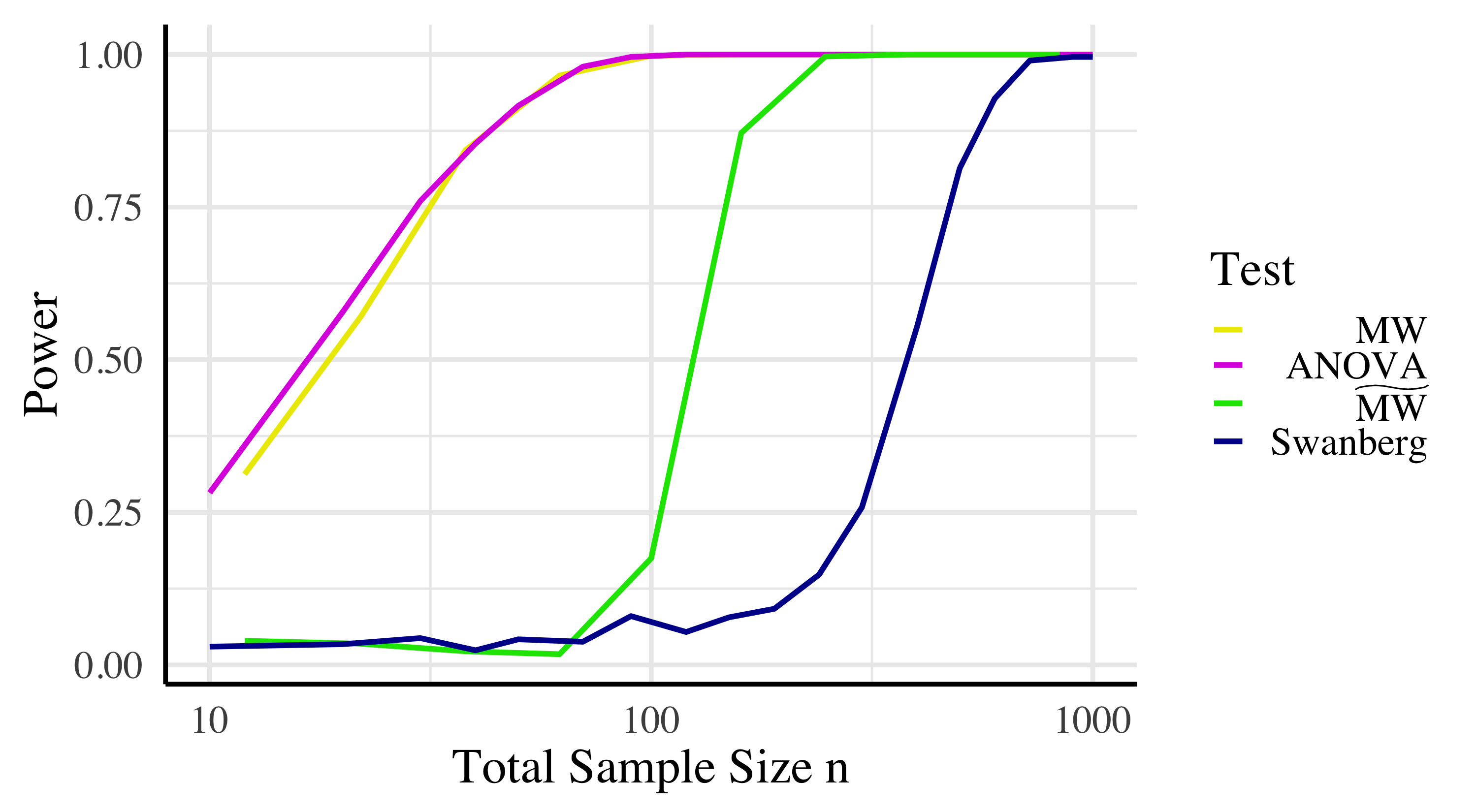}
    \caption{Power of $\privMWP$ and Swanberg et. al.'s test at various $n$. ($\epsilon_{tot} = 1$; Effect size: $\mu_{1} - \mu_{2} = 1 \sigma$; proportion of $\epsilon_{tot}$ to $\epsilon_{m} = .65$; $\alpha=.05$; $m$:($n-m$) $= 1$), normally distributed sample data}\label{fig:mw_vs_anova}
\end{figure}

\paragraph{Comparing \privMWP and \privKWPA} Both the Mann-Whitney and the Kruskal-Wallis can be used to compare the distributions of samples from two groups. As shown in Figure \ref{fig:kw_vs_mw}, we find that in the private setting, \privKWPA is more statistically powerful than \privMWP. This is perhaps surprising, since one might expect the test developed specifically for the two-group case to perform better.  But this is an example of how some tests privatize more easily than others.  \privMWP requires knowledge of the group sizes, using up a fraction of the privacy budget, while the \privKWPA statistic is not dependent on group size.

\begin{figure} [!htb] 
    \centering
    \includegraphics[width=\linewidth, keepaspectratio]{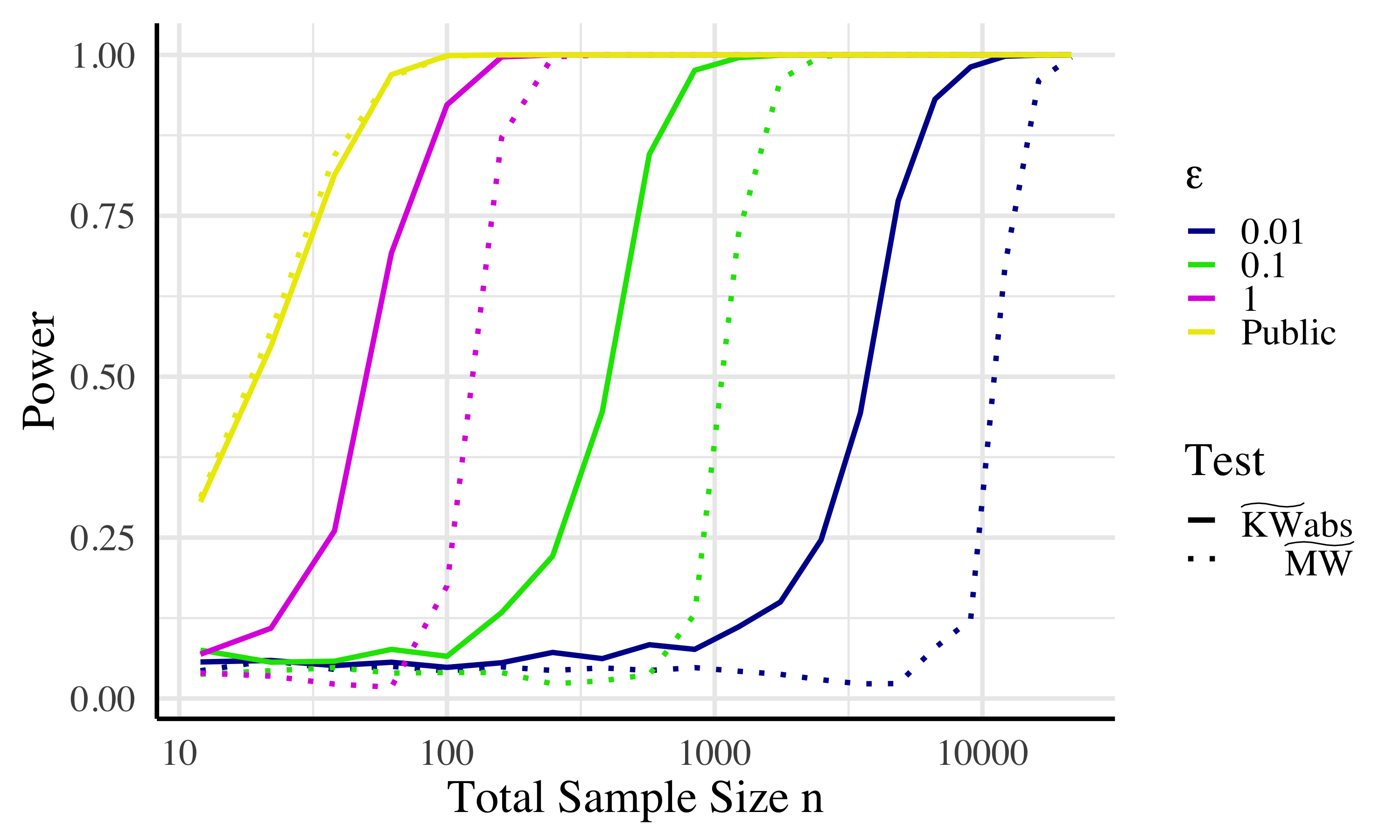}
    \caption{Power of $\privMWP$ and $\privKWPA$ at various $n$ and $\epsilon$. (Effect size: $max(\mu_{n}) - min(\mu_{n}) = 1 \sigma$; $g = 2$; $\alpha=.05$; equal group sizes; normally distributed sample data)}
    \label{fig:kw_vs_mw}
\end{figure}

We did find one exception to this finding.  If the analyst knows \textit{a priori} that the two groups are of equal size (e.g., the data collection method guaranteed an equal number in each group) then \privMWP can be run using an exact value of $n/2$ for the local sensitivity without the need to dedicate any privacy budget to estimating $m$.  This increases the accuracy of \privMW both by reducing the sensitivity used to add noise and by increasing the privacy budget allocated to $U$.  We find that in this case \privMWP \textit{is} superior to \privKWPA.  See Appendix \ref{sec:mw_app} for more details.

\section{Paired Data}\label{sec:wc}

We now consider a third situation, where there are two groups and the observations in those groups are paired.  While this scenario did not exist in the original psychotic disease study, one can imagine recording the methylation levels of one of the groups before and after administering medication.  Each subject then contributes a pair of data $(u_i, v_i)$ that are highly correlated with one another. One can assess the impact of the medication by considering whether the set of $n$ differences, $\{v_i-u_i\}_i$, is plausibly centered at zero. The standard nonparametric hypothesis test for this situation is the Wilcoxon signed-rank test, proposed in 1945 by Frank Wilcoxon \cite{WilcoxonDefinition}.  The parametric alternative is a simple one-sample t-test run on the set of differences.

This is the one setting where we are aware of prior work on a nonparametric test. Task and Clifton \cite{DPWilcoxon} gave the first private analogue of the Wilcoxon signed-rank test, referred to from here on as the \textit{TC test}, in 2016. Our test makes two key improvements to theirs and exhibits higher power.  We also correct some errors in their work.  We discuss the differences in more detail in Section \ref{sec:wc_results}.

Despite its status as one of the most commonly used hypothesis tests, to our knowledge there is no practical, implementable private version of a one-sample t-test in the literature.  In Section \ref{sec:ttest} we discuss some work that comes close, and then we give our own first attempt at a private t-test.  We again find that our nonparametric test has significantly higher power than this parametric alternative.

\subsection{The Wilcoxon signed-rank test}\label{wc}


The function calculating the Wilcoxon test statistic is formalized in Algorithm $\mathcal{W}$.  Given a database $\data$ containing sets of pairs $(u_i, v_i)$, the test computes the difference $d_i$ of each pair, drops any with $d_i=0$, and then ranks them by magnitude.  (If magnitudes are equal for several differences, all are given a rank equal to the average rank for that set.)

\begin{algorithm}
\algrule
\DontPrintSemicolon
\textbf{Algorithm } \W \textbf{:} Wilcoxon Test Statistic \;
\algrule
\KwIn{$\data$}
    \For{row $i$ of $\data$}{
       $d_i \longleftarrow |v_i - u_i|$\;
        $s_i \longleftarrow \textup{Sign}(v_i - u_i)$\;
    }
    Order the terms from lowest to highest $d_i$\;
    Drop any $d_i = 0$\;
    \For{row $i$ of $\data$}{
        $r_i \longleftarrow \textup{rank of row }i$
   }
   $w \longleftarrow \sum_i s_i r_i$\;
\KwOut{$w$}
\algrule
\label{alg:wc}
\end{algorithm}






Under the null hypothesis that $u_i$ and $v_i$ are drawn from the same distribution, the distribution of the test statistic $W$ can be calculated exactly using combinatorial techniques. This becomes computationally infeasible for large databases, but an approximation exists in the form of the normal distribution with mean 0 and variance $\frac{n_r(n_r+1)(2n_r+1)}{6}$, where $n_r$ is the number of rows that were not dropped.  Knowing this, one can calculate the p-value for any particular value of $w$.

\subsection{Our Differentially Private Algorithm}

At a high level, our algorithm is quite straightforward and similar to prior work: we compute a test statistic as one might in the public case and add Laplacian noise to make it private.  However, there are several important innovations relative to Task and Clifton that greatly increase the power of our test.

Our first innovation is to use a different variant of the Wilcoxon test statistic.  While the version introduced in Section \ref{wc} is the one most commonly used, other versions have long existed in the statistics literature.  In particular, we look at a variant introduced by Pratt in 1959 \cite{Pratt}.  In this variant, rather than dropping rows with $d_i = 0$, those rows are included.  When $d_i=0$ we set $s_i = \textup{Sign}(d_i) = 0$, so those rows contribute nothing to the resultant statistic, but they do push up the rank of other rows. 

\begin{algorithm}
\DontPrintSemicolon
\algrule
\textbf{Algorithm }$\WP$\textbf{:} Wilcoxon Test Statistic - Pratt Variant \;
\algrule
\KwIn{$\data$}
\For{row $i$ of $\data$}{
        $d_i \longleftarrow |v_i - u_i|$\;
        $s_i \longleftarrow \textup{Sign}(v_i - u_i)$ \;
    }
    Order the terms from lowest to highest $d_i$\;
    \For{row $i$ of $\data$}{
        $r_i \longleftarrow \textup{rank of row }i$
   }
   $w \longleftarrow \sum_i s_i r_i$\;
\KwOut{$w$}
\algrule
\end{algorithm}



In the public setting, the Pratt variant is not very different from the standard Wilcoxon, being slightly more or less powerful depending on the exact effect one is trying to detect \cite{WCcomp}.  In the private setting, however, the difference is substantial.

The benefit to the Pratt variant comes from how the test statistics are interpreted.  In the standard Wilcoxon, it is known that the test statistic follows an approximately normal distribution, but the variance of that distribution is a function of $n_r$, the number of non-zero $d_i$ values.  In the private setting, this number is not known, and this has caused substantial difficulty in prior work.  (See Section \ref{sec:wc_results} for more discussion.)  On the other hand, the Pratt variant produces a test statistic that is always compared to the same normal distribution, which depends only on $n$.  The algorithm \privWP that outputs a differentially private analogue is shown below. 


\begin{algorithm}[!htb]
\DontPrintSemicolon
\algrule
\textbf{Algorithm } $\privWP$ \textbf{:} Private Wilcoxon Test Statistic\;
\algrule
\KwIn{$\data$, $\epsilon$}
    $w \longleftarrow \WP(\data)$\;
    $\widetilde{w} \longleftarrow w + \lap \Big (\frac{2n}{\epsilon} \Big )$\;
\KwOut{$\widetilde{w}$}
\algrule
\label{alg:wcalg}
\end{algorithm}

\begin{theorem} \label{thm:wcalg}
Algorithm \privWP  is $\epsilon$-differentially private. 
\end{theorem}

See Appendix \ref{sec:wc_pf} for proof of Theorem \ref{thm:wcalg}.

To complete the design of our test, we compute a reference distribution through simulation as was done in \privKWPA and \privMWP.  Here we use the standard normal approximation for the distribution of the $w$ test statistic, though one could simulate full databases as well.  We call this algorithm \privWPP.


\begin{algorithm} [!htb]\label{alg:complete_wc}
\DontPrintSemicolon
\algrule
\textbf{Algorithm } $\privWPP$ \textbf{:} Complete Wilcoxon Test\;
\algrule
\KwIn{\data, $\epsilon$, $z$}
	$\widetilde{w} \longleftarrow \privWP(\data, \epsilon)$\;
    \For{$k = 1$ \textup{to} $z$}{
                $w_k \longleftarrow {\sf Normal}(0, n(n+1)(2n+1)/6) + \lap(2n/\epsilon)$;
    }
    $p \longleftarrow$ fraction of $w_k$ more extreme than $\widetilde{w}$ \;
\KwOut{$\widetilde{w}, p$}
\algrule
\end{algorithm}

\begin{theorem}
Algorithm \privWPP is $\epsilon$-differentially private. 
\end{theorem}

\begin{proof}
The computation of $\widetilde{w}$ was already shown to be private.  The remaining computation needed to find the p-value does not need access to the database---it is simply post-processing.  By Theorem \ref{thm:pp}, it follows that the \privWPP algorithm is also private.
\end{proof}

\subsection{Experimental Results}\label{sec:wc_results}

\paragraph{Power analysis} We assess the power of our differentially-private Wilcoxon signed-rank test first on synthetic data.  (For tests with real data, see Appendix \ref{sec:realdata}.) In order to measure power, we must first fix an effect size.  We chose to have the $u_i$ and $v_i$ values both generated according to normal distributions with means one standard deviation apart. We then measure the statistical power of Algorithm \privWPP (for a given choice of $n$ and $\epsilon$) by repeatedly randomly sampling a database \data from that distribution and then running \privWPP on that database.\footnote{Our actual implementation differs slightly from this.  To save time when running a huge number of tests with identical $n$ and $\epsilon$, we first generate the reference distribution $W_k$ values, which can be reused across runs.} The power is the percentage of the time \privWPP returns a p-value less than $\alpha$. See Appendix \ref{sec:wc_appendix} for a similar analysis of power, varying effect size rather than sample size.

\begin{figure} [!htb] 
    \centering
    \includegraphics[width=\linewidth]{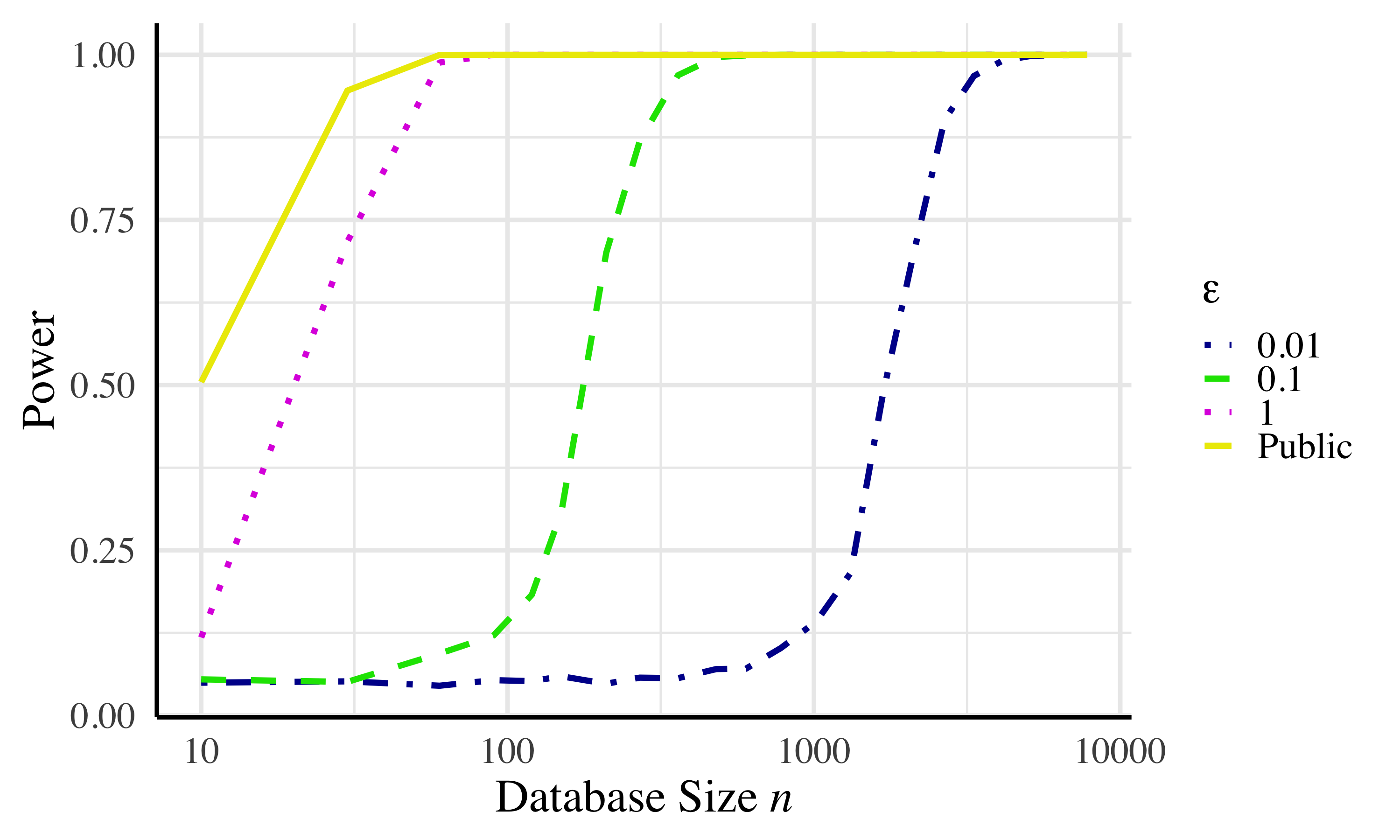}
    \caption{Power of \privWPP at various $\epsilon$ and $n$. (Effect size: $\mu_u - \mu_v = 1 \sigma$; $\alpha = .05$; normally distributed sample data)}
    \label{fig:noncomp_without_ties}
\end{figure}


\paragraph{Uniformity of p-values} 
In algorithm \privWPP we draw our reference distribution samples (the $w_k$ values) assuming there are no $d_i=0$ rows.  The distribution will technically differ slightly when there are many rows with $d_i=0$, so we need to confirm experimentally that the difference is inconsequential or otherwise acceptable.

Figure \ref{fig:qqplot_wc} shows a Q-Q plot of \privWPP on three sets of p-values, all generated under $H_0$, with $\epsilon = 1$, $n = 500$. When there are no ties in the original data (0\% of $d_i = 0$), the Q-Q plot line is indistinguishable from the identity line, indicating that the test is properly calibrated. Encouragingly, introducing a substantial number of ties into the data (30\% of $d_i = 0$) has little noticeable effect. In order to induce non-uniformity in the p-values, one needs an extremely high proportion of rows with $d_i=0$.  The curve with 90\% zero values is shown as an illustration.  When the proportion of zeros is very high, the variance of the p-values will be narrower than the reference distribution, resulting in a lower critical value.  Since the value we are using is higher, our test is overly conservative,\footnote{One could try to estimate the number of zeros to be less conservative, but that would require allocating some of the privacy budget towards that estimate, which is not worth it in most circumstances.} but this is acceptable as type I error is still below $\alpha$.

\begin{figure} [!htb]
    \centering
    \includegraphics[width=\linewidth]{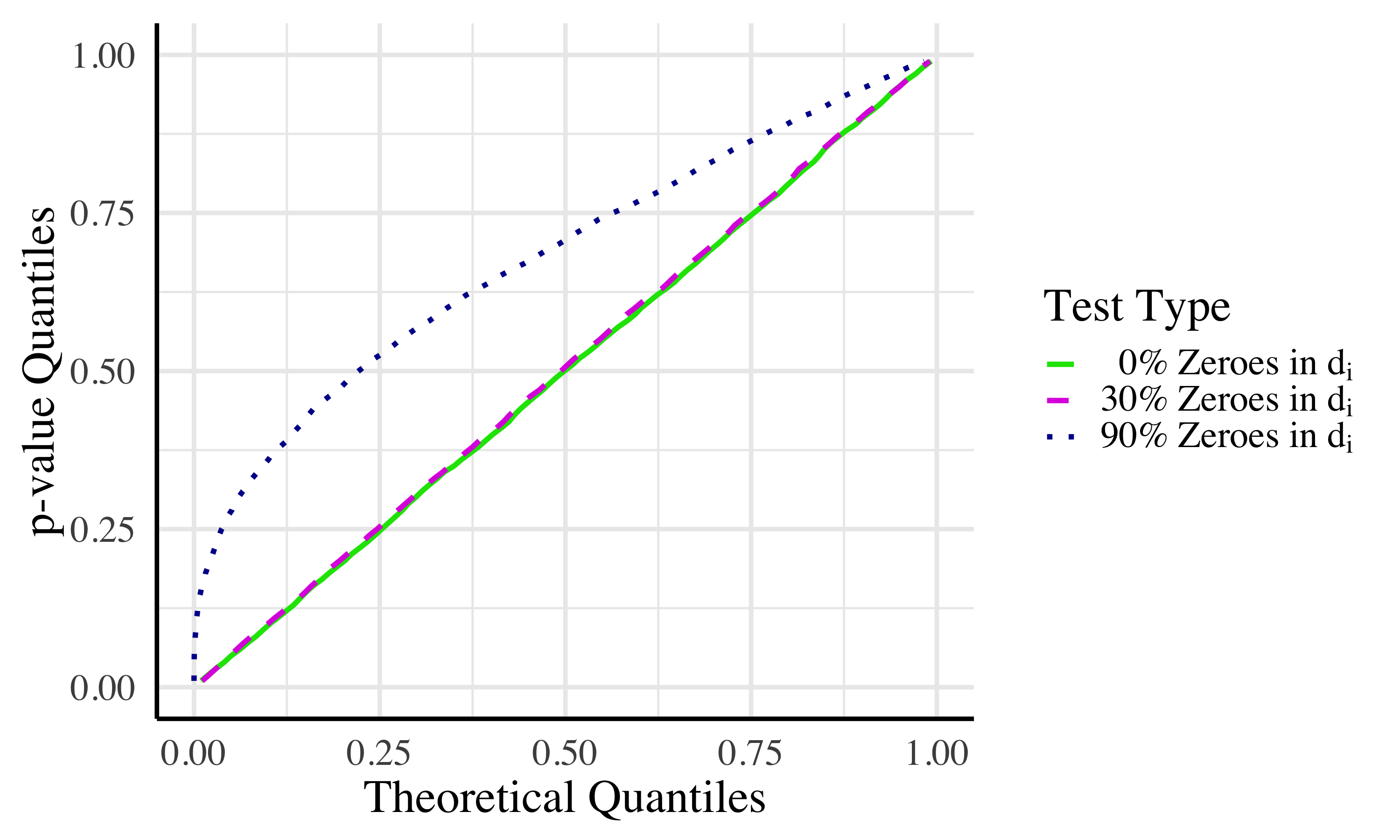}
    \caption{A quantile-quantile plot of \privWPP comparing the distribution of simulated p-values to the uniform distribution ($\epsilon = 1$, $n = 500$; normally distributed sample data).}
    \label{fig:qqplot_wc}
\end{figure}

\paragraph{Comparison to previous work} In 2016, Task and Clifton \cite{DPWilcoxon} introduced the first differentially private version of the Wilcoxon signed-rank test, from here on referred to as the TC test. Our work improves upon their test in two ways. We describe the two key differences below, and then compare the power of our test to theirs.  We also found a significant error in their work.\footnote{This error has been confirmed by Task and Clifton in personal correspondence.}  All comparisons are made to our implementation of the TC test with the relevant error corrected.

Task and Clifton compute an analytic upper bound on the critical value $t^*$.  For a given $n$ and $\epsilon$, the private test statistic $\widetilde{w}$ under $H_0$ is sampled according to a sum $W + \Lambda$, where $W$ is a random draw from a normal distribution (scaled according to $n$) and $\Lambda$ is a Laplace random variable (scaled according to $n$ and $\epsilon$).  In particular, say that $b$ is a value such that $\Pr[W>b] < \beta$ and $g$ is a value such that $\Pr[\Lambda > g] < \gamma$.  Then we can compute the following bound.  (The last line follows from the fact that the two events are independent.)
\begin{align*}
\Pr[\widetilde{W} > b + g] &< \Pr[W > b \text{ or } \Lambda > g] \\
&= \Pr[W > b] + \Pr[\Lambda > g]\\ 
&\ \ \ - \Pr[W > b \text{ and } \Lambda > g] \\
&= \beta + \gamma - \beta\gamma
\end{align*}

Task and Clifton always set $\gamma = .01$ and then vary  the choice of $\beta$ such that they have $\alpha = \beta + \gamma - \beta\gamma$ for whatever $\alpha$ is intended as the significance threshold.\footnote{This is where Task and Clifton make an error.  This formula is correct, but they used an incorrect density function for the Laplace distribution and as a result calculated incorrect values of $g$.}

The bound described above is correct but very loose, and our simulation method gives drastically lower critical values.
Table \ref{tab:n_100_comp} contains examples of the critical values achieved by each method for several parameter choices. More values can be found in Appendix \ref{sec:wc_appendix}, where we also experimentally confirm that these values result in acceptable type 1 error. 

\begin{table}[ht]
\centering
\caption{Critical Value Comparison for $n=100$}\label{tab:n_100_comp}
\begin{tabular}{llrrr}
  \hline
$\epsilon$ & $\alpha$ & Public & New & TC \\ 
  \hline
1 & 0.1 & 1.282 & 1.417 & 2.680 \\ 
   & 0.05 & 1.645 & 1.826 & 3.091 \\ 
   & 0.025 & 1.960 & 2.186 & 3.511 \\ 
   \hline
0.1 & 0.1 & 1.282 & 5.684 & 14.786 \\ 
   & 0.05 & 1.645 & 8.063 & 15.197 \\ 
   & 0.025 & 1.960 & 10.438 & 15.617 \\ 
   \hline
0.01 & 0.1 & 1.282 & 55.350 & 135.843 \\ 
   & 0.05 & 1.645 & 79.233 & 136.254 \\ 
   & 0.025 & 1.960 & 103.116 & 136.674 \\ 
   \hline
\end{tabular}
\\[10pt]
\caption*{Critical values for \textit{n} = 100 and several values of $\epsilon$ and $\alpha$.  To allow easy comparison, these values are for a normalized $W$ statistic, i.e., $W$ has been divided by the relevant constant so that it is (before the addition of Laplacian noise) distributed according to a standard normal. See Appendix \ref{sec:wc_appendix} for the equivalent table at n = 1000.}
\end{table}

Our second key change from the TC test, mentioned earlier, is that we handle rows with $d_i=0$ according to the Pratt variant of the Wilcoxon, rather than dropping them completely as is more traditional.  The reason the traditional method is so difficult in the private setting is that the reference distribution one must compare to depends on the number of rows that were dropped.  If $n_r$ is the number of non-zero rows (i.e., rows that weren't dropped), one is supposed to look up the critical value associated with $n_r$, rather than the original size $n$ of the database.  

Unfortunately, $n_r$ is a sensitive value and cannot be released privately.\footnote{A private estimate could be released, but one would have to devote a significant portion of the privacy budget for the hypothesis test to this estimate, greatly decreasing the accuracy/power of $\privW$.}  Task and Clifton show that it is \textit{acceptable} (in that it does not result in type I error greater than $\alpha$) to compare to a critical value for a value of $n_r$ that is lower than the actual value.  This allows them to give two options for how one might deal with the lack of knowledge about $n_r$.

\begin{description}
\item[High Utility] This version of the TC test simply assumes $n_r \geq .3n$ and uses the critical value that would be correct for $n_r = .3n$.  We stress that this algorithm is \textit{not} actually differentially private, though it could easily be captured by a sufficiently weakened definition that limited the universe of allowable databases.  Another problem is that for most realistic data, $n_r$ is much greater than $.3n$ and using this loose lower bound still results in a large loss of power.
\item[High Privacy] This version adds $k$ dummy values to the database with $d_i = \infty$ and $k$ with $d_i = -\infty$.\footnote{Task and Clifton do not discuss how to choose $k$, and in our experimental comparisons we set $k=15$, the same value they use.}  Then one can be certain of the bound $n_r \geq 2k$.  This is a guaranteed bound so this variant truly satisfies differential privacy.  On the other hand it is a very loose lower bound in most cases, leading to a large loss of power. 
\end{description}

\paragraph{Experimental comparison}  We compare the statistical power of our test to that of the TC test.  We begin by again measuring the power when detecting the difference between two normal distributions with means one standard deviation apart.  The results can be seen in Figure \ref{fig:comp_no_ties}.  If we look at the database size needed to achieve 80\% power, we find that the 32 data points we need, while more than the public test (14), are many fewer than the TC High Utility variant (80) or the TC High Privacy variant (122).  Appendix \ref{sec:wc_appendix} includes a figure for $\epsilon=.1$ as well.  What we see is that, while all private tests require more data, our test (requiring $n\approx 236$) still requires about 40\%  as much data as the TC High Utility variant (588).  The TC High Privacy variant, however, scales much less well to low $\epsilon$ and requires roughly 2974 data points.

\begin{figure} [!htb] 
    \centering
    \includegraphics[width=\linewidth]{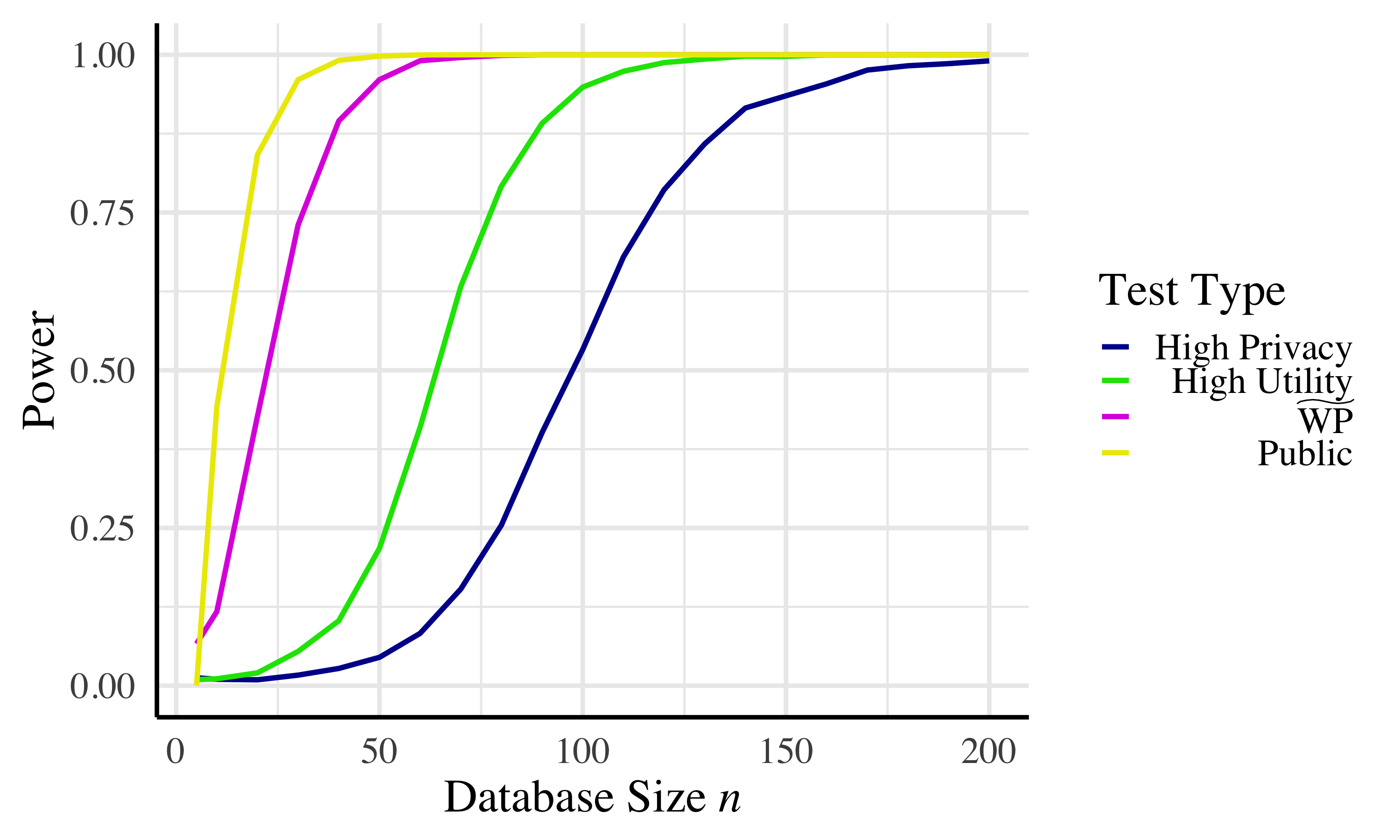}
    \caption{Power of the TC test, \privWPP, and the public test at various $n$. (Effect size: $\mu_u - \mu_v = 1 \sigma$, $\epsilon = 1$; $\alpha=.05$; normally distributed sample data)}
    \label{fig:comp_no_ties}
\end{figure}

The results in Figure \ref{fig:comp_no_ties} use a continuous distribution for the real data, so there are no data points with $d_i=0$.  Because one of the crucial differences between our algorithms is the method for handling these zero values, we also consider the effect when there are a large number of zeros in Appendix \ref{sec:wc_appendix}.  Overall, we see that both in situations with no zero values and situations with many, our test achieves the rigorous privacy guarantees of the TC High Privacy test while achieving greater utility than the TC High Utility test.

\paragraph{Relative contribution of improvements} Given that we make two meaningful changes to the TC test, one might naturally wonder whether both are truly useful or whether the vast majority of the improvement comes from one of the two changes.  To test this, we compare to an updated variant of the TC test where we calculate critical values exactly through simulation, as we do in our algorithm, but otherwise run the TC test unchanged (referred to as "High Privacy +" and "High Utility +"). The result is presented in Figure \ref{fig:comp_no_ties_w_plus}, where we find the resulting algorithm to rest comfortably between the original TC test and our proposed test.  This means that both the change to the critical value calculation and the switch to the Pratt method of handling $d_i=0$ rows are important contributions to achieving the power of our test.

\begin{figure} [!htb] 
    \centering
    \includegraphics[width=\linewidth]{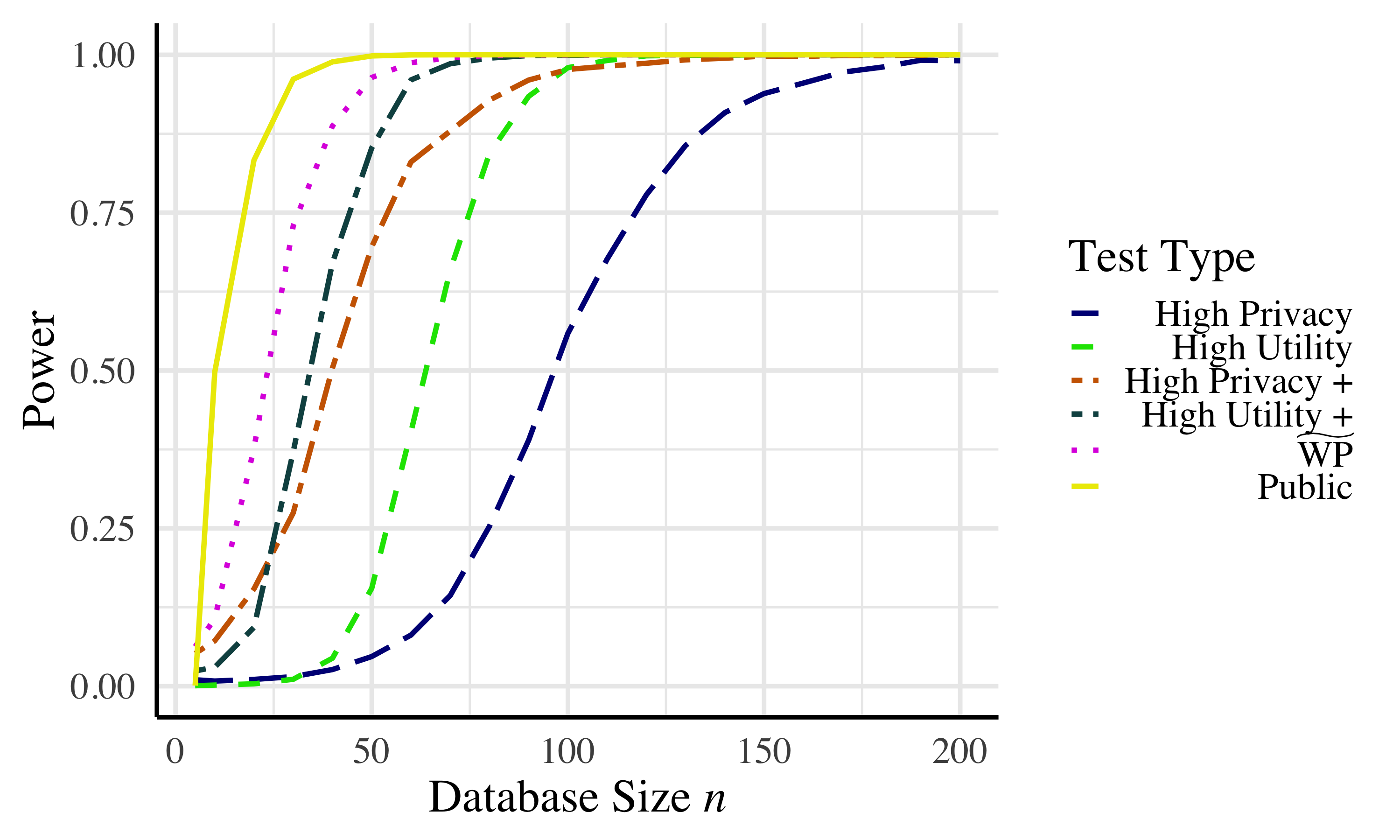}
    \caption{Power comparison of the TC algorithms, the TC algorithms with our critical values (denoted with a +), our new algorithm, and the public algorithm at various sample sizes $n$. (Effect size: $\mu_u - \mu_v = 1 \sigma$; $\epsilon=1$; $\alpha=.05$)}
    \label{fig:comp_no_ties_w_plus}
\end{figure}

\subsection{Parametric Alternative: A New T-test}\label{sec:ttest}

The parametric analog to the Wilcoxon test is to run a one sample t-test on the set of differences $\{v_i-u_i\}_i$ to see if their mean is significantly different from zero (also called a paired t-test).  There has been surprisingly little work on the creation of a private version of a one sample t-test.  Karwa and Vadhan \cite{karwa2017finite} study private confidence intervals, which are in a sense equivalent to a t-test.  However, their analysis is asymptotic and they say that the algorithm does not give practical results with database size in the thousands.  Sheffet \cite{sheffet2015differentially} provides a method for calculating private coefficient estimates for linear regression and also transforms the t-distribution to provide an appropriate reference distribution for inference.  In the public setting, one can convert a test on regression coefficients to a one sample t-test but choosing a constant independent variable and making the sample data the dependent variable.  However, Sheffet's method only works when all variables are significantly spread out, so this method fails.  

Here we propose what we believe is the first private version of a one sample t-test, with two arguable exceptions.  The first is simultaneous work by Gaboardi et al.~\cite{gaboardi2018locally} in the local privacy model.  We compare our results to theirs in more detail in Section \ref{sec:ttestexp}.  The other work is that of Solea \cite{solea2014differentially}, but according to Solea's own experiments that test often gives type 1 error rates well above the chosen $\alpha$ for many parameter choices, so we don't consider it a usable test.

The database for a one sample t-test has observations $x_1, \ldots, x_n$ assumed to come from a normal distribution with mean $\mu$ and standard deviation $\sigma$. (For paired data, each observation is the differences between the observation in the two groups). The test statistic is given by $\T(\data) =  \frac{\bar{x}}{s/\sqrt{n}}$, where $\bar{x}$ is the mean of the data and $s$ is the standard deviation of the data.

\paragraph{A private t-test}

As before, we achieve privacy through the addition of Laplacian noise, but the sensitivity of \T is unbounded, so we instead release separate private estimates of the numerator and denominator.  For this analysis, similar to the private ANOVA tests \cite{new_anova}, we assume that the data is scaled such that all observations are on the interval $[-1,1]$. We first find the sensitivities of $\bar{x}$ and $s^2$ and then use post-processing, composition, and the Laplace Mechanism to combine these to obtain the private t-statistic.  In the case where $s^2$ is estimated to be negative, the test statistic cannot be computed as normal, and we return 0, indicating an unwillingness to reject the null hypothesis.

\begin{theorem} \label{thm:mean_sensitivity}
The sensitivity of $\bar{x}$ is $\frac{2}{n}$ .
\end{theorem}

\begin{theorem} \label{thm:s_squ_sensitivity}
The sensitivity of $s^2$ is $\frac{5}{n-1}$.
\end{theorem}

See Appendix \ref{sec:t_test_pf} for proof of Theorem \ref{thm:mean_sensitivity} and \ref{thm:s_squ_sensitivity}.

\begin{algorithm}
\DontPrintSemicolon
\algrule
\textbf{Algorithm  }$\privT$\textbf{:}  Private t-Test Statistic\;
\algrule
\KwIn{$\data$, $\epsilon_{\bar{x}}$, $\epsilon_{s^2}$}
$\widetilde{\bar{x}} = \bar{x} +  \lap(\frac{1/n}{\epsilon_{\bar{x}}})$\\
$\widetilde{s^2} = s^2+  \lap(\frac{5/(n-1)}{\epsilon_{s^2}})$\\
\eIf{$\widetilde{s^2} < 0$}{$\widetilde{T}=0$}{
$\widetilde{T}= \frac{\widehat{\bar{x}}/n}{\sqrt{\widehat{s^2}}/\sqrt{n}}$  }
\KwOut{$\widetilde{T}$}
\algrule
\label{alg:wcp}
\end{algorithm}

\begin{theorem}
Algorithm $\privT$ is $(\epsilon_{\bar{x}}+\epsilon_{s^2})$-differentially private. 
\end{theorem}

\begin{proof}
By the Laplace mechanism, the computation of $\widetilde{\bar{x}}$ is $\epsilon_{\bar{x}}$-differentially private and the computation of $\widetilde{s^2}$ is $\epsilon_{s^2}$-differentially private.  Since the computation of $\widetilde{T}$ does not require access to the database, it is only post-processing and its release is $(\epsilon_{\bar{x}}+\epsilon_{s^2})$-differentially private.
\end{proof}

To carry out the full paired t-test, we estimate the reference distribution through simulation and release a private p-value.

\begin{algorithm} [!htb]\label{alg:complete_t}
\DontPrintSemicolon
\algrule
\textbf{Algorithm } $\privTP$ \textbf{:} Complete t-Test\;
\algrule
\KwIn{\data, $\epsilon_{\bar{x}}$, $\epsilon_{s^2}$, $z$}
	$\widetilde{t} := \privT(\data, \epsilon_{\bar{x}}, \epsilon_{s^2})$\;
    \For{$k = 1$ \textup{to} $z$}{
                $\data^* \longleftarrow $ a database with $n$ independent draws from $N(\mu = 0, \sigma \approx 0.3)$, each truncated to $[-1, 1]$ \\
                $t_k \longleftarrow \privT(\data^*)$
    }
    $p \longleftarrow$ fraction of $t_k$ more extreme than $\widetilde{t}$ \;
\KwOut{$\widetilde{t}, p$}
\algrule
\end{algorithm}

\begin{theorem}
Algorithm \privTP is $\epsilon_{\bar{x}} + \epsilon_{s^2}$-differentially private. 
\end{theorem}

\begin{proof}
The computation of $\widetilde{t}$ was already shown to be private.  The remaining computation needed to find the p-value does not need access to the database---it is simply post-processing.  By Theorem \ref{thm:pp}, it follows that the \privTP algorithm is also private.
\end{proof}

\subsection{Experimental t-Test evaluation}\label{sec:ttestexp}

We first must set a parameter in our \privTP algorithm.  In particular, for a given total $\epsilon$, we must decide how to allocate the budget between $\epsilon_{\bar{x}}$ and $\epsilon_{s^2}$.  We choose this allocation experimentally, deciding to allocate 50\% of the budget towards each value.  This is nontrivial, and Appendix \ref{sec:t_appendix} contains experimental results and further discussion.  Luckily, the exact choice of this allocation does not seem to have a large effect on the power of the test.

We then evaluate the power and validity of the final \privTP test.

\paragraph{Comparison to other work} Simultaneous to our work, Gaboardi et al.~\cite{gaboardi2018locally} developed a private one sample t-test under the more restrictive local differential privacy model.  As one might expect, our test in the more standard setting is much higher power.  They develop both a t-test and a z-test, which is equivalent to the t-test except that the variance of the data is assumed to be already known.  Only the z-test is given experimental evaluation, but with an effect size three times the size we use in our experiments, their test (at $\epsilon=1$) requires roughly 4000 data points to reach 80\% power, while our test requires roughly 100.  Their t-test would presumably require even more data.

\paragraph{Comparison to nonparametric test} Since we have already developed a test for the paired-data use case, we assessed the power of \privTP in comparison to \privWPP by simulating synthetic data as described in Section \ref{sec:wc_results}.  Just as in the many groups and two groups scenarios, the nonparametric test substantially outperforms its parametric counterpart, as shown in Figure \ref{fig:paired_test_comparison}. In this case, \privWPP needs 8\% of the data required by \privTP to reach the same power.

\begin{figure} [!htb] 
    \centering
    \includegraphics[width=\linewidth]{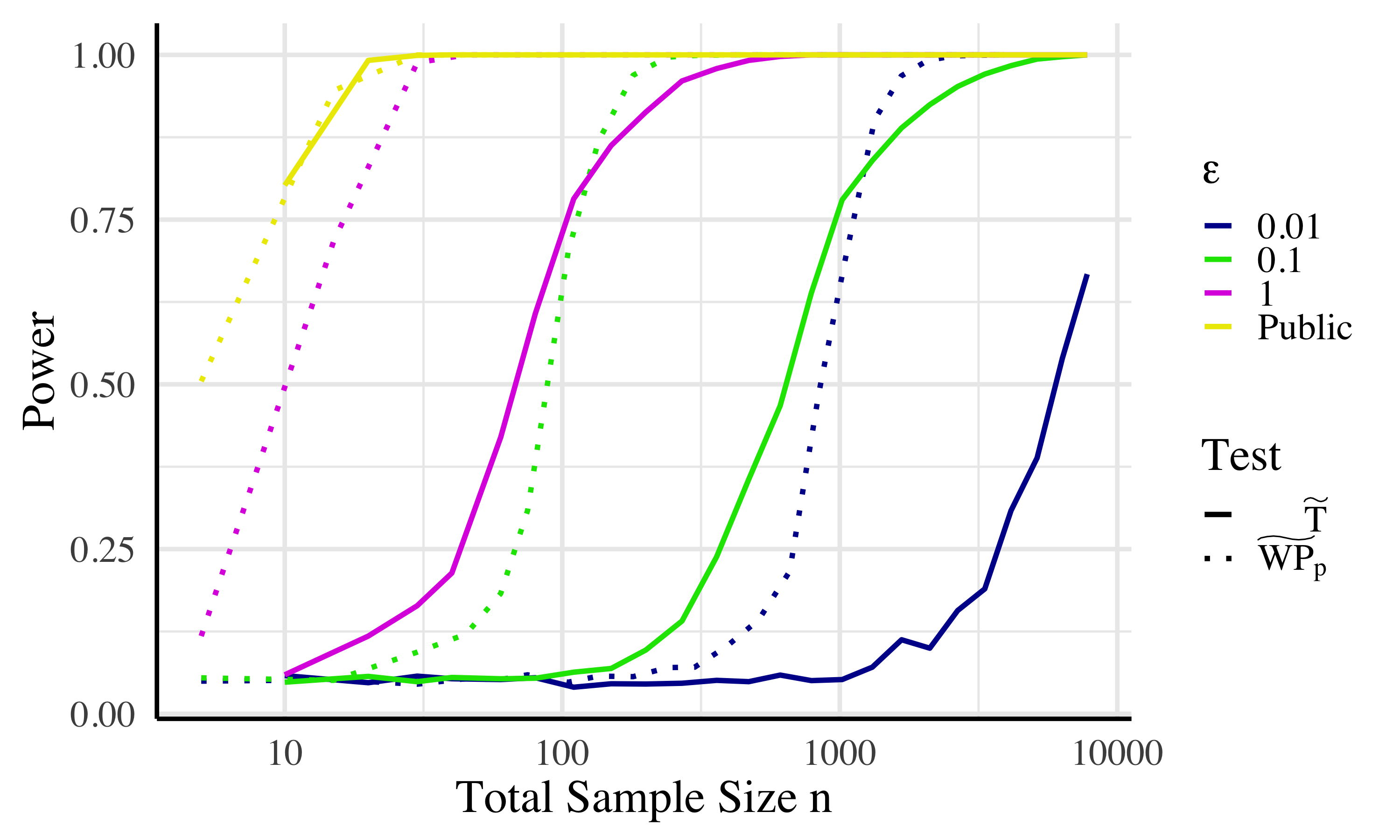}
    \caption{Power of \privTP and \privWPP at various $\epsilon$ and $n$. (Effect size: $\mu_u - \mu_v = 1 \sigma$; $\alpha = .05$; normally distributed sample data)}
    \label{fig:paired_test_comparison}
\end{figure}

\paragraph{Uniformity of p-values} As with all of our tests, we experimentally ensure that type I error rate is bounded by $\alpha$ in Figure \ref{fig:t_qqplot1}. This figure confirms the fact that our type I error rate is bounded above by $\alpha$.  For small sample sizes, the line on the quantile-quantile plot goes above the diagonal.  This is the acceptable direction, the sign of a conservative test.  In this case it occurs because some test statistics in the reference distribution are set to zero (as a result of noise added for privacy overwhelming $\widetilde{s^2}$).  If, for example, 10\% of the reference distribution samples are at zero, then p values below 10\% are impossible.  As shown by the $n = 1000$ line, at sufficiently large sample sizes this effect essentially vanishes.

\begin{figure} [!htb] 
    \centering
    \includegraphics[width=\linewidth, keepaspectratio]{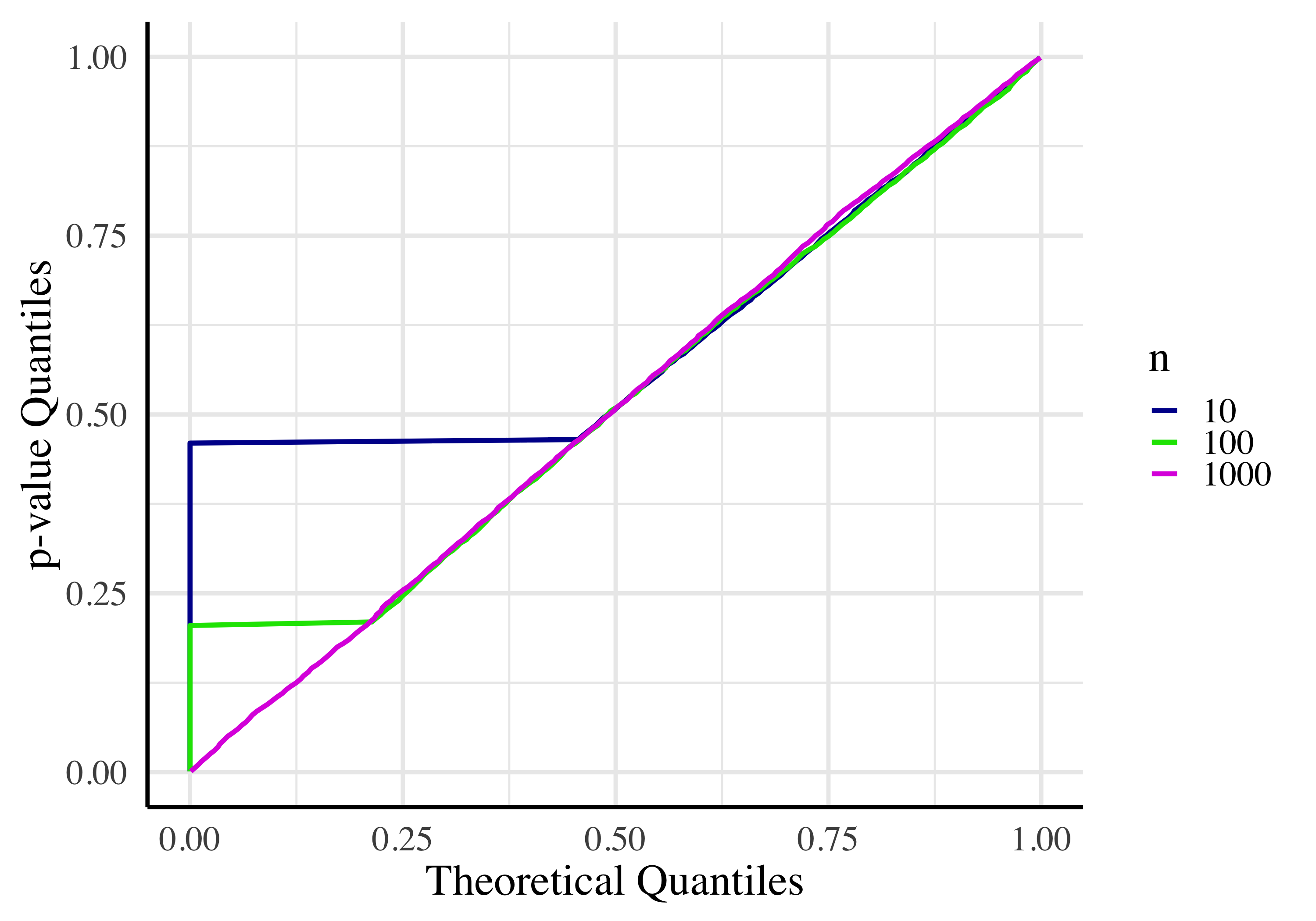}
    \caption{A quantile-quantile plot of \privTP at various $n$. ($\epsilon = 1$; equal $\epsilon$ allotment to each statistic)}\label{fig:t_qqplot1}
\end{figure}

\section{Conclusion}

We have introduced several new tests, of which three (\privKWPA, \privMWP, and \privWPP) are improvements on the state of the art. These allow researchers to address inferential questions using nonparametric methods while preserving the privacy of the data. More broadly, we found that the basic idea of using ranks in the private setting is potent. Not only do they remove the need to assume a bound on the data, they also directly increase statistical power. When working with many groups, two group, or with paired data, rank-based tests are more powerful than their parametric analogues and can be made yet more powerful through sensible adaptations.  We hope others will push this technique forward --- we have no reason to believe that our tests are optimal.


\subsection*{Acknowledgments}
We would like to thank Christine Task and Chris Clifton for generous and enlightening discussions regarding their previous work. This material is based upon work supported by the National Science Foundation under Grant No. SaTC-1817245 and the Richter Funds.

\clearpage

\bibliographystyle{plain}
\bibliography{sources}

\clearpage

\appendix

\section{Supplementary Proofs}\label{sec:pfs}

\subsection{Proof of Theorem \ref{thm:squaresensi}}\label{sec:kw_sq_sens}

We begin by introducing the following lemma and corollaries, which will be useful in the proof.

\begin{lemma}
    For integers $a$ and $b$: $$\sum_{k = a}^b k = \frac{(b-a+1)(b+a)}{2}$$ 
\end{lemma}

\begin{corollary}
    For group $i$ which has $n_i$ elements with rankings $r_{i1}, \ldots, r_{in_i}$, we can bound the sum of the ranks by:
    $$\frac{n_i (n_i + 1)}{2} \leq \sum_{j = 1}^{n_i} r_{ij} \leq \frac{n_i (2n - n_i + 1)}{2}$$
\end{corollary}

\begin{proof}
The maximum occurs when all the terms in the sum are as high as possible. I.e., they are $n - n_i + 1, n - n_i + 2, \ldots , n - 1, n$. Thus, 
\begin{align*}
    \sum_{j = 1}^{n_i} r_{ij} &\leq \sum_{k = n - n_i + 1}^n k \\
    &= \frac{(n - (n - n_i + 1) + 1)(n + (n - n_i + 1))}{2} \\ 
    &= \frac{n_i (2n - n_i + 1)}{2}
\end{align*}

The minimum occurs when all the terms in the sum are as low as possible. I.e., they are $1, 2, \ldots, n_i - 1, n_i$. Thus,
\begin{align*}
    \sum_{j = 1}^{n_i} r_{ij} &\geq \sum_{k = 1}^{n_i} k \\
    &= \frac{(n_i - 1 + 1)(n_i + 1)}{2} \\ 
    &= \frac{n_i (n_i + 1)}{2}
\end{align*}

\end{proof}

\begin{corollary}
    Suppose there are $g$ groups in the database, then $\forall \; 1 \leq i \leq g$, $$\frac{1+n_i}{2} \leq \bar{r}_i \leq \frac{2n-n_i+1}{2},$$ where $n_i$ is the group size.
\end{corollary}

\begin{proof}
    We define $\bar{r}_i$ to be $\frac{1}{n_i} \sum_{j=1}^{n_i} r_{ij}$. The proof follows directly from this definition.
\end{proof}

We now prove the main theorem:

\begin{theorem*}
The sensitivity of $\KW$ is bounded by 87.
\end{theorem*}

\begin{proof}
The proof proceeds as a series of cases. Neighboring databases $\data, \ \data'$ differ by a change in one row, and this difference will cause either the row's rank, group, or both to change. The first case bounds the change in the test statistic when only the row's rank changes. The second case is when the row's group changes.  Because our algorithm first randomly breaks any ties, we assume there are no tied values in the database when $h$ is computed.

\medskip \noindent \textbf{Case 1: }Changing row retains group

Consider two databases that differ only in that one row is ``moved" from rank $p$ to rank $q$. Let us assume without loss of generality that the row is an element of group 1, so it follows that $n_1 \geq 1$. We define the distance between the rows to be $d = |q-p|$ and $a_i$ to be the number of elements in group $i$ between row $p$ and row $q$ (excluding row $p$ and including row $q$). 

\begin{align*}
    |\statsqu' - \statsqu| &= \bigg | \frac{12}{n(n+1)}\sum_{i=1}^g n_i \bar{r'}_i^2 - 3(n+1) - \frac{12}{n(n+1)}\sum_{i=1}^g n_i \bar{r}_i^2 + 3(n+1) \bigg | \\
    &= \frac{12}{n(n+1)} \bigg | \sum_{i=1}^g n_i \big (\bar{r'}_i^2 - \bar{r}_i^2 \big) \bigg | \\
    &= \frac{12}{n(n+1)} \bigg | \sum_{i=1}^g n_i \bigg ( \frac{1}{n_i^2}\bigg (\sum_{j=1}^{n_i} r'_{ij} \bigg )^2 -  \frac{1}{n_i^2}\bigg (\sum_{j=1}^{n_i} r_{ij} \bigg )^2 \bigg) \bigg | \\
    &= \frac{12}{n(n+1)} \bigg | \sum_{i=1}^g \frac{1}{n_i} \bigg ( \bigg (\sum_{j=1}^{n_i} r'_{ij} \bigg )^2 - \bigg (\sum_{j=1}^{n_i} r_{ij} \bigg )^2 \bigg) \bigg | \\
    &\leq \frac{12}{n(n+1)} \sum_{i=1}^g \bigg | \frac{1}{n_i} \bigg ( \bigg (\sum_{j=1}^{n_i} r'_{ij} \bigg )^2 - \bigg (\sum_{j=1}^{n_i} r_{ij} \bigg )^2 \bigg) \bigg |
\end{align*}

The final step follows by the triangle inequality. Let us define $z_i$ to be $$z_i =  \frac{1}{n_i} \bigg | \bigg (\sum_{j=1}^{n_i} r'_{ij} \bigg )^2 - \bigg (\sum_{j=1}^{n_i} r_{ij} \bigg )^2 \bigg|$$

The ``movement" of the row causes all of the elements of group 1 between $p$ and $q$ to change in the opposite direction of the movement (if $q>p$, then all the ranks decrease by 1, while if $q<p$, then all the ranks increase by 1). This yields a total change in the rankings of group 1 of $d - a_1$. Thus, 
\begin{align*}
    z_1 &= \frac{1}{n_1} \bigg|\Big( \sum_{j=1}^{n_i} r_{ij} \pm (d - a_1) \Big)^2 -\Big (\sum_{j=1}^{n_i} r_{ij}\Big)^2 \bigg|\\
    &= \frac{1}{n_1} \Big| \pm 2 (d - a_1)\sum_{j=1}^{n_i} r_{ij} + (d - a_1)^2 \Big| \\
    &= \frac{1}{n_1} \Big|\pm 2 (d - a_1)n_1\bar{r}_1 + (d - a_1)^2 \Big|\\
    &= \Big| \frac{(d - a_1)^2}{n_1} \pm 2\cdot (d - a_1) \cdot \bar{r}_1 \Big| \\
    &\leq \frac{(n-1)^2}{n_1}+2\cdot (n-1) \cdot \bar{r}_1 \\
    &\leq (n-1)^2 + 2\cdot (n-1) \cdot \frac{2n-n_1+1}{2} \\
    &\leq (n-1)^2 + 2\cdot (n-1) \cdot n
\end{align*}

Similarly,
\begin{align*}
    \sum_{i=2}^g |z_i| &= \sum_{i=2}^g \frac{1}{n_i} \bigg| \Big (\sum_{j=1}^{n_i} r_{ij}\pm a_i\Big)^2 - \Big (\sum_{j=1}^{n_i} r_{ij} \Big)^2 \bigg| \\
    &= \sum_{i=2}^g \frac{1}{n_i} \Big|a_i^2 \pm 2a_i\sum_{j=1}^{n_i} r_{ij} \Big| \\
    &= \sum_{i=2}^g \frac{1}{n_i} \Big|a_i^2 \pm 2a_in_i\bar{r}_i \Big| \\
    &= \sum_{i=2}^{g} \Big |\frac{a_i^2}{n_i} \pm 2a_i\bar{r}_i \Big| \\
    &\leq \sum_{i=2}^{g}\frac{a_i^2}{n_i} + 2\sum_{i=2}^{g}a_i\bar{r}_i \\
    &\leq \sum_{i=2}^{g}\frac{n_i^2}{n_i} + 2\sum_{i=2}^{g} n_i\bar{r}_i \\
    &\leq n + 2\sum_{i=1}^{g} \sum_{j=1}^{n_i} r_{ij} \\
    &\leq n + (n+1)n
\end{align*}

Then,
\begin{align*}
|\statsqu' - \statsqu|&\leq \frac{12}{n(n+1)} \bigg ((n-1)^2 + 2\cdot (n-1) \cdot n+n + (n+1)n \bigg )\\
    &\leq  \frac{12}{n(n+1)} \bigg (4(n+1)n\bigg )\\
    &\leq 48
\end{align*}

Thus with only a row change, $\Delta \statsqu = 48$.

\medskip \noindent \textbf{Case 1: }Changing row changes group

Consider two databases that differ only in that one row's rank is changed from $p$ to $q$ and its group is changed from, without loss of generality, group 1 to group 2. 

We define $a_i$ as previously. We note again that the sign of $a_i$ depends on the relationship between $p$ and $q$. If $q>p$, then the sign is negative; if $q<p$, then the sign is positive.

We begin by introducing a lemma we will use throughout this argument:
\begin{lemma} \label{lma:restterms}
    We can bound this summation as follows:
    $$\sum_{i=3}^g \frac{1}{n_i} \bigg ( \Big ( \sum_{j=1}^{n_i} r_{ij} \pm a_i  \Big)^2 - \Big ( \sum_{j=1}^{n_i} r_{ij} \Big)^2 \bigg ) \leq n^2+2n$$ 
\end{lemma}

\begin{proof}
\begin{align*}
 \sum_{i=3}^g \frac{1}{n_i} \bigg ( \Big ( \sum_{j=1}^{n_i} r_{ij} &\pm a_i  \Big)^2 - \Big ( \sum_{j=1}^{n_i} r_{ij} \Big)^2 \bigg ) \\
 &\leq  \sum_{i=3}^g \frac{1}{n_i} \bigg ( \Big ( \sum_{j=1}^{n_i} r_{ij} + a_i  \Big)^2 - \Big ( \sum_{j=1}^{n_i} r_{ij} \Big)^2 \bigg ) \\
  &= 2\sum_{i=3}^g \sum_{j=1}^{n_i} \frac{a_i r_{ij}}{n_i}+ \sum_{i=3}^g \frac{a_i^2}{n_i}\\
 &\leq 2\sum_{i=3}^g \sum_{j=1}^{n_i} r_{ij}+ \sum_{i=3}^g a_i\\
 &\leq 2\Big(\frac{n(n+1)}{2}\Big) +\sum_{i=3}^g n_i\\
 &\leq (n^2+n)+n \leq n^2 + 2n
\end{align*}
\end{proof}

When we change the group of the moving row, we must consider the sensitivity in each of three sub-cases: \\

\textbf{Case 2a: $n_1>1, \ n_2>0$}
\begin{align*}
    |\statsqu' - \statsqu| &= \bigg | \frac{12}{n(n+1)}\sum_{i=1}^g n'_i \bar{r'}_i^2 - 3(n+1) - \frac{12}{n(n+1)}\sum_{i=1}^g n_i \bar{r}_i^2 + 3(n+1) \bigg | \\
    &= \frac{12}{n(n+1)} \bigg |\sum_{i=1}^g n_i' \bar{r'}_i^2 - \sum_{i=1}^g n_i \bar{r}_i^2 \bigg | \\
    &= \frac{12}{n(n+1)} \bigg | \sum_{i=1}^g \bigg( \frac{1}{n_i'} \bigg (\sum_{j=1}^{n_i} r'_{ij} \bigg )^2 -\frac{1}{n_i} \bigg (\sum_{j=1}^{n_i} r_{ij} \bigg )^2 \bigg) \bigg |\\
\end{align*}   

Let us define $z_i$ to be $$\frac{1}{n_i'} \bigg (\sum_{j=1}^{n_i} r'_{ij} \bigg )^2 -\frac{1}{n_i} \bigg (\sum_{j=1}^{n_i} r_{ij} \bigg )^2 .$$

When the expression inside the absolute value is less than zero, $\Delta h$ can be written as
$$|\statsqu' - \statsqu|=  \frac{12}{n(n+1)} \bigg ( \sum_{i=1}^g -z_i\bigg ).$$  

Then we can bound this quantity as follows.

\begin{align*}
   -z_1 &= - \frac{1}{n_1 -1} \Big (\sum_{j=1}^{n_1} r_{1j} - p - a_1 \Big )^2 + \frac{1}{n_1} \Big (\sum_{j=1}^{n_1} r_{1j} \Big )^2\\
   &\leq -\frac{1}{n_1 -1} \Big (\sum_{j=1}^{n_1} r_{1j} - n - (n_1-1) \Big )^2 + \frac{1}{n_1} \Big (\sum_{j=1}^{n_1} r_{1j} \Big )^2 \\
   &= - \frac{ \Big (\sum_{j=1}^{n_1} r_{1j}\Big )^2}{(n_1 -1)n_1} - \frac{ (n+n_1-1)^2 }{n_1-1}+\frac{ 2 (n+n_1-1)\sum_{j=1}^{n_1} r_{1j} }{n_1-1}\\
   &\leq -0-0+\frac{  (n+n_1-1)(2n-n_1+1)n_1}{n_1-1}\\
   &= (n+n_1-1)(2n-n_1+1)+\frac{ (n+n_1-1)(2n-n_1+1)}{n_1-1}
\end{align*}

Through differentiation of the first term with respect to $n_1$, we get that the maximum occurs when $n_1=\frac{n+2}{2}$. This gives that the first term is bounded by $\frac{9n^2}{4}.$ We then differentiate the second term with respect to $n_1$ and get that the slope is always negative, which means that the maximum of the second term maximum occurs when $n_1=2$. Then the second term is bounded by $2n^2+n-1$.
 
Therefore,  
\begin{align*}
   -z_1 &\leq \frac{9n^2}{4} + 2n^2+n-1 \leq \frac{17n^2}{4}+n-1.
\end{align*}

When $i=2$:
\begin{align*}
    -z_2 &= - \frac{1}{n_2 + 1} \Big (\sum_{j=1}^{n_2} r_{2j} + q - a_2 \Big )^2 + \frac{1}{n_2} \Big (\sum_{j=1}^{n_2} r_{2j} \Big )^2\\
    &\leq - \frac{1}{n_2 + 1} \Big (\sum_{j=1}^{n_2} r_{2j} + 1 - n_2 \Big )^2 + \frac{1}{n_2} \Big (\sum_{j=1}^{n_2} r_{2j} \Big )^2 \\
    &=\frac{\Big (\sum_{j=1}^{n_2} r_{2j} \Big )^2}{n_2(n_2+1)} -\frac{(1-n_2)^2}{n_2+1}+\frac{2(n_2-1)\Big (\sum_{j=1}^{n_2} r_{2j} \Big )}{n_2+1}\\
    &\leq \frac{(2n-n_2+1)^2n_2}{4(n_2+1)}-0 +\frac{(n_2-1) (2n-n_2+1)n_2}{n_2+1}\\
    &\leq \frac{(2n-n_2+1)^2}{4}  +(n_2-1) (2n-n_2+1)\\
\end{align*}

It is clear that the first term is bounded by $n^2$. Then for the other term, we can take the derivative with respect to $n_1$, which gives that the maximum occurs when $n_2=n+1$. But as $n_2 \leq n$ and the graph is concave down, we can use $n_2=n$. Thus, that term is bounded by $(n-1) (2n-n+1) = n^2-1.$

Therefore,  
\begin{align*}
   -z_2 &\leq 2n^2-1.
\end{align*}

When $i>2$, we consider the sum of them:
\begin{align*}
-\sum_{i=3}^{g} z_i &= -\sum_{i=3}^{g} \frac{1}{n_i} \bigg ( \Big ( \sum_{j=1}^{n_i} r_{ij} - a_i  \Big)^2 - \Big ( \sum_{j=1}^{n_i} r_{ij} \Big)^2 \bigg ) \\
&=  -\sum_{i=3}^g n_i+2\sum_{i=3}^g \sum_{j=1}^{n_i} r_{ij}\\
&\leq  -(n-n_1-n_2)+2\sum_{i=3}^g \sum_{j=1}^{n_i} r_{ij}\\
&\leq - (n-n_1-n_2)+2\cdot \frac{(n_1+n_2+1+n)(n-n_1-n_2)}{2} \\
&= - (n-n_1-n_2)+n^2-(n_1+n_2)^2 + (n-n_1-n_2)\\
&\leq n^2-9
\end{align*}




Then combining them together, it follows that 
\begin{align*}
    |\statsqu' - \statsqu| &\leq \frac{12}{n(n+1)} \bigg (\frac{17n^2}{4}+n-1+ n^2 +n^2-1 + n^2-9 \bigg )\\
    &=\frac{12}{n(n+1)}\bigg(\frac{29n^2}{4}+n-1\bigg)\\
    &=87
\end{align*}

When the expression inside the absolute value bars is greater than zero, we can bound this quantity as follows:

\begin{align*}
z_1 &= \frac{1}{n_1 -1} \Big (\sum_{j=1}^{n_1} r_{1j} - p \pm a_1 \Big )^2 - \frac{1}{n_1} \Big (\sum_{j=1}^{n_1} r_{1j} \Big )^2 \\
&\leq \frac{1}{n_1 -1} \Big (\sum_{j=1}^{n_1} r_{1j} - p + a_1 \Big )^2 - \frac{1}{n_1} \Big (\sum_{j=1}^{n_1} r_{1j} \Big )^2
\end{align*}

Then consider $-p+a_1$. As it is $+a_1$, it matches the case $p>q$, and thus $-p+a_1$ is negative and we can drop it directly.

It follows that 
\begin{align*}
z_1 &\leq \frac{1}{n_1 -1} \Big (\sum_{j=1}^{n_1} r_{1j} \Big )^2 - \frac{1}{n_1} \Big (\sum_{j=1}^{n_1} r_{1j} \Big )^2\\
&=\frac{\Big (\sum_{j=1}^{n_1} r_{1j}\Big )^2 }{n_1(n_1-1)}\\
&\leq \frac{(2n-n_1+1)^2n_1^2}{4(n_1-1)n_1}\\
&\leq \frac{(2n-2+1)^2n_1^2}{2}\\
&\leq \frac{4n^2-4n+1}{2}
\end{align*}

Similarly,
\begin{align*}
z_2 &= \frac{1}{n_2 + 1} \Big (\sum_{j=1}^{n_2} r_{2j} + q \pm a_2 \Big )^2-\frac{1}{n_2}\Big (\sum_{j=1}^{n_2} r_{2j} \Big )^2 \\
&\leq \frac{1}{n_2 + 1} \Big (\sum_{j=1}^{n_2} r_{2j} + q + a_2 \Big )^2-\frac{1}{n_2}\Big (\sum_{j=1}^{n_2} r_{2j} \Big )^2 \\
&= \frac{1}{n_2 + 1} \Big (\sum_{j=1}^{n_2} r_{2j} + (q+n_2) \Big )^2-\frac{1}{n_2}\Big (\sum_{j=1}^{n_2} r_{2j} \Big )^2\\
&= -\frac{\Big (\sum_{j=1}^{n_2} r_{2j} \Big )^2}{n_2(n_2+1)}+\frac{(q+n_2)^2}{n_2+1}+\frac{2(q+n_2)\sum_{j=1}^{n_2} r_{2j}}{n_2+1}\\
&\leq -\frac{((1+n_2)n_2)^2}{4(1+n_2)n_2}+\frac{q^2}{n_2+1}+\frac{n_2^2}{n_2+1} \\
&\hspace{.5cm}+\frac{2 q \cdot n_2}{n_2+1}+\frac{(n+n_2)(2n-n_2+1)n_2}{n_2+1}\\
&\leq -\frac{((1+n_2)n_2)^2}{4(1+n_2)n_2}+\frac{n^2}{n_2+1}+n_2+2n+(n+n_2)(2n-n_2+1)\\
\end{align*}

Consider $(n+n_2)(2n-n_2+1)$. The derivative with respect to $n_2$ is $1 + n - 2 n_2$. As the graph of the equation is concave down, when the derivative is $0$, it reaches its maximum. Thus the maximum occurs when $n_2=\frac{n+1}{2}$, which leads to the maximum $\frac{(3n+1)^2}{4}$.\\

Then it follows that,
\begin{align*}
z_2 &\leq -\frac{(1+n_2)n_2}{4}+\frac{n^2}{2}+n+2n+\frac{(3n+1)^2}{4}\\
&= -\frac{1}{2}+\frac{n^2}{2}+3n+\frac{9n^2}{4}+\frac{6n}{4}+\frac{1}{4}\\
&=\frac{11n^2}{4}+\frac{9n}{2}-\frac{1}{4}
\end{align*}

The rest of the terms, by Lemma \ref{lma:restterms}, are bounded by $n^2-2n$. It follows that
\begin{align*}
    |\statsqu' - \statsqu| &\leq \frac{12}{n(n+1)} \Big(\frac{4n^2-4n+1}{2}+\frac{11n^2}{4}+\frac{9n}{2}-\frac{1}{4}+ n^2+2n\Big)\\
    &=\frac{12}{n(n+1)} \Big( \frac{23n^2}{4}+\frac{9n}{2} +\frac{1}{4}\Big)\\
    &\leq 69
\end{align*}

\textbf{Case 2b: $n_1=1, \ n_2>0$}
\begin{align*}
    |\statsqu' - \statsqu| &= \bigg | \frac{12}{n(n+1)}\sum_{i=1}^g n'_i \bar{r'}_i^2 - 3(n+1) \\
    & \qquad - \frac{12}{n(n+1)}\sum_{i=1}^g n_i \bar{r}_i^2 + 3(n+1) \bigg | \\
    &= \frac{12}{n(n+1)} \bigg |\sum_{i=1}^g n_i' \bar{r'}_i^2 - \sum_{i=1}^g n_i \bar{r}_i^2 \bigg | \\
    &= \frac{12}{n(n+1)} \bigg |(n_1-1)\bar{r'}_1^2 -n_1\bar{r}_1^2+(n_2+1)\bar{r'}_2^2 -n_2\bar{r}_2^2 \\
    & \qquad +\sum_{i=3}^g n_i (\bar{r'}_i^2 - \bar{r}_i^2) \bigg |\\
    &\leq \frac{12}{n(n+1)} \bigg( |-p^2|+|(n_2+1)\bar{r'}_2^2 -n_2\bar{r}_2^2 | +n^2+2n\bigg)\\
    &\leq  \frac{12}{n(n+1)} \bigg( |(n_2+1)\bar{r'}_2^2 -n_2\bar{r}_2^2 | +2n^2+2n\bigg)
\end{align*}

Now consider the term $|(n_2+1)\bar{r'}_2^2 -n_2\bar{r}_2^2|$,
\begin{align*}
    |(n_2+1)\bar{r'}_2^2 -n_2\bar{r}_2^2| &=\Big| \frac{1}{n_2 + 1} \Big (\sum_{j=1}^{n_2} r_{2j} + q \pm a_2 \Big )^2-\frac{1}{n_2}\Big (\sum_{j=1}^{n_2} r_{2j} \Big )^2\Big|\\
    &\leq \bigg|-\frac{\Big (\sum_{j=1}^{n_2} r_{2j} \Big )^2}{n_2(n_2+1)} \bigg|+\frac{(q\pm a_2)^2}{n_2+1} \\
    & \qquad +\frac{2|q\pm a_2|\sum_{j=1}^{n_2} r_{2j}}{n_2+1}\\
    &\leq \frac{(2n-n_2+1)^2n_2^2}{4n_2(n_2+1)}+\frac{(q+ a_2)^2}{n_2+1} \\
    & \qquad +\frac{2(q+ a_2)}{n_2+1}\cdot \frac{(2n-n_2+1)n_2}{2}\\
    &\leq n^2 + \frac{(n+ n_2)^2}{n_2+1}+(2n-n_2+1)(n+ n_2)\\
    &= n^2 + \frac{n^2}{n_2+1}+\frac{n_2}{n_2+1}+\frac{2n\cdot n_2}{n_2+1} \\
    & \qquad + 2n^2+nn_2-n_2^2+n_2+n\\
    &\leq n^2+\frac{n^2}{2}+1+2n+2n^2+n^2+2n\\
    &\leq \frac{9}{2}n^2+4n+1
\end{align*}

It follows that 
\begin{align*}
    |\statsqu' - \statsqu| &\leq \frac{12}{n(n+1)} \bigg( \frac{9}{2}n^2+4n+1 +(2n^2+2n)\bigg)\\
    &= \frac{12}{n(n+1)} \bigg( \frac{13}{2}n^2+6n+1\bigg)\\
    &\leq \frac{12}{n(n+1)} \bigg( \frac{13}{2}n(n+1)-\frac{n}{2} + \frac{n}{2} \bigg)\\
    &\leq 78\\
\end{align*}

\textbf{Case 2c: $n_2=0, \ n_1 > 1$}

\begin{align*}
    |\statsqu' - \statsqu| &= \bigg | \frac{12}{n(n+1)}\sum_{i=1}^g n'_i \bar{r'}_i^2 - 3(n+1) - \frac{12}{n(n+1)}\sum_{i=1}^g n_i \bar{r}_i^2 + 3(n+1) \bigg | \\
    &= \frac{12}{n(n+1)} \bigg |\sum_{i=1}^g n_i' \bar{r'}_i^2 - \sum_{i=1}^g n_i \bar{r}_i^2 \bigg | \\
    &= \frac{12}{n(n+1)} \bigg |(n_1-1)\bar{r'}_1^2 -n_1\bar{r}_1^2+(n_2+1)\bar{r'}_2^2 -n_2\bar{r}_2^2 \\
    & \qquad +\sum_{i=3}^g n_i (\bar{r'}_i^2 - \bar{r}_i^2) \bigg |\\
    &\leq \frac{12}{n(n+1)} \bigg(|(n_1-1)\bar{r'}_1^2 -n_1\bar{r}_1^2 |+q^2 +\sum_{i=3}^{g} n_i (\bar{r'}_i^2 - \bar{r}_i^2)\bigg)\\
    &\leq  \frac{12}{n(n+1)} \bigg( |(n_1-1)\bar{r'}_1^2 -n_1\bar{r}_1^2 | +(n^2+2n+n^2)\bigg)
\end{align*}

Now consider the term $|(n_1-1)\bar{r'}_1^2 -n_1\bar{r}_1^2|$,
\begin{align*}
    |(n_1-&1)\bar{r'}_1^2 -n_1\bar{r}_1^2| \\
    &=\Big| \frac{1}{n_1 - 1} \Big (\sum_{j=1}^{n_1} r_{1j} -p \pm a_1 \Big )^2-\frac{1}{n_1}\Big (\sum_{j=1}^{n_1} r_{1j} \Big )^2\Big|\\
    &\leq \bigg|\frac{\Big (\sum_{j=1}^{n_1} r_{1j} \Big )^2}{n_1(n_1-1)}+\frac{(-p\pm a_1)^2}{n_1-1}  \\
    &\qquad +\frac{2(-p\pm a_1)\sum_{j=1}^{n_1} r_{1j}}{n_1-1}\bigg|
\end{align*}

note that $-p\pm a_1$ is always negative. If the sign is $-$, then this directly follows. If the sign is $+$, then $p>q$, which indicates $p>a_1$ and thus it again follows that the term is negative. 

Then if the whole inside expression is negative, 
\begin{align*}
    |(n_1-1)\bar{r'}_1^2 -n_1\bar{r}_1^2| &\leq \frac{2(p\pm a_1)\sum_{j=1}^{n_1} r_{1j}}{n_1-1} \\
    &\leq \frac{2(-p+ a_1)\sum_{j=1}^{n_1} r_{1j}}{n_1-1} \\
    &\leq \frac{(n+n_1)(2n-n_1+1)n_1}{n_1-1}\\
    &\leq 2(n+n_1)(2n-n_1+1)\\
    &\leq 2 \cdot \frac{(3n+1)^2}{4}\\
    &\leq \frac{9n^2}{2}+3n+\frac{1}{2}
\end{align*}

Then if the whole inside expression is positive, 
\begin{align*}
    |(n_1-1)\bar{r'}_1^2 -n_1\bar{r}_1^2| &\leq \frac{\Big (\sum_{j=1}^{n_1} r_{1j} \Big )^2}{n_1(n_1-1)}+\frac{(-p\pm a_1)^2}{n_1-1}\\
    &\leq \frac{n_1}{n_1-1}\cdot \frac{(2n-n_1+1)^2}{4}+\frac{(n+n_1)^2}{n_1-1}\\
    &\leq 2n^2+\frac{n^2}{n_1-1}+\frac{2nn_1}{n_1-1}+\frac{n_1^2}{n_1-1}\\
    &\leq 2n^2+n^2+4n+2n\\
    &\leq 3n^2+6n
\end{align*}

As $n\geq 2$, then $\frac{3n^2}{2}-3n+\frac{1}{2} \geq 0$. It follows that $\frac{9n^2}{2}+3n+\frac{1}{2} \geq 3n^2+6$, and thus $$|(n_1-1)\bar{r'}_1^2 -n_1\bar{r}_1^2| \leq \frac{9n^2}{2}+3n+\frac{1}{2}.$$

Therefore, \begin{align*}
    |\statsqu' - \statsqu| &\leq  \frac{12}{n(n+1)} \bigg( \frac{9n^2}{2}+3n+\frac{1}{2} +(n^2+2n+n^2)\bigg)\\
    &\leq \frac{12}{n(n+1)} \bigg( \frac{13n^2}{2}+5n+\frac{1}{2} \bigg)\\
    &\leq \frac{12}{n(n+1)} \bigg( \frac{13n(n+1)}{2}-\frac{3n}{2}+n\bigg)\\
    &\leq 78
\end{align*}

\end{proof}

\subsection{Proof of Theorem \ref{thm:pKW}}\label{sec:kw_privacy_pf}

\begin{theorem*}
    Algorithm $\privKW$ is $\epsilon$-differentially private.
\end{theorem*}
\begin{proof}
There are two sources of randomness inside Algorithm $\privKW$: the random breaking of ties and the Laplace noise. So let us write $\privKW(\data) := \privKW(\data;r_1,r_2),$ where $r_1$ indicates the randomness of breaking ties and $r_2$ indicates the randomness of the Laplace noise.

For an arbitrary choice of $r_1^*$, let $\privKW^{r_1^*}(\data;r_2)$, be the function obtained by fixing $r_1 = r_1^*$ in \privKW. $\privKW^{r_1^*}$ is simply a Laplace mechanism, so $\epsilon$-differential privacy follows by Theorems \ref{thm:squaresensi} and \ref{thm:lm}.

Now \privKW is simply randomly outputting one of a large (but finite) set of possible  $\epsilon$-differentially private outputs.  Doing so is always private, so \privKW is private.
\end{proof}


\subsection{Simplification of \statabs}\label{sec:kw_abs_deriv}

Changing the squares to absolute values in the Kruskal-Wallis yields the formula:
\begin{equation}
    \statabs = (n-1)\frac{\sum_{i=1}^{g}n_i|\bar{r}_i-\bar{\bar{r}}|}{\sum_{i=1}^{g}\sum_{j=1}^{n_i}|r_{ij}-\bar{\bar{r}}|}
\end{equation}

When there are no tied rows in the data, this formula can be simplified a bit. The simplification depends on whether $n$ is even or odd.

\textbf{When $n$ is even:}

\begin{align*}
    \sum_{i=1}^{g}\sum_{j=1}^{n_i}|r_{ij}-\bar{\bar{r}}| &= \Big(\frac{n+1}{2} - 1 \Big) + \Big(\frac{n+1}{2} - 2 \Big) + \ldots + \Big(\frac{n+1}{2} - \frac{n}{2} \Big) \\
    & \quad + \Big(- \frac{n+1}{2} + \frac{n}{2} + 1 \Big) + \Big(- \frac{n+1}{2} + \frac{n}{2} + 2 \Big) + \ldots\\
    & \quad + \Big(- \frac{n+1}{2} + \frac{n}{2} + \frac{n}{2} \Big) \\
    &= -1 -2 - \ldots - \frac{n}{2} + \Big(\frac{n}{2} + 1 \Big) + \Big(\frac{n}{2} + 2 \Big) + \ldots \\
    & \quad + \Big(\frac{n}{2} + \frac{n}{2} \Big) \\
    &= \frac{n}{2} \cdot \frac{n}{2} = \frac{n^2}{4}
\end{align*}

Thus, the simplified form of $\statabs$ for even $n$ is:
\begin{equation}
    \statabs = \frac{4(n-1)}{n^2}\sum_{i=1}^{g}n_i|\bar{r}_i-\bar{\bar{r}}|
\end{equation}

\textbf{When $n$ is odd}

\begin{align*}
    \sum_{i=1}^{g}\sum_{j=1}^{n_i}|r_{ij}-\bar{\bar{r}}|&= \Big(\frac{n+1}{2}-1\Big)+\ldots+\Big(\frac{n+1}{2}-(\frac{n+1}{2}-1)\Big)\\
    & \quad +\Big(\frac{n+1}{2}-\frac{n+1}{2}\Big)+\Big(\Big(\frac{n+1}{2}+1\Big)-\frac{n+1}{2}\Big)\\
    & \quad +\ldots+\Big(n-\frac{n+1}{2}\Big)\\
    &=\Big(\frac{n-1}{2}\Big)+\ldots+1+0+1+\ldots +\Big(\frac{n-1}{2}\Big)\\
    &= 2 \cdot\bigg(1+2+\ldots+\Big(\frac{n-1}{2}\Big)\bigg)\\
    &=  \frac{n-1}{2}\cdot\frac{n+1}{2}
\end{align*}

Thus, the simplified form of $\statabs$ for odd $n$ is:
\begin{equation}
    \statabs = \frac{4}{n+1}\sum_{i=1}^{g}n_i|\bar{r}_i-\bar{\bar{r}}|
\end{equation}

As $\frac{4}{n+1} \geq \frac{4(n-1)}{n^2}$, we will use $\tilde{h}= \tilde{h}_{O} = \frac{4}{n+1}\sum_{i=1}^{g}n_i|\bar{r}_i-\bar{\bar{r}}|$ for sensitivity analysis.

\subsection{Proof of Theorem \ref{thm:abssensi}}\label{sec:kw_abs_sens}

\begin{proof}
Here we bound the sensitivity of our \KWA test statistic.

\begin{theorem*} 
The sensitivity of $\KWA$ is bounded by 8.
\end{theorem*}

As $\frac{4}{n+1} \geq \frac{4(n-1)}{n^2}$ for $n>0$, we will use \\$\statabs= \frac{4}{n+1}\sum_{i=1}^{g}n_i|\bar{r}_i-\bar{\bar{r}}|$ for sensitivity analysis. As in the proof of Theorem \ref{thm:squaresensi}, the proof proceeds in cases.  Again, we consider neighboring databases $\data, \ \data'$ that differ by a change in one row, and our cases are differentiated by whether or not the group of that row changes.

\medskip \noindent \textbf{Case 1: }Changing row retains group

Consider a ``movement" of a row from row $p$ to row $q$. Assume $p \leq q$ (the other case is symmetric).

\begin{align*}
    \Delta \statabs = \frac{4}{n+1} \left | \sum_{i=1}^g n_i \left(\left | \bar{r}_i' - \frac{n+1}{2} \right | - \left | \bar{r}_i - \frac{n+1}{2} \right |\right)\right |
\end{align*}

Then we can analyze the expression inside the outer absolute value.

When $i=1$, 
\begin{align*}
    n_1 \bigg(\Big | \bar{r}_1' - &\frac{n+1}{2} \Big | - \Big | \bar{r}_1 - \frac{n+1}{2} \Big |\bigg) \\
    &=n_1 \bigg(\Big | \frac{n_1\bar{r}_1-p+q-a_1}{n_1} - \frac{n+1}{2} \Big | \\
    & \qquad - \Big | \bar{r}_i - \frac{n+1}{2} \Big |\bigg)\\
    &= n_1 \bigg(\Big | \bar{r}_1+\frac{-p+q-a_1}{n_1} - \frac{n+1}{2} \Big | \\ 
    & \qquad - \Big | \bar{r}_i - \frac{n+1}{2} \Big |\bigg)\\
    &\leq  n_1\Big |\frac{-p+q-a_1}{n_1} \Big |\\
    &=q-p-a_1
\end{align*}

When $i \neq 1$,
\begin{align*}
    n_i \bigg(\Big | \bar{r}_i' - \frac{n+1}{2} \Big | - \Big | \bar{r}_i - \frac{n+1}{2} \Big |\bigg) &= n_i \bigg(\Big|\frac{n_i\bar{r}_i-a_i}{n_i}-\frac{n+1}{2}\Big|\\
    & \qquad -\Big|\bar{r}_i-\frac{n+1}{2}\Big|\bigg)\\
    &\leq \frac{a_i}{n_i} \cdot n_i\\
    &= a_i
\end{align*}

It follows that,
\begin{align*}
    \Delta \statabs &= \frac{4}{n+1} \bigg(q-p-a_1+\sum_{i=2}^{g}a_i\bigg)\\
    &= \frac{4}{n+1} \bigg(q-p-2a_1+(q-p)\bigg)\\
    &\leq \frac{8}{n+1} \cdot (n-1)\\
    &\leq 8 - \frac{16}{n+1}
\end{align*}

\medskip \noindent \textbf{Case 2: }Changing row changes group

Consider a ``movement" of a row from row $p$ to row $q$, where, without loss of generality, its group changes from group 1 to group 2. As before, we must consider three cases: \\

\textbf{Case 2a: $n_1>1, \ n_2>0$}

\begin{align*}
    \Delta \statabs &= \frac{4}{n+1} \bigg | \sum_{i=1}^g \bigg(n_i' \Big | \bar{r}_i' - \frac{n+1}{2} \Big | - n_i \Big | \bar{r}_i - \frac{n+1}{2} \Big | \bigg) \bigg |
\end{align*} 

For $i = 1$,

\begin{align*}
     n_1' \Big | \bar{r}_1' - &\frac{n+1}{2} \Big | - n_1 \Big | \bar{r}_1 - \frac{n+1}{2} \Big | \\
     &= (n_1 - 1) \Big | \frac{n_1 \bar{r}_1 - p - a_1}{n_1 - 1} - \frac{n+1}{2} \Big | \\
     & \qquad - n_1 \Big | \frac{n_1 \bar{r}_1}{n_1} - \frac{n+1}{2} \Big | \\
     &= \Big | n_1 \bar{r}_1 - p - a_1 - \frac{(n+1)(n_1 - 1) }{2} \Big | \\
     & \qquad - \Big | n_1 \bar{r}_1 - \frac{n_1(n+1)}{2} \Big | \\
     &= \Big | n_1 \bar{r}_1 - \frac{(n+1)n_1}{2} - p - a_1 + \frac{n+1}{2} \Big | \\
     & \qquad - \Big | n_1 \bar{r}_1 - \frac{n_1(n+1)}{2} \Big | \\
     &\leq \Big | n_1 \bar{r}_1 - \frac{(n+1)n_1}{2}\Big | + \Big | - p - a_1 + \frac{n+1}{2} \Big | \\
     & \qquad - \Big | n_1 \bar{r}_1 - \frac{n_1(n+1)}{2} \Big | \\
     &= \Big | - p - a_1 + \frac{n+1}{2} \Big |
\end{align*}

For $i = 2$,

\begin{align*}
     n_2' \Big | \bar{r}_2' - \frac{n+1}{2} \Big | - &n_2 \Big | \bar{r}_2 - \frac{n+1}{2} \Big | \\
     &= (n_2 + 1) \Big | \frac{n_2 \bar{r}_2 + q - a_2}{n_2 + 1} - \frac{n+1}{2} \Big | \\
     & \qquad - n_2 \Big | \frac{n_2 \bar{r}_2}{n_2} - \frac{n+1}{2} \Big | \\
     &= \Big | n_2 \bar{r}_2 + q - a_2 - \frac{(n_2 + 1)(n+1)}{2} \Big | \\
     & \qquad - \Big | n_2 \bar{r}_2 - \frac{n_2(n+1)}{2} \Big | \\
     &= \Big | n_2 \bar{r}_2 - \frac{n_2(n+1)}{2} + q - a_2 - \frac{n+1}{2} \Big | \\
     & \qquad - \Big | n_2 \bar{r}_2 - \frac{n_2(n+1)}{2} \Big | \\
     &\leq \Big | n_2 \bar{r}_2 - \frac{n_2(n+1)}{2} \Big | + \Big | q - a_2 - \frac{n+1}{2} \Big | \\
     & \qquad - \Big | n_2 \bar{r}_2 - \frac{n_2(n+1)}{2} \Big | \\
     &= \Big | q - a_2 - \frac{n+1}{2} \Big |
\end{align*}

For $i > 2$,

\begin{align*}
    n_i \bigg(\Big | \bar{r}_i' - \frac{n+1}{2} \Big | &- \Big | \bar{r}_i - \frac{n+1}{2} \Big |\bigg) \\
    &= n_i \bigg(\Big|\frac{n_i\bar{r}_i-a_i}{n_i}-\frac{n+1}{2}\Big|-\Big|\bar{r}_i-\frac{n+1}{2}\Big|\bigg)\\
    &\leq \frac{a_i}{n_i} \cdot n_i\\
    &= a_i
\end{align*}

Thus, we can bound the sensitivity by:
\begin{align*}
    \Delta \statabs &\leq \frac{4}{n+1} \bigg ( \, \Big | - p - a_1 + \frac{n+1}{2} \Big | + \Big | q - a_2 - \frac{n+1}{2} \Big | + \sum_{i=3}^g a_i \; \bigg ) \\
    &\leq \frac{4}{n+1} \bigg ( \, \Big | - p + \frac{n+1}{2} \Big | + \Big | q - \frac{n+1}{2} \Big | + \sum_{i=1}^g a_i \; \bigg ) \\
    &\leq \frac{4}{n+1} \bigg ( \, \Big | - p + \frac{n+1}{2} \Big | + \Big | q - \frac{n+1}{2} \Big | + q - p \; \bigg )\\
    &< \frac{4}{n+1} \bigg ( \,  \frac{n+1}{2} +\frac{n+1}{2} + q - p \; \bigg )\\
    &\leq \frac{4}{n+1} \cdot 2n\\
    &\leq 8
\end{align*}

\textbf{Case 2b: $n_1=1, \ n_2>0$}

As before, 
\begin{align*}
    \Delta \statabs &= \frac{4}{n+1} \bigg | \sum_{i=1}^g \bigg(n_i' \Big | \bar{r}_i' - \frac{n+1}{2} \Big | - n_i \Big | \bar{r}_i - \frac{n+1}{2} \Big | \bigg) \bigg |
\end{align*} 

For $i = 1$,
\begin{align*}
     n_1' \Big | \bar{r}_1' - \frac{n+1}{2} \Big | - n_1 \Big | \bar{r}_1 - \frac{n+1}{2} \Big | &= 0 - 1 \cdot \Big | p - \frac{n+1}{2} \Big |\\
     &\leq - \frac{n-1}{2} 
\end{align*}

The $i \neq 1$ cases are as in case I. Thus,
\begin{align*}
    \Delta \statabs &= \frac{4}{n+1} \bigg | - \frac{n-1}{2} + \Big | q \pm a_2 - \frac{n+1}{2} \Big | + \sum_{i=3}^g a_i \bigg |\\
    &\leq \frac{4}{n+1} \bigg (\Big| - \frac{n-1}{2}\Big| + \Big | q - \frac{n+1}{2} \Big | + \sum_{i=2}^g a_i \bigg )\\
    &\leq \frac{4}{n+1} \bigg ( \frac{n-1}{2} +  \frac{n-1}{2} + n\bigg )\\
    &\leq \frac{4}{n+1}  ( n-1 + n)\\
    &\leq 8
\end{align*} 

\textbf{Case 2c: $n_2=0, \ n_1 > 1$}

As before,
\begin{align*}
    \Delta \statabs &= \frac{4}{n+1} \bigg | \sum_{i=1}^g \bigg(n_i' \Big | \bar{r}_i' - \frac{n+1}{2} \Big | - n_i \Big | \bar{r}_i - \frac{n+1}{2} \Big | \bigg) \bigg |
\end{align*}

For $i=2$,$$n_2' \Big | \bar{r}_2' - \frac{n+1}{2} \Big | - n_2 \Big | \bar{r}_2 - \frac{n+1}{2} \Big | = \Big |q - \frac{n+1}{2} \Big | \leq \frac{n-1}{2}.$$

For $i \neq 2$, it is the same as in case II. Thus,
\begin{align*}
    \Delta \statabs &= \frac{4}{n+1} \bigg | \Big | - p \pm a_1 + \frac{n+1}{2} \Big |+\frac{n-1}{2}+\sum_{i=3}^{g}a_i \bigg |\\
    &\leq \frac{4}{n+1} \bigg ( \Big | - p  + \frac{n+1}{2} \Big |+\frac{n-1}{2}+\sum_{i=2}^{g}a_i \bigg )\\
    &\leq \frac{4}{n+1} \bigg ( \frac{n-1}{2} +\frac{n-1}{2}+n \bigg )\\
    &\leq \frac{4}{n+1} \bigg ( n-1+n )\\
    &\leq 8
\end{align*}
\end{proof}

\subsection{Type I Error Rate of $\privKWA$}\label{sec:kw_abs_t1}
When we compute the reference distribution for $\statabs$ in the public setting, we require the size of each group to be known. In the private setting, however, the group sizes are not public and therefore cannot be used in the reference distribution. Thus, we must determine the worst case group sizes and use these to compute the reference distribution in Algorithm $\privKWP$.

\begin{theorem} \label{thm:exp_val}
Let $H_{abs}$ be the random variable of which $\statabs$ is an instance. Then the expected value of $H_{abs}$ is maximized when the sizes of the groups are $n/g$.
\end{theorem}

\begin{proof}
Let $\bar{R}_i$ be the random variable of which $\bar{r}_i$ is an instance. The distribution of $\bar{R}_i$ is known to be $\bar{R}_i \sim \mathcal{N} \left(\mu = \frac{n+1}{2}, \sigma^2 = \frac{n^2 - 1}{12n_i} \right)$. Let $Z_i$ be a standard normal random variable: $Z_i \sim \mathcal{N} \left( \mu = 0, \sigma^2 = 1 \right)$. We can relate $\bar{R}_i$ and $Z_i$ by $$\bar{R}_i = \sqrt{\frac{n^2 - 1}{12n_i}}Z_i + \frac{n+1}{2}$$.

Additionally, it is known that a half normal distribution from a standard normal, $|Z_i|$, has mean $\sqrt{\frac{2}{\pi}}$ and variance $1-\frac{2}{\pi}$.

We will use the formulation of $H_{abs}$ when $n$ is odd. The algebra is identical for an even $n$.
\begin{align*}
    H_{abs} &= \frac{4}{n+1}\sum_{i=1}^{g}n_i \left|\bar{R}_i-\frac{n+1}{2} \right| \\
    &= \frac{4}{n+1}\sum_{i=1}^{g}n_i \left|\sqrt{\frac{n^2 - 1}{12n_i}}Z_i + \frac{n+1}{2}-\frac{n+1}{2} \right| \\
    &= \sqrt{\frac{4(n-1)}{3(n+1)}}\sum_{i=1}^{g}\sqrt{n_i} \left|Z_i \right|
\end{align*}

Let us consider the expected value of $H_{abs}$:
\begin{align*}
    \mathbf{E}(H_{abs}) &= \mathbf{E} \left(\sqrt{\frac{4(n-1)}{3(n+1)}}\sum_{i=1}^{g}\sqrt{n_i} \left|Z_i \right| \right) \\
    &= \sqrt{\frac{4(n-1)}{3(n+1)}}\sum_{i=1}^{g}  \sqrt{n_i} \; \mathbf{E}( | Z_i | ) \\
    &= \sqrt{\frac{4(n-1)}{3(n+1)}}\sum_{i=1}^{g}  \sqrt{n_i} \; \mathbf{E}( | Z_i | ) \\
    &= \sqrt{\frac{4(n-1)}{3(n+1)}}\sum_{i=1}^{g}  \sqrt{n_i} \cdot \sqrt{\frac{2}{\pi}} \\
    &= \sqrt{\frac{8(n-1)}{3\pi(n+1)}}\sum_{i=1}^{g}  \sqrt{n_i}
\end{align*}

The expected value of $H_{abs}$ is maximized when $\sum_{i=1}^{g}  \sqrt{n_i}$ is maximized. This will occur when the sizes of each of the groups are $n/g$.
\end{proof}

\subsection{Proof of Theorem \ref{thm:mwsensi}}\label{sec:mw_sens}

Recall that $$\MW = U = \min \{ U_1, U_2 \},$$ where $$U_1 = r_1 - \frac{n_1(n_1 + 1)}{2} \qquad \textup{and} \qquad U_2 = r_2 - \frac{n_2(n_2 + 1)}{2}. $$

We now prove the following:

\begin{theorem*} [Sensitivity of \MW] 
The local sensitivity is given by $LS_{\MW}(\data) = \max \{ n_1, n_2 \}$, where $n_1$ and $n_2$ are the sizes of the two groups in $\data$.
\end{theorem*}

\begin{proof}  

We bound the sensitivity of $U_1$ and $U_2$, and note that this bound also applies to their minimum.  We consider two cases, one where the row that differs between $\data$ and $\data'$ retains its group and where it changes its group.

\textbf{``Changing" row retains group}

Consider two databases $\data$ and $\data'$ that differ only in row $j$ and where row $j$ is in the same group in both $\data$ and $\data'$. Without loss of generality, suppose $j$ is in group $1$ and that $v_{j}' > v_{j}$ (where $v_{j}$ and $v_{j}'$ are the values of row $j$ in $\data$ and $\data'$ respectively). Suppose that the row was originally (in $\data$) part of a tie in ranks with $b$ other rows and it becomes part of a tie in ranks with $d$ other rows (in $\data'$). Let the number of rows between these two groups of ties be $c$ and let the number of rows before the first tie be $a$ and the number of rows after the second tie be $e$. This setup is shown in Figure \ref{fig:ls}.

\begin{figure} [!htb]
    \centering
    \includegraphics[width=\linewidth]{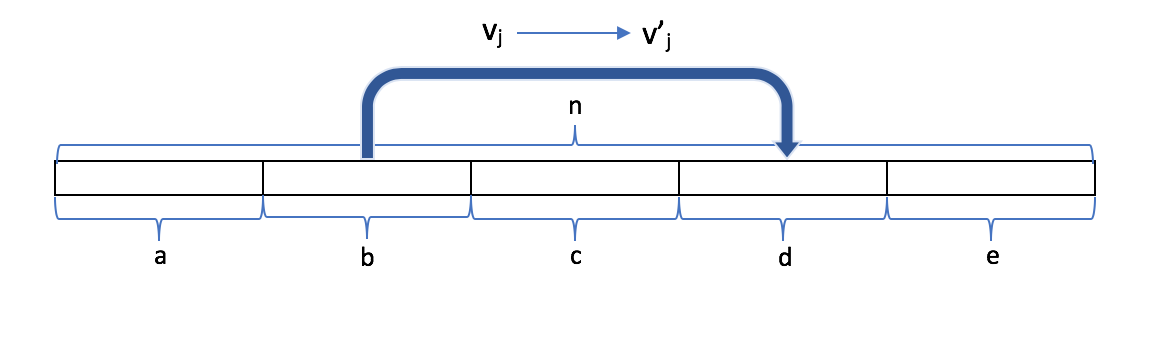}
    \caption{ $\data$ and $\data'$ differ only in that $v_j$ is ``changed" such that its rank is ``moved" from region $b$ to region $d$.} \label{fig:ls}
\end{figure}

The rank of $v_j$ in $\data$ is $\frac{a+1+a+b}{2}$ and the rank of $v_j'$ in $\data'$ is $\frac{a+b+c+a+b+c+d}{2}$. Note that when an element is removed from a tie, the rank of every element in the tie will change by $\frac{1}{2}$.

Let $b_1$ and $b_2$ be the number of the $b$ elements (i.e., those tied with row $j$) in database $\data$ that are in groups 1 and 2, respectively. Let $c_1$, $c_2$, $d_1$, and $d_2$ be similiarly defined for the regions of $c$ and $d$ elements.

We find that the sensitivity, $LS_{U_1}$, is the sum of five effects. The first is the removal of the row in $\data$, which decreases the sum by $\frac{a+1+a+b}{2}$. The second is the addition of the row in $\data'$, which increases the sum by $\frac{a+b+c+a+b+c+d}{2}$. The third is the effect on the rank of the $b-1$ tied elements (excluding row $j$ itself), which causes the sum to decrease by $\frac{b_1-1}{2}$. The fourth is the effect on the $d$ tied elements, which causes the sum to decrease by $\frac{d_1}{2}$. The fifth is decrease in rank for the $c$ elements that are now below instead of above row $j$, which causes the sum to decrease by $c_1$.

\begin{align*}
 LS_{U_1}  &= \bigg|-\frac{a+1+a+b}{2}+\frac{a+b+c+a+b+c+d}{2}\\
 & \qquad-\frac{b_1-1}{2}-\frac{d_1}{2}-c_1 \bigg|\\
 &=\frac{(b-b_1)+(2c-2c_1)+(d-d_1)}{2}\\
 &= \frac{b_2}{2}+c_2+\frac{d_2}{2}\\
 &\leq n_2\\
\end{align*}

Similarly, considering how ranks of group $2$ elements get changed, we get 
\begin{align*}
 LS_{U_2}  &= \left|-\frac{b_2}{2} -c_2-\frac{d_2}{2}\right|\\
 &\leq n_2
\end{align*}

The proof is identical for the movement of a row in group 2 and yields a bound of $n_1$.\\

\textbf{"Changing" row switches groups}

Now let us consider $\data$ and $\data'$ such that one row is moved to a new group.  Without loss of generality, we assume the row moves from group $1$ to group $2$. We also assume as before that $v_j' > v_j$ and we define $a$, $b$, $c$, $d$ and $e$ as before.

For $U_1$, the removal of the element itself from group $1$ first causes the rank sum to decrease by $\frac{a+1+a+b}{2}$. The decrease in $n_1$ causes the second term of $U_1$ change from $\frac{n_1(n_1+1)}{2}$ to $\frac{(n_1-1)n_1}{2}$, which leads to the contribution to $LS_{U_1}$ to be $-\frac{(n_1-1)n_1}{2}+\frac{n_1(n_1+1)}{2}$. The rest of the effects on $U_1$ is the same as in the previous part. Thus,

\begin{align*}
 LS_{U_1}  &= \bigg|-\frac{a+1+a+b}{2}-\frac{(n_1-1)n_1}{2}+\frac{n_1(n_1+1)}{2} \\
 &\qquad\qquad\qquad -\frac{b_1-1}{2}-\frac{d_1}{2}-c_1 \bigg|\\
 &=\left|-\frac{2a+b+1}{2}+n_1-\frac{b_1-1}{2}-\frac{d_1}{2}-c_1\right|\\
\end{align*}
 
If the expression inside the absolute value is positive, it follows that 
 \begin{align*} 
 LS_{U_1}  &= n_1-\frac{2a+b+1}{2}-\frac{b_1-1}{2}-\frac{d_1}{2}-c_1\\ &\leq n_1 
\end{align*}

If the expression inside the absolute value is negative, it follows that
\begin{align*}
LS_{U_1} &=\frac{2a+b+1}{2}+\frac{b_1-1}{2}+\frac{d_1}{2}+c_1 - n_1 \\
&\leq \frac{2a+b+b+d+2c}{2}-n_1\\
&\leq a+b+c+d-n_1\\
&\leq n_1+n_2-n_1\\
&= n_2 
\end{align*}

Then similarly for $U_2$,
\begin{align*}
 LS_{U_2} &= \bigg|\frac{a+b+c+a+b+c+d}{2}-\frac{(n_2+2)n_2}{2}+\frac{n_2(n_2+1)}{2}\\
     & \qquad-\frac{b_2-1}{2}-\frac{d_2}{2}-c_2\bigg|\\
 &=\left|\frac{2a+2b+2c+d}{2}-n_2-1-\frac{b_2-1}{2}-\frac{d_2}{2}-c_2\right|
 \end{align*}
 
If the expression inside the absolute value is positive, it follows that 
\begin{align*}
   LS_{U_2} &= \frac{2a+2b+2c+d}{2}-n_2-1-\frac{b_2-1}{2}-\frac{d_2}{2}-d_2\\
   &\leq \frac{2(n_1+n_2)}{2}-n_2\\
   &\leq n_1 
\end{align*} 

If the expression inside the absolute value is positive, it follows that 
\begin{align*}
LS_{U_2} &= n_2+1+\frac{b_2-1}{2}+\frac{d_2}{2}+c_2 -\frac{2a+2b+2c+d}{2}\\
&\leq n_2+1+\frac{b-1}{2}+\frac{d}{2}+c -\frac{2a+2b+2c+d}{2} \\
&=  n_2+1-\frac{2a+b+c+1}{2} \\
&\leq n_2+1-1\\
&\leq n_2
\end{align*}

The proof is identical for the movement of a row in group 2 and yields a bound of the opposite row's group size in each case. The proof is symmetric for the movement to a lower row and yields the same bound as the movement to a higher row case.

Thus, the overall sensitivity is bounded by

\begin{equation*}
    LS_{U} = \max \{ n_1, n_2 \}
\end{equation*}

\end{proof}

\subsection{Privacy of $\privMW$}\label{sec:mw_pf}

\begin{theorem} \label{thm:dpprrof}
$\privMW(\data)$ is $(\epsilon_{m} + \epsilon_U$, $\delta$)-differentially private.
\end{theorem}

\begin{proof}

We first divide $\privMW$ into two steps.  First a function $f$ outputs $\widetilde{m}$.  Then, for a given value of $\tilde{m}$, a function $g$ outputs $\widetilde{U}$.  We let $g_w$ be the function $g$ when we have fixed $\widetilde{m} = w$.  We consider an arbitrary output set $S \in \mathbb{R}^2$ and let $B$ be the set of outputs where $n-m^*$ was not a correct bound on the local sensitivity of $U$.  That is, $B=\{(w,z)\mid w > m+c \}$.

First we show that output in $B$ (i.e., a bad bound for $m$) occurs only with probability $\delta$.

\begin{lemma} \label{lma:cdelta}
$\Pr[\privMW(\data) \in B] \leq \delta$ .
\end{lemma}

\begin{proof}
First, recall that by definition $\privMW(\data) \in B$ exactly when $f(\data) > m+c$.  All that remains is a simple calculation using the cumulative density function of the Laplace distribution.

Let $F(t)$ be the cumulative density function of the Laplace distribution. Then,
\begin{align*}
   \Pr\Big [(\mathcal{M}(\data),\mathcal{N}_{\mathcal{M}(\data)}(\data)) \in B] &= \Pr[\;|\widetilde{m}-m|>c]\\ &= \Pr\Big[\;|\widetilde{m}-m|>-\frac{\ln (2\delta)}{\epsilon_{m}}\Big] \\
    &=  \Pr\Big[\; \Big|\lap \Big( \frac{1}{\epsilon_{m}} \Big)\Big|>-\frac{\ln (2\delta)}{\epsilon_{m}}\Big]\\
    &= 1 - F \left(\frac{\ln(2\delta)}{\epsilon_m} \right) \\
    &= 1 - \left( 1 - \frac{1}{2} e^{\frac{\ln(2\delta)}{\epsilon_m} \cdot \epsilon_m} \right) \\
    &= \frac{1}{2} e^{\ln(2\delta)} \\
    &= \delta
\end{align*}
\end{proof}

As a result of Lemma \ref{lma:cdelta}, we can restrict ourselves to $S\setminus B$.  Conditioned on the output being in $S \setminus B$, we know that the noise added to $U$ is sufficient to maintain privacy.  Below we write $\Pr[f(\data) = w]$ to denote the probability density function of $f(\data$) at $w$ (and analogously for $g$).

\begin{align}
    \Pr[\privMW(\data) \in S] &\leq \Pr[\privMW(\data) \in S \setminus B] + \Pr[\privMW(\data) \in B] \nonumber\\ 
    &= \Pr[\privMW(\data) \in S \setminus B] + \delta \nonumber \\
    &= \int_{S \setminus B} \Pr[f(\data) = w] \cdot \Pr[g_w(\data) = z] \, dw\, dz + \delta \nonumber \\
    &\leq \int_{S \setminus B} e^{\epsilon_m}\Pr[f(\data') = w] \nonumber\\
    & \qquad \cdot e^{\epsilon_U}\Pr[g_w(\data') = z] \, dw\, dz + \delta  \label{eq:lapsub}\\
    &= e^{\epsilon_{m} + \epsilon_U} \Pr[\privMW(\data') \in S \setminus B] + \delta \nonumber \\
    &\leq e^{\epsilon_{m} + \epsilon_U} \Pr[\privMW(\data') \in S] + \delta \nonumber
\end{align}

Equation \ref{eq:lapsub} isn't technically an application of Theorem \ref{thm:lm}, but it follows from an almost identical argument.  (The values of $f$ and $g_w$ before the addition of Laplace noise differ by a bounded amount and the noise added is scaled so that the bound on the ratio follows immediately from the formula for the Laplace distribution.)

\end{proof}

\subsection{Normal approximation for \privMWP}\label{sec:mw_t1}

Here we prove the acceptability of a normal approximation that increases the efficiency of computing the reference distribution for \privMWP.

\begin{theorem}\label{thm:mw_t1}
Using ${\sf Normal}(\frac{m^*(n-m^*)}{2},\frac{m^*(n-m^*)(n+1)}{12}) + \lap \Big( \frac{n-m^*}{\epsilon_{U}} \Big)$ to generate samples $U_k$ from the reference distribution will only make the type I error smaller compared to a full simulation.
\end{theorem}

\begin{proof}
There are two parts that need to be proved in order to prove the theorem, and we consider them separately.
 
Initially, we need to show that it is acceptable for the reference distribution to be computed based on the assumption that there are no ties in the database.

When $n$ becomes very large, $U$ approximately follows a normal distribution with $\mu_U=\frac{n_1 n_2}{2}$ and $\sigma_U=\sqrt{\frac{n_1 n_2}{12}\Big((n+1)-\sum_{i=1}^{k}\frac{t_i^3-t_i}{n(n-1)}\Big)}$, where $n_1$ and $n_2$ are the group sizes, $k$ is the number of distinct ranks that are tied, and $t_i$ is the number of rows sharing rank $i$ \cite{MW}.

However, in the private setting, both $k$ and $t_i$ are private information and cannot be used to generate the null distribution. Since the standard deviation when there are no ties, $\sigma'_U=\sqrt{\frac{n_1 n_2}{12}(n+1)}$, is greater than $\sigma_U$, $$\Pr[Z<z'^*]<\Pr[Z<z^*]=\Pr[Z<z'^*],$$ where $Z,\ Z'$ are normal random variables when there are ties and no ties, respectively, $z^*, \ z'^*$ are critical values for the distributions of $Z$ and $Z'$, respectively.

Thus, according to the relation above, we can use the normal approximation as if there were no ties in the null distribution. 

The second part is that, when we generate the null distribution, we need $m$ to be used for the normal distribution. However, in the differential privacy setting, $m$ indicates one of the group sizes and is private information, so it cannot be used directly to generate the private null distribution. Instead, we use $m^*$ to replace $m$, but we need to ensure that this substitution does not increase the type I error rate.

By previous discussion, we know that with $1-\delta$ probability, $m^*\leq m$. Then compare  Norm$(\frac{m^*(n-m^*)}{2},\frac{m^*(n-m^*)(n+1)}{12})$ and \\Norm$(\frac{m(n-m)}{2},\frac{m(n-m)(n+1)}{12})$.  As $m^*(n-m^*) < m(n-m)$, \\Norm$(\frac{m^*(n-m^*)}{2},\frac{m^*(n-m^*)(n+1)}{12})$ has smaller expected value and smaller variance, and therefore it has smaller critical values. 

Thus, $$\Pr[Z<z^{**}]<\Pr[Z<z^{**}]=\Pr[Z^*<z^{**}]=\alpha,$$ where $Z,\ Z^*$ are random variables of the Normal Laplace distribution generated by $m$ and $m^*$, respectively, and $z^{**}$ is the critical value for the distribution $Z^*$, respectively.
 
Therefore, it is acceptable to use $m^*$ to generate the reference distribution (with $1-\delta$ probability).

In conclusion, we can generate the reference distribution using $${\sf Normal}(\frac{m^*(n-m^*)}{2},\frac{m^*(n-m^*)(n+1)}{12}) + \lap \Big( \frac{n-m^*}{\epsilon_{U}} \Big).$$ Doing so causes the type I error rate to decrease, which is acceptable, even if it means a decrease in power.

\end{proof}

\subsection{Proof of Theorem \ref{thm:wcalg}}\label{sec:wc_pf}

\begin{theorem*} 
Algorithm \privWP  is $\epsilon$-differentially private. 
\end{theorem*}

\begin{proof}

It is sufficient to show that the sensitivity of \WP is bounded by $2n$.  The privacy of the algorithm then follows directly from Theorem \ref{thm:lm}.

\begin{figure} [!htb]
    \centering
    \includegraphics[width=\linewidth]{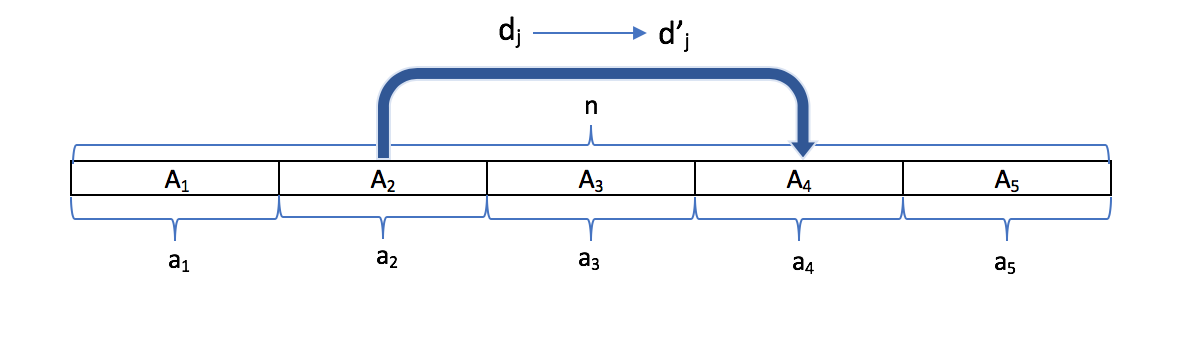}
    \caption{ $\data$ and $\data\mathbf{'}$ differ only in that $d_j$ is "moved" from region $A_2$ to region $A_4$.} \label{fig:Ggs}
\end{figure}

Let $\data$ and $\data\mathbf{'}$ be neighbouring databases. Let row $j$ be the row that differs between $\data$ and $\data\mathbf{'}$. Let $d_j$ and $d'_j$ indicate the difference between the data values in row $j$ in $\data$ and $\data\mathbf{'}$, respectively. Similarly, let $r_j$ and $r'_j$ indicate the rank of row $j$ in $\data$ and $\data\mathbf{'}$.  We assume without loss of generality that $d'_j > d_j$.

We divide the rows of $\data$ (or $\data\mathbf{'}$) into five regions, $A_1$ through $A_5$.  $A_1$ is the set of rows $i$ with $d_i < d_j$.  $A_2$ is the set with $d_i=d_j$.  $A_3$ has $d_i$ values between $d_j$ and $d'_j$.  $A_4$ has $d_i$ values equal to $d'_j$.  $A_5$ has values greater than $d'_j$ (though we don't need to use $A_5$ in the proof).  For convenience, we use $a_i$ to represent the size of the set $A_i$ in $\data$.

As shown in Figure \ref{fig:Ggs}, $\data$ and $\data\mathbf{'}$ differ only in that $d_j$ is ``moved'' from region $A_2$ to region $A_4$. Let $w$ and $w'$ be test statistics produced from databases $\data$ and $\data\mathbf{'}$, respectively.

Recall that $w = \sum_i s_i r_i$.  Let $w_i = s_i r_i$ be the contribution of row $i$ to this sum.  We then want to bound $\Delta w= |w-w'|$. We similarly set $\Delta w_i$ to be $|s_i r_i - s'_i r'_i|$.  Note that for $i$ in $A_1$ or $A_5$ we have $\Delta w_i = 0$.

Next, we consider $\Delta w_j$ the change in the term corresponding to the row where the data changes. As row $j$ is tied with all other rows in region $A_2$, $r_j = \frac{(a_1+1)+(a_1+a_2)}{2} = \frac{2a_1+a_2+1}{2}$ and \\$r'_j = \frac{(a_1+a_2+a_3)+(a_1+a_2+a_3+a_4)}{2} = \frac{2a_1+2a_2+2a_3+a_4}{2}$. In the worst case, $s_j$ changes when the move occurs, so 
\begin{align*}
\Delta w_j < r_j + r'_j &= \frac{2a_1+a_2+1}{2}+\frac{2a_1+2a_2+2a_3+a_4}{2} \\
&= \frac{4a_1+3a_2+2a_3+a_4+1}{2}.
\end{align*}

Now we consider $\Delta w_i$ where $i \in A_2$ (excluding row $j$). As there is one fewer row tied, each of their ranks (excluding row $j$) is changed from $\frac{2a_1+a_2+1}{2}$ to $\frac{2a_1+a_2}{2}$, which means that each row is changed by $\frac{1}{2}$. As there are $a_2-1$ rows in this region, the contribution to $\Delta w$ from this set of rows is bounded by $\frac{a_2-1}{2}$.

Similarly, each of the row ranks (excluding row $j$) in region $A_4$ is changed from $\frac{2a_1+2a_2+2a_3+a_4+1}{2}$ to $\frac{2a_1+2a_2+2a_3+a_4}{2}$ as one more row is tied with rows in $A_4$. As there are $a_4$ rows, $\Delta w$ from the rows in region $A_4$ is bounded by $\frac{a_4}{2}$. 

Now we consider region $A_3$. As $d_j$ has moved from before this region to after this region, each of the row ranks is changed by $1$. Therefore the contribution to $\Delta w$ from the rows in region $A_3$ is bounded by $a_3$.

It follows that 
\begin{align*}
    \Delta w &= |w-w'|\\
    &\leq \frac{4a_1+3a_2+2a_3+a_4+1}{2} + \frac{a_2-1}{2} + \frac{a_4}{2} + a_3\\
    &= 2a_1+2a_2+2a_3+a_4\\
    &\leq 2(a_1+a_2+a_3+a_4)\\
    &\leq 2n
\end{align*}

Having bounded the sensitivity, the proof is complete.

\end{proof}

\subsection{Proof of Theorems \ref{thm:mean_sensitivity} and Theorem \ref{thm:s_squ_sensitivity}}\label{sec:t_test_pf}

\begin{theorem*} 
The sensitivity of $\bar{x}$ is $\frac{2}{n}$ .
\end{theorem*}

\begin{proof}
Consider two neighboring databases, $D$ and $D'$, of size $n$. Without loss of generality suppose that the observation that differs is $x_n$. Then,
\begin{align*}
    \Delta \bar{x} &= \left | \bar{x}' - \bar{x} \right | \\
    &= \left | \frac{x_1 + \ldots + x_n'}{n} - \frac{x_1 + \ldots + x_n}{n} \right | \\
    &= \left | \frac{x_n' - x_n}{n} \right | \\
    &\leq \frac{2}{n} 
\end{align*}
\end{proof}

\begin{theorem*}
The sensitivity of $s^2$ is $\frac{5}{n-1}$.
\end{theorem*}

\begin{proof}
Consider two neighboring databases, $D$ and $D'$, of size $n$. Without loss of generality suppose that the observation that differs is $x_n$. We first simplify $s^2$ for sensitivity analysis.

\begin{align*}
s^2 &= \frac{1}{n-1} \sum_{k=1}^n (x_k - \bar{x})^2\\
&= \frac{1}{n-1} \left(\sum_{k=1}^n x_k^2 -2\bar{x}\sum_{k=1}^n x_k+\sum_{k=1}^n \bar{x}^2\right)\\
&= \frac{1}{n-1} \left(\sum_{k=1}^n x_k^2 -2n\bar{x}^2+n\bar{x}^2\right)\\
&= \frac{1}{n-1} \left(\sum_{k=1}^n x_k^2 -n\bar{x}^2\right)\\
\end{align*}

Then we consider how much $s^2$ is differed with one row $x_n$ differed.
\begin{align*}
    \Delta s^2 &= \left | (s^2)' - s^2 \right | \\
    &= \left | \frac{1}{n-1} \left [ \sum_{i = 1}^{n} (x_i')^2 - n (\bar{x}')^2  \right ] - \frac{1}{n-1} \left [ \sum_{i = 1}^{n} x_i^2 - n \bar{x}^2  \right ] \right | \\
    &= \frac{1}{n-1} \left |   \sum_{i = 1}^{n} (x_i')^2 -  \sum_{i = 1}^{n} x_i^2 - n (\bar{x}')^2  + n \bar{x}^2  \right | \\
    &= \frac{1}{n-1} \left |  (x_n')^2 - (x_n)^2 + n \left(\bar{x}^2  - (\bar{x}')^2 \right) \right | 
\end{align*}

By the triangle inequality, 
\begin{align*}
    \Delta s^2 &\leq \frac{1}{n-1} \left(\left |  (x_n')^2 - (x_n)^2 \right|+\left| n \left(\bar{x}^2  - (\bar{x}')^2 \right) \right | \right)\\
    &\leq \frac{1}{n-1} \left ( 1+ n|\bar{x}+\bar{x'}||\bar{x}-\bar{x'}| \right ) \\
\end{align*}

By Theorem \ref{thm:mean_sensitivity}, $|\bar{x}-\bar{x'}| \leq \frac{2}{n}$. It follows that 
\begin{align*}
    \Delta s^2 &\leq \frac{1}{n-1} \left ( 1+ 2|\bar{x}+\bar{x'}| \right ) \\
    &\leq \frac{1}{n-1}(1+4)\\
    &= \frac{5}{n-1}
\end{align*}

\end{proof}

\clearpage

\section{Supplementary Results}\label{sec:supp_rslt}

\subsection{Many Groups}\label{sec:kw_appendix}
\paragraph{Application to Real-World Data} The Global Findex database maintains survey data from a large number of countries and years about financial inclusion around the world \cite{financial_inclusion_data}. We will use our differentially private Kruskal Wallis tests to examine whether income quintiles statistically significantly differ by age in the United States.

In Figure \ref{fig:kw_app}, we sample a large number of subsets of several different sizes in order to quantify how likely our differentially private variant of the Kruskal-Wallis would be to detect the effect in this dataset as compared to the public version. At an $\epsilon$ of $1$, the absolute value test variant follows extremely closely behind the public version. Even at an $\epsilon$ an order of magnitude smaller, the absolute value test outperforms the standard version. Because these values are randomly sampled, we can also learn from this figure that the test is robust to unequal sample sizes.

\begin{figure} [!htb] 
    \centering
    \includegraphics[width=\linewidth, keepaspectratio]{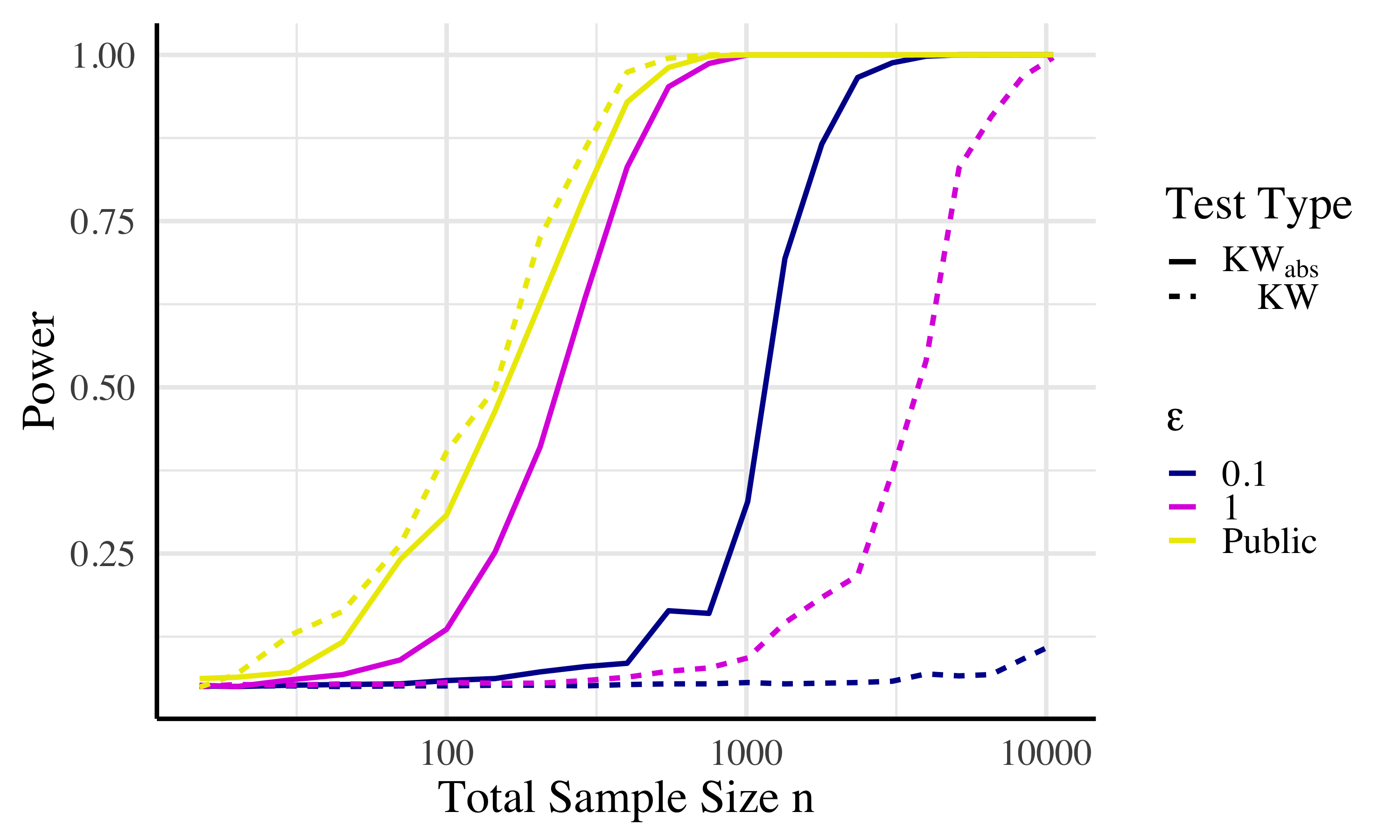}
    \caption{Power of \privKWPA at various values of $\epsilon$ and $n$ on subsets of the Global Financial Inclusion Data. ($g = 5$, $\alpha=.05$, unequal group sizes)}\label{fig:kw_app}
\end{figure}
\paragraph{Dealing with ties} We now examine the effect of our tiebreaking procedure (randomly ordering all tied values) on the power of the test. We subtract the Gini coefficient \cite{Gini}, a measure of the statistical spread of data, from 1, in order to quantify the proportion of ties per trial of the test procedure.

\begin{figure} [!htb] 
    \centering
    \includegraphics[width=\linewidth, keepaspectratio]{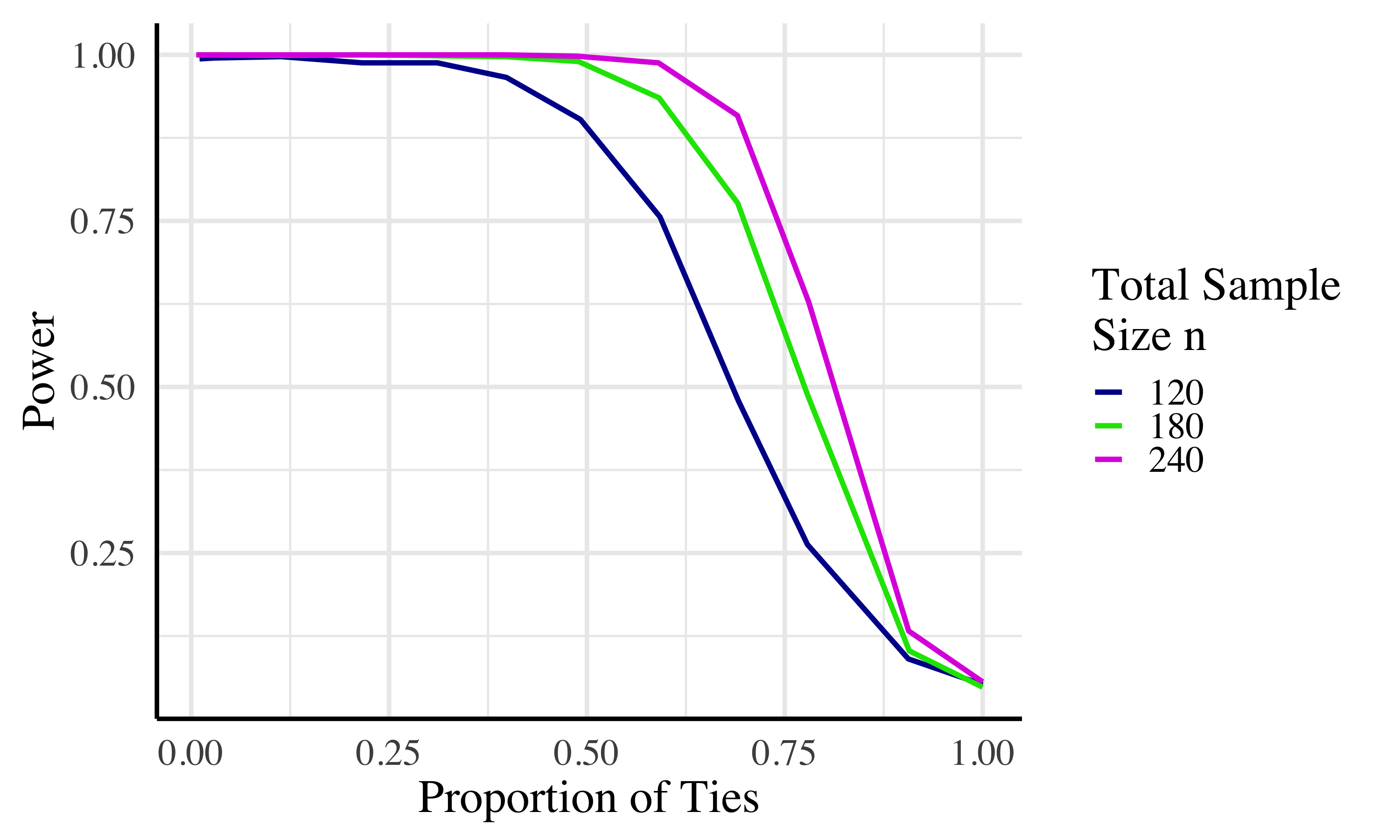}
    \caption{Power of \privKWPA at various proportions of tied values and $n$. (Effect size: $max(\mu_{n}) - min(\mu_{n}) = 2 \sigma$; $g = 3$; $\epsilon = 1$; $\alpha=.05$; equal group sizes; normally distributed sample data)}\label{fig:with_ties}
\end{figure}

\paragraph{Varying effect size} In Figure \ref{fig:varying_effect_size_kw}, we examine the use case of large sample sizes attempting to detect small effect sizes, as is increasingly common in the scientific literature \cite{bland2009tyranny}. We find that, at large sample sizes, the test at $\epsilon = 1$ needs only marginally more data than the public test to achieve the same power.

\begin{figure} [!htb] 
    \centering
    \includegraphics[width=\linewidth, keepaspectratio]{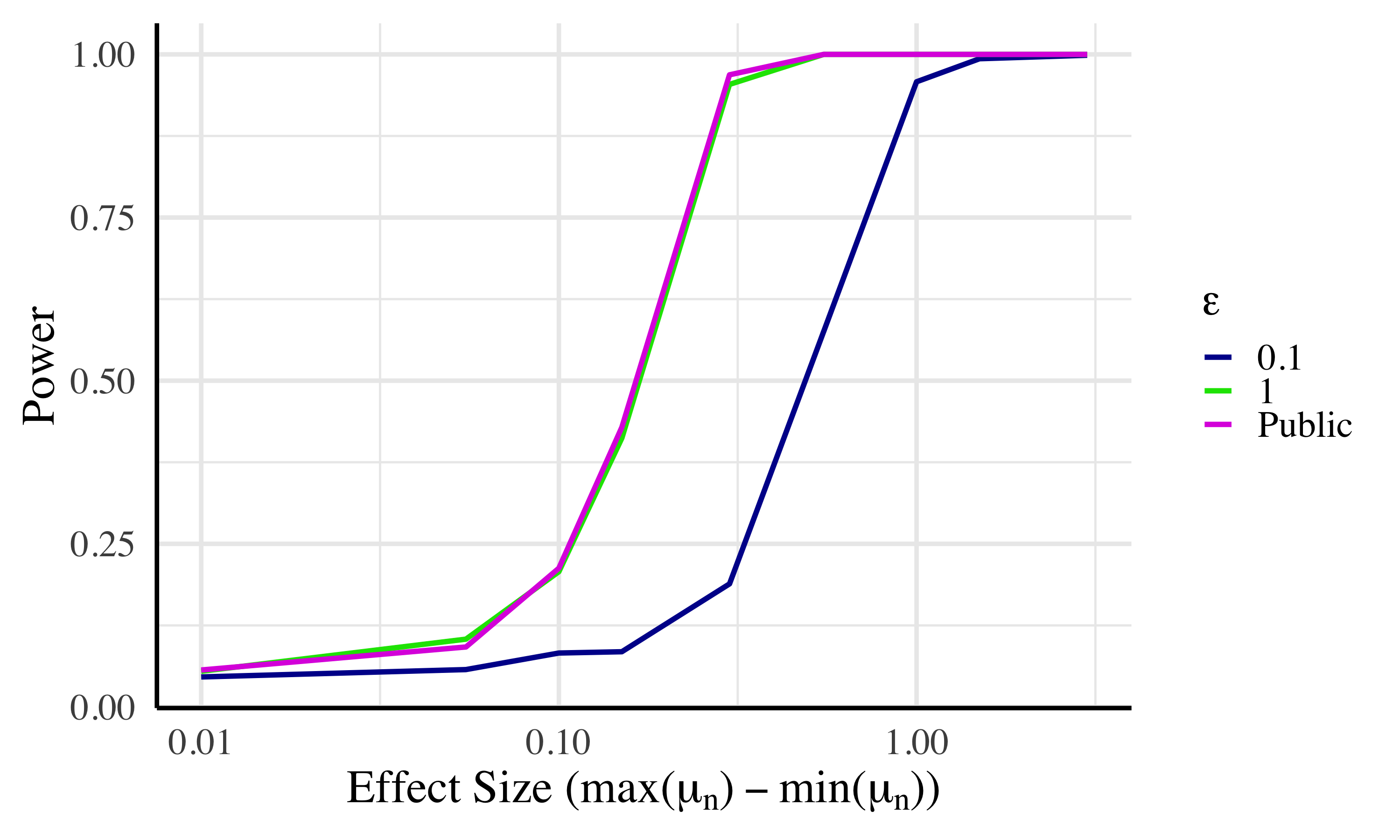}
    \caption{Power of \privKWPA at various effect sizes and $\epsilon$. ($n = 900$; $g = 3$; $\alpha=.05$; equal group sizes; normally distributed sample data)}\label{fig:varying_effect_size_kw}
\end{figure}

\paragraph{Uniformity of p-values (cont.)} Theorem \ref{thm:exp_val} shows that the expected value of the statistic is maximized when the groups are exactly equal in size. Thus, in place of the true group sizes, algorithm $\privKWPA$ assumes that the group sizes are of size $\lfloor\frac{n}{g}\rfloor$ when simulating the reference distribution (in order to maintain privacy.) One may wonder, then, if the type I error rate of our test increases when sample sizes are unequal. We thus fix the total sample size and vary the proportion of the total sample size $n$ that the maximum group sample size $n$ takes up. The results are shown in Figure \ref{fig:qqplot_uneq_kw}.

\begin{figure} [!htb] 
    \centering
    \includegraphics[width=\linewidth, keepaspectratio]{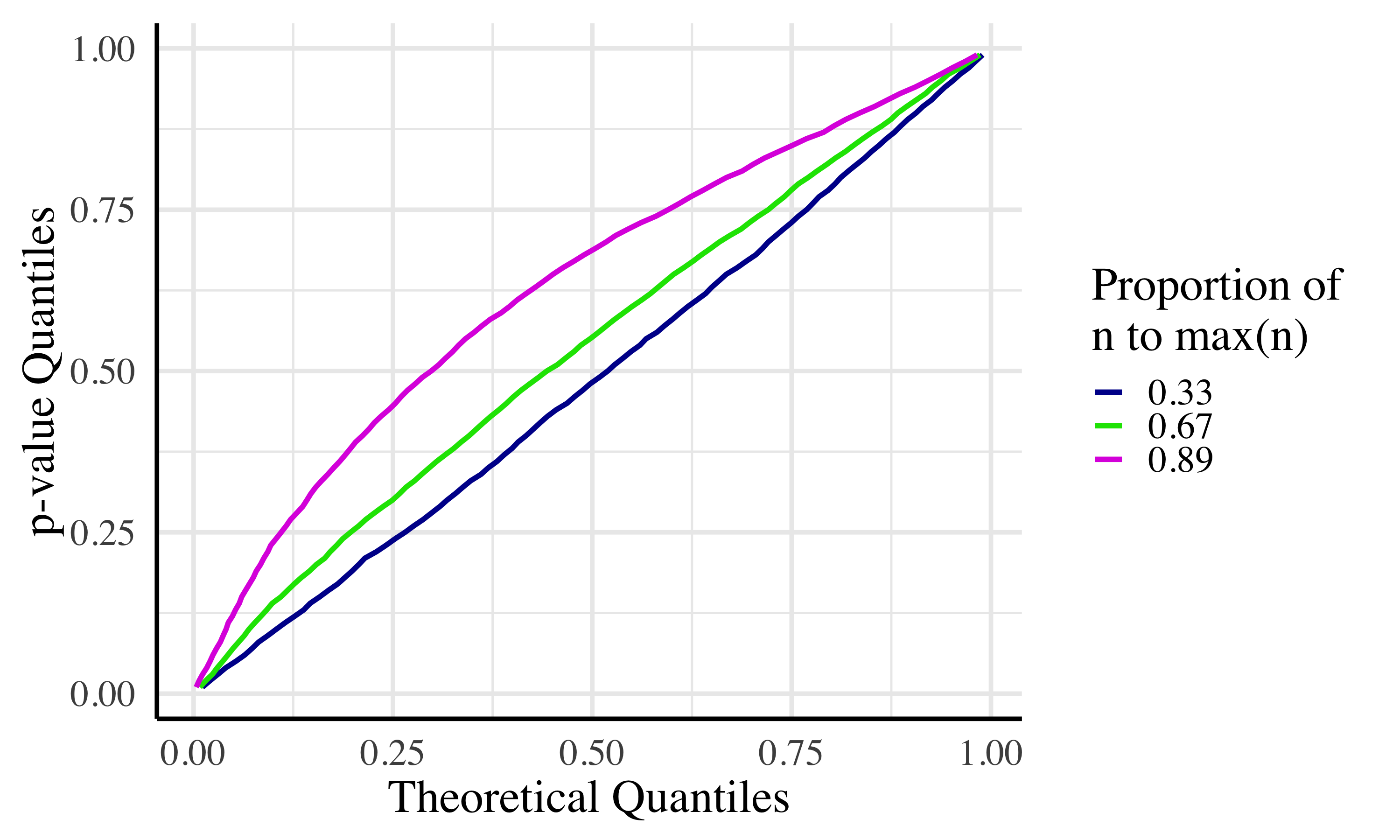}
    \caption{A quantile-quantile plot of \privKWPA comparing the distribution of simulated p-values to the uniform distribution with unequal group sizes. ($g = 3$; $\epsilon = 1$; $n = 100$)}\label{fig:qqplot_uneq_kw}
\end{figure}

As displayed in Figure \ref{fig:qqplot_uneq_kw}, the type I error rate of our test is consistently below alpha. Further, the power loss due to unequal sample sizes is minimal, even with drastically different group sample sizes.

\paragraph{Further Comparisons to Previous Work} We continue our comparisons to Swanberg et. al.'s test introduced in Section \ref{sec:kw} at varying parameterizations in order to better understand the relationship between their test and \privKWPA.

Due to the rank-based nature of \privKWPA, the public nonparametric test statistic (before the addition of noise) does not change once all of the ranks of one group are on one side of the distribution. Thus, we take a particular interest in sampling distributions with an effect size well beyond two standard deviations (the maximum studied in the main body) since the parametric statistic ought to perform better in such situations. We find, however, that \privKWPA persists in outperforming Swanberg et. al.'s test at the parameterizations chosen to reflect such a setting in Figure \ref{fig:anova_comp_big_effect}.

\begin{figure} [!htb] 
    \centering
    \includegraphics[width=\linewidth, keepaspectratio]{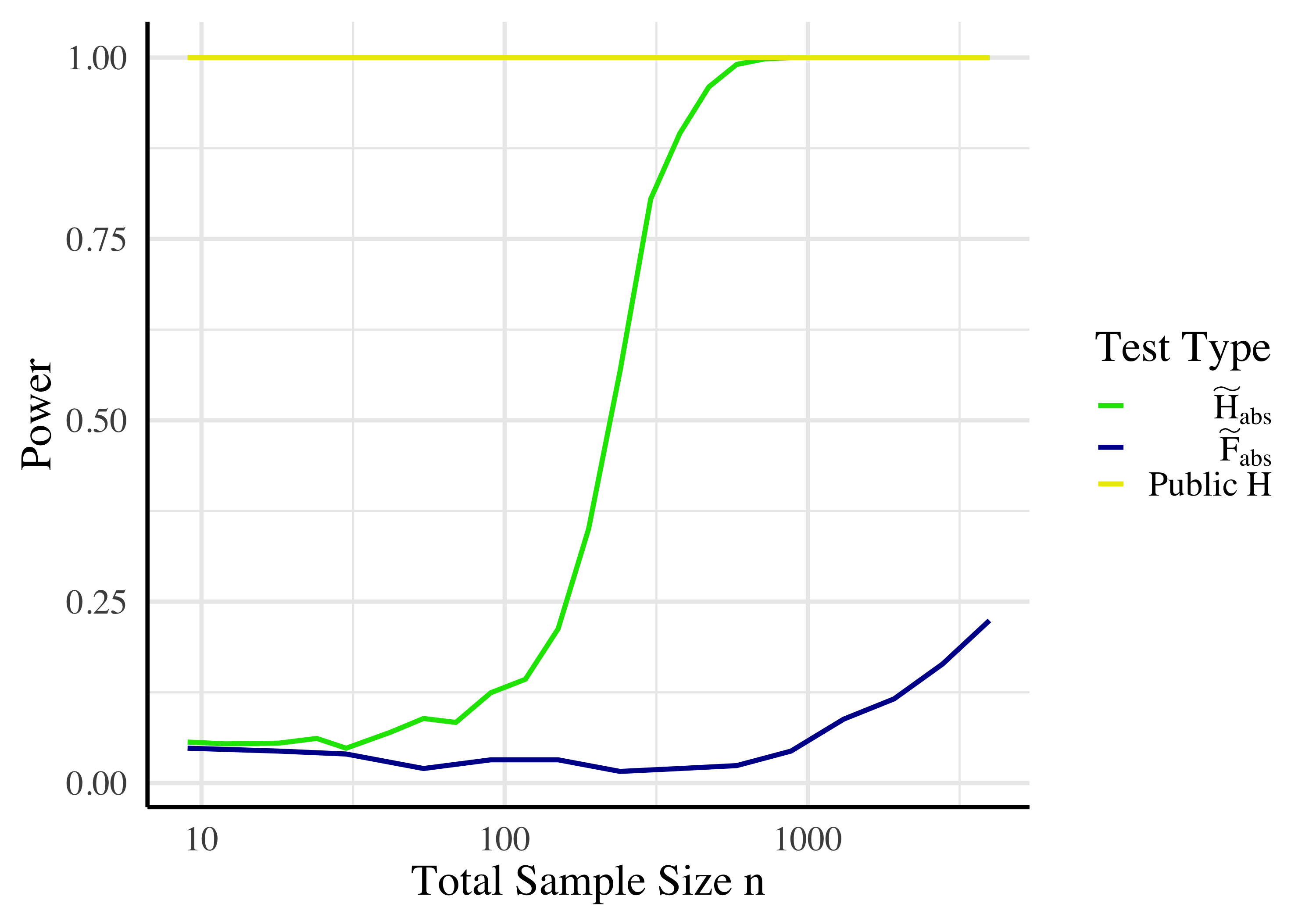}
    \caption{Power comparison of \privKWPA and Swanberg et. al.'s test at a large effect size with various $n$. ($g = 3$; $\alpha = .05$; $\epsilon = .1$; effect size: $max(\mu_{n}) - min(\mu_{n}) = 10 \sigma$; normally distributed sample data)}\label{fig:anova_comp_big_effect}
\end{figure}

We also consider settings where effect sizes are notably smaller than those studied in the main body in Figure \ref{fig:anova_comp_small_effect}.

\begin{figure} [!htb] 
    \centering
    \includegraphics[width=\linewidth, keepaspectratio]{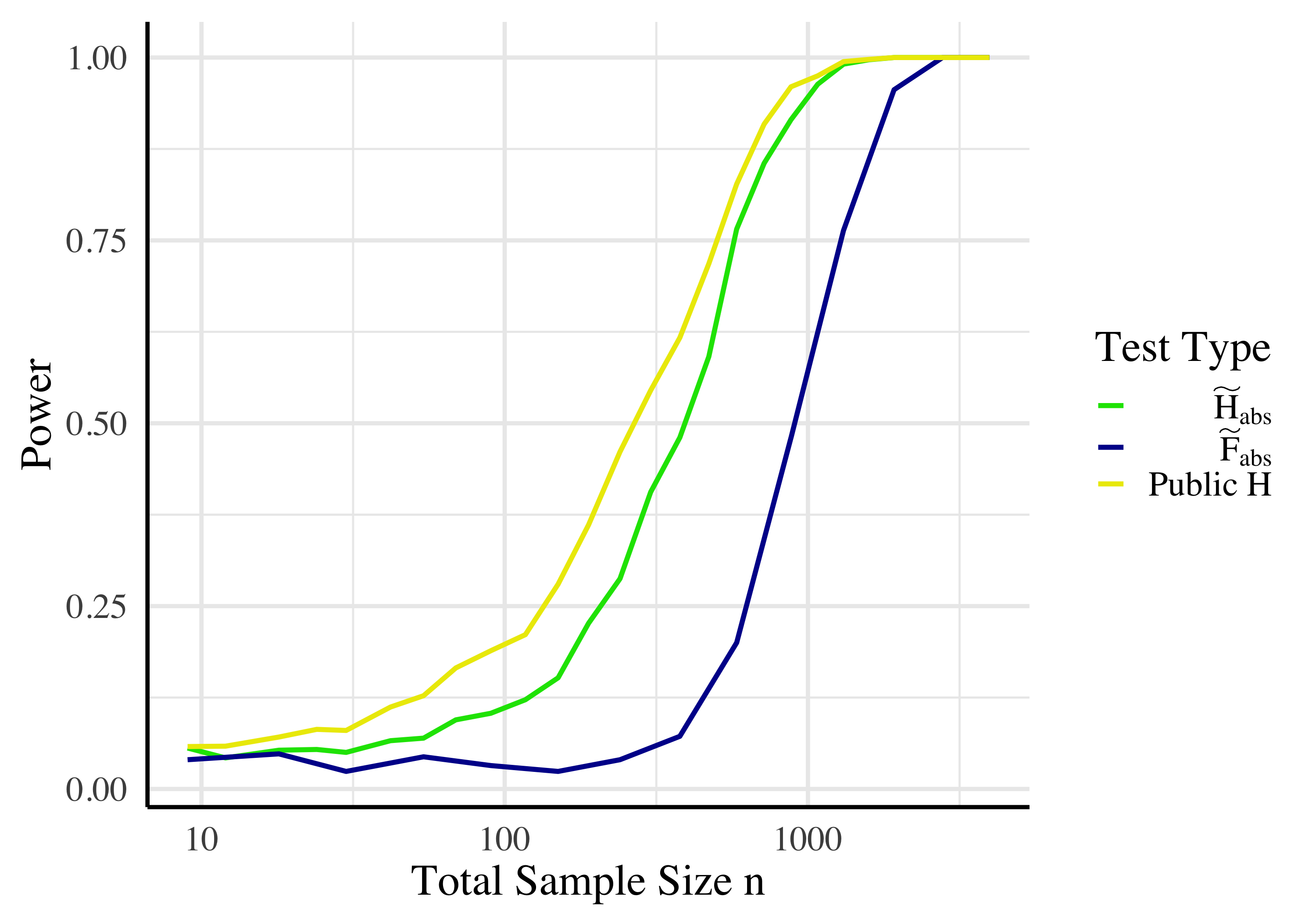}
    \caption{Power comparison of \privKWPA and Swanberg et. al.'s test at a small effect size with various $n$. ($g = 3$; $\epsilon = 1$; effect size: $max(\mu_{n}) - min(\mu_{n}) = \frac{\sigma}{3}$; normally distributed sample data)}\label{fig:anova_comp_small_effect}
\end{figure}

\clearpage

\subsection{Two Groups}\label{sec:mw_app}
\paragraph{Application to Real-World Data} In 2017, Andersen et al. conducted a study on patient's and community health volunteer's understanding of eligibility for a legal medical abortion in Nepal \cite{nepal_study}. The dataset collected for the study was released on the condition that users of the data not seek to deanonymize any of the study participants. We demonstrate that our test reaches a meaningful level of statistical power for a variety of parameters---and could have done so with smaller sample sizes than that taken in this study---and thus a differentially-private query interface would have been a sufficient tool for the release of this data while greatly decreasing the risk of reidentification.

Of the many variables in this dataset, we focus on two; the first is simply a categorical variable indicating whether or not a medical abortion was granted, and the second is the number of days since the patient's last period. One may wonder if, as the number of days since the patients last period increases, their likelihood of being granted a medical abortion changes. This question, naturally, can be answered by the Mann-Whitney U test. We have 472 cases where medical abortions were not granted (with an average number of days since last period of around 72), and 2620 cases where medical abortions were granted (with an average number of days since last period of around 49). The means of the two groups differ by slightly less than 1.5 standard deviations (where the number of days since last period is greater for subjects not granted a medical abortion).\footnote{These values come from a slightly filtered and tidied version of the original dataset---our source code can be found at: \textit{github.com/simonpcouch/non-pm-dpht}}

\begin{figure} [!htb] 
    \centering
    \includegraphics[width=\linewidth, keepaspectratio]{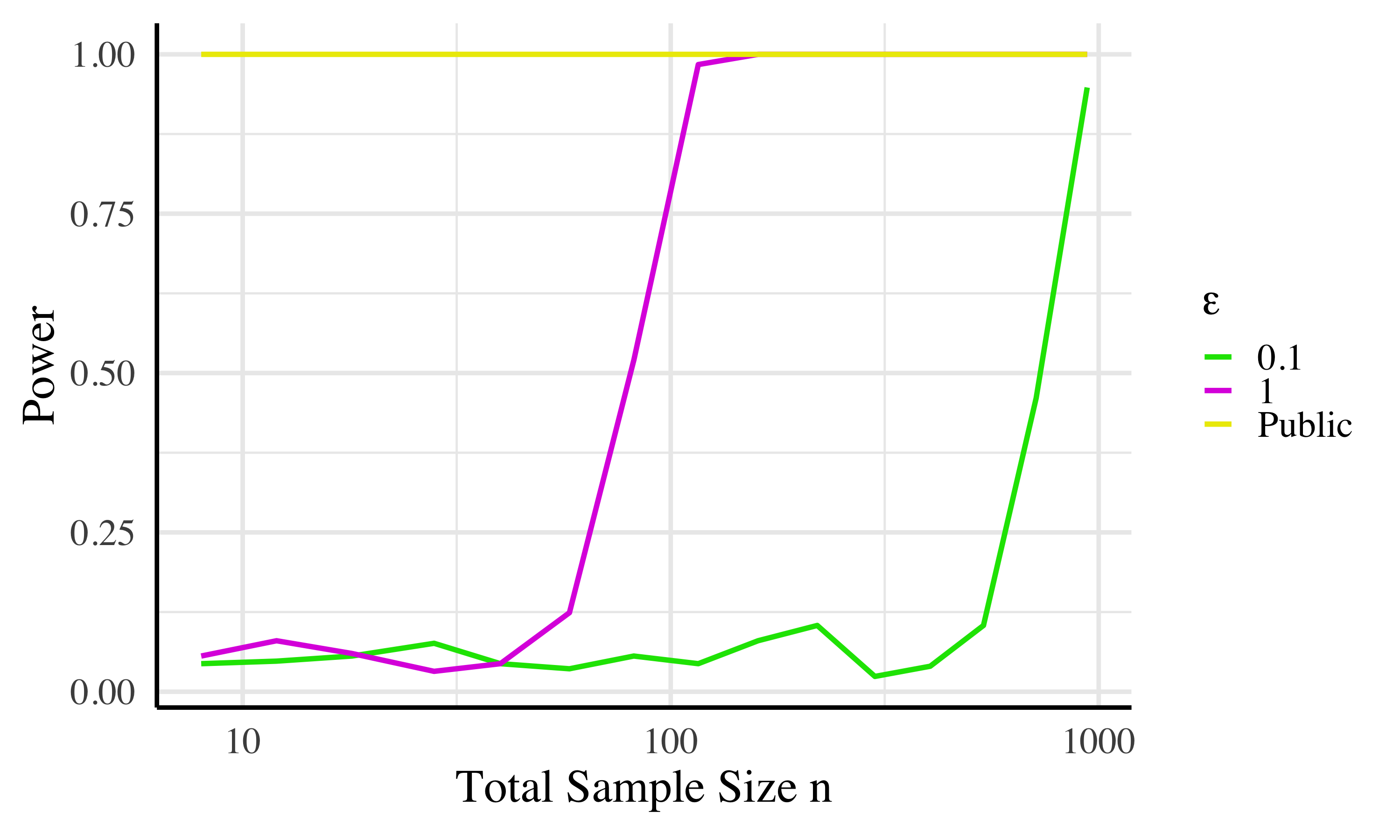}
    \caption{Power of \privMWP at various values of $\epsilon_{tot}$ and $n$. (Proportion of $\epsilon_{tot}$ to $\epsilon_{m} = .65$; $\alpha=.05$; $m$:($n-m$) $= 1$; effect size $\mu_{1} - \mu_{2} \approx 1.5 \sigma$)}\label{fig:real_world_app}
\end{figure}

In Figure \ref{fig:real_world_app} we sample a varying number of observations from the groups and evaluate the power of both our and the public test. The far right of the figure shows the power using almost all of the provided data (while still maintaining equal group sizes). For $\epsilon$ $=$ $1$, 100\% power is achieved well before all data is used, while that for $\epsilon$ $=$ $.1$ is more likely than not to detect the effect by the time all of the data is used.
\paragraph{Uniformity of p-values (cont.)} In addition to the examination of uniformity of p-values when group sizes are equal, we also examine the effect of unequal sample sizes in Figure \ref{fig:qqplot_uneq_mw} and find that as sample sizes become more unequal, our critical values are increasingly conservative.

\begin{figure} [!htb] 
    \centering
    \includegraphics[width=\linewidth, keepaspectratio]{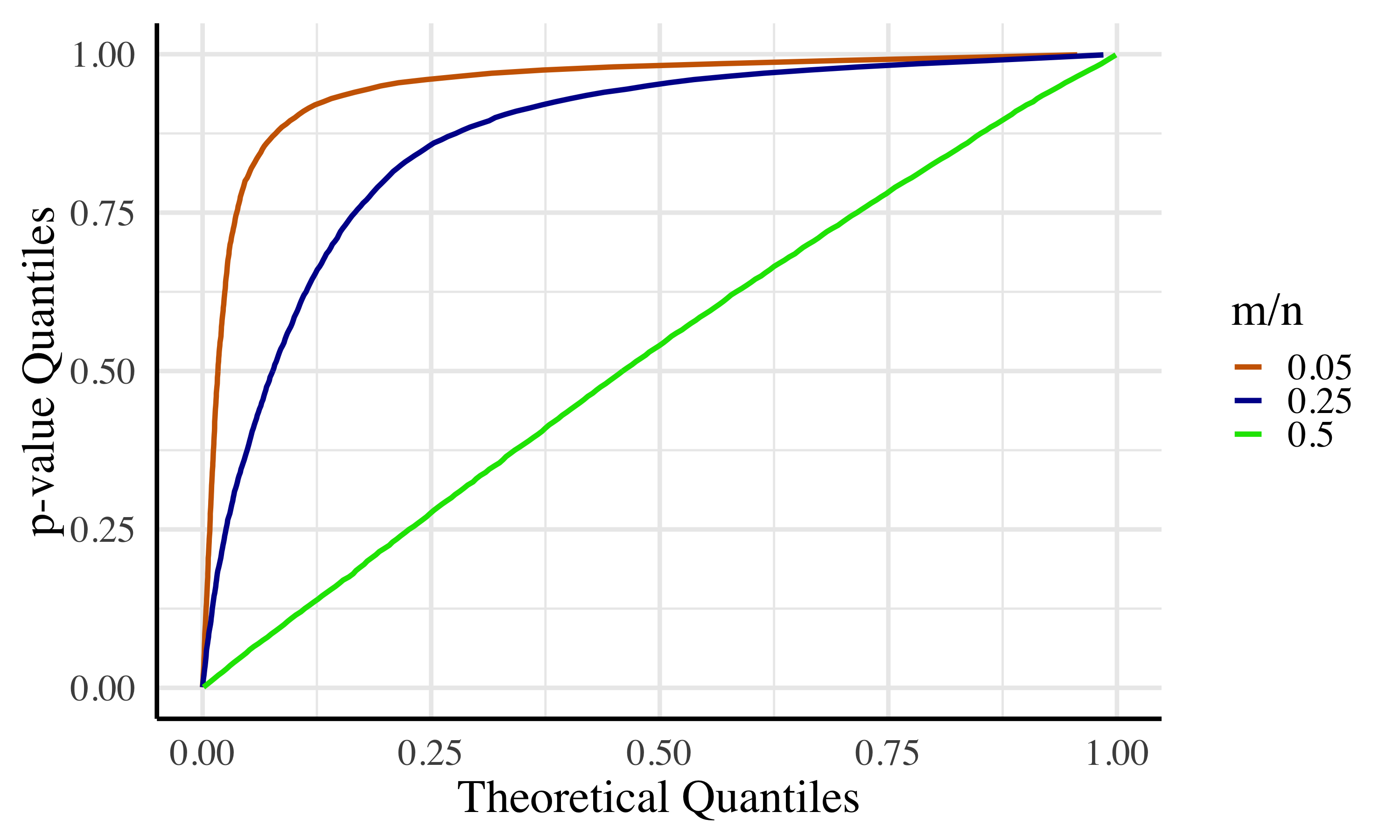}
    \caption{A quantile-quantile plot of \privMWP examining the effect of unequal group sizes. ($\epsilon_{tot} = 1$; proportion of $\epsilon_{tot}$ to $\epsilon_{m} = .65$; $n = 1000$; normally distributed sample data)}\label{fig:qqplot_uneq_mw}
\end{figure}

While our critical values do not exceed $\alpha$ in any case, they are much too conservative when total sample size is too small or group sample sizes are largely different. 

Naturally, one might wonder whether we can assume that group sizes are equal when we generate our reference distribution, as we do with our Kruskal-Wallis, so that we can allot of our privacy budget to estimating the test statistic. In Figure \ref{fig:qqplot_assume_mw}, we address this question with a quantile-quantile plot where we fix total sample size and use the same reference distribution (assuming equal group sizes) on null distributions generated with unequal group sizes. We find that, even with $n$ $=$ $1000$ and $\epsilon_{tot} = 1$, small differences in group sizes cause egregious increases in the type I error rate. We conclude, as a result, that our allotment of some portion of the privacy budget to estimating group sizes is essential.

\begin{figure} [!htb] 
    \centering
    \includegraphics[width=\linewidth, keepaspectratio]{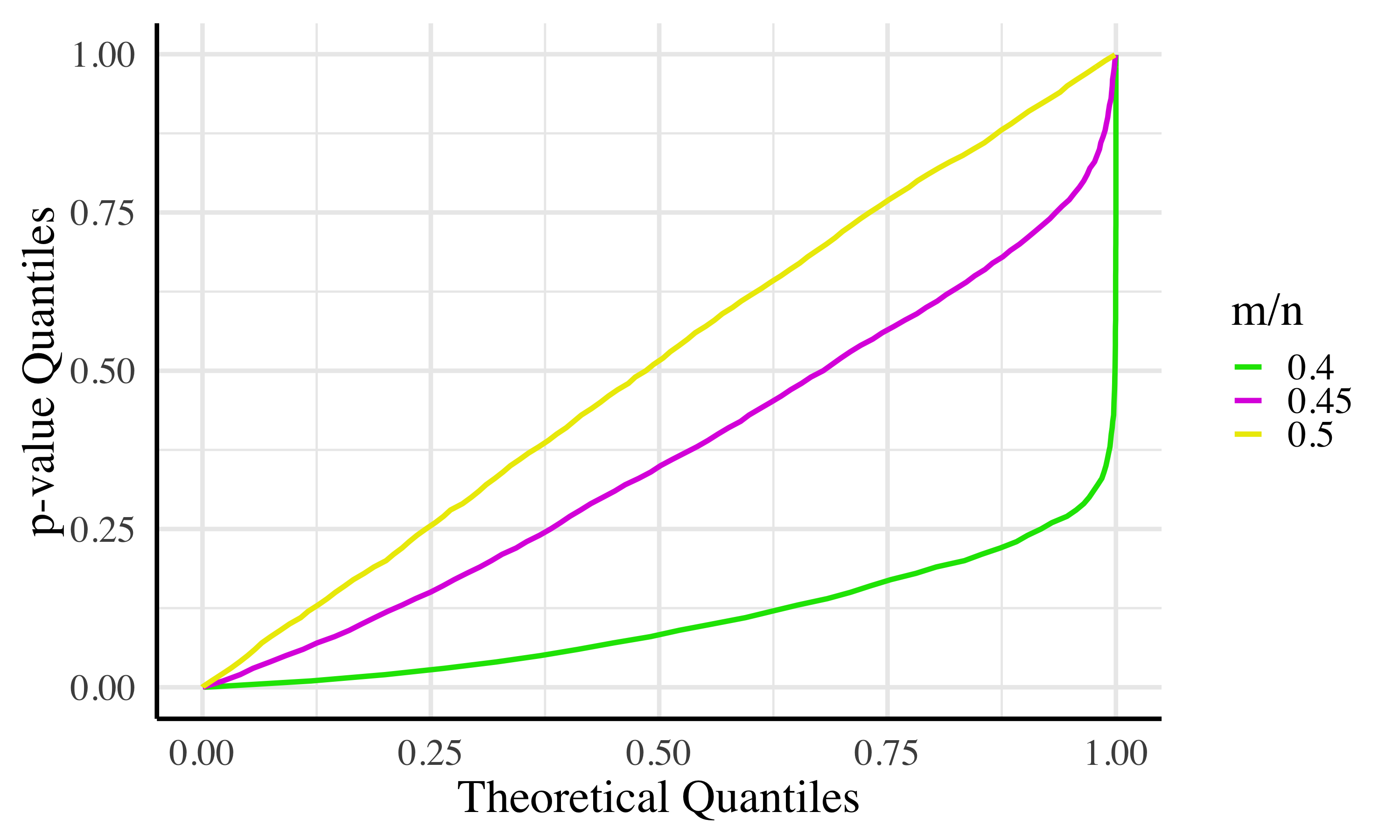}
    \caption{A quantile-quantile plot of a modification of \privMWP where total sample size is fixed but group sizes vary, using the same reference distribution (assuming equal group sizes) for all trials. ($\epsilon_{tot} = 1$; proportion of $\epsilon_{tot}$ to $\epsilon_{m} = .65$; $n = 1000$; normally distributed sample data)}\label{fig:qqplot_assume_mw}
\end{figure}

\paragraph{Varying effect size} We examine the power of our test at varying effect sizes, as shown in Figure \ref{fig:vary_effect}. Two features of this graph display notable characteristics of our test. Initially, note how minimal the difference is between the public test and our test at $\epsilon_{tot} = 1$; at large sample sizes (in this case, $n = 1500$), the power of the public test and that of our test is almost indistinguishable. Further, the $\epsilon_{tot} = .06$ line shows that, once the effect size reaches a difference in $\mu$ of roughly 3$\sigma$, the power of the test seems to level off before reaching $1$. This is due to the rank-based nature of the Mann-Whitney; once all of the ranks of one group are above or below those of the other, increasing effect size has no consequence. Thus, at certain choices of $\epsilon_{tot}$ and $n$, there is effectively an upper bound on the power test that is below $1$, regardless of effect size.

\begin{figure} [!htb] 
    \centering
    \includegraphics[width=\linewidth, keepaspectratio]{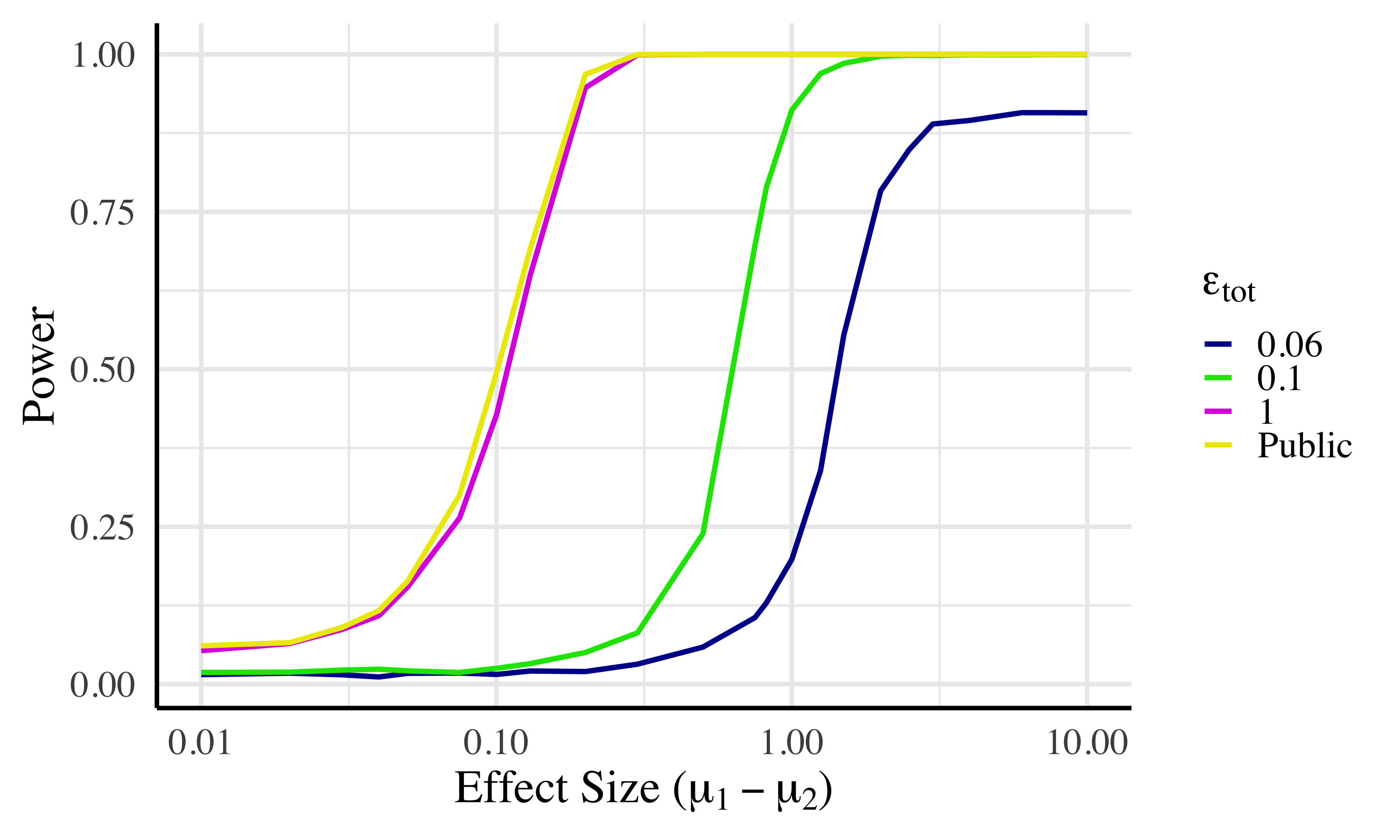}
    \caption{Power of \privMWP at various effect sizes ($\mu_{1} - \mu_{2} = x \sigma$) and privacy parameters $\epsilon$. (Proportion of $\epsilon_{tot}$ to $\epsilon_{m} = .65$; $n = 1500$; $\alpha=.05$; $m$:($n-m$) $= 1$; normally distributed sample data)}\label{fig:vary_effect}
\end{figure}

\paragraph{The effect of unequal group sizes} In Figure \ref{fig:var_n1_tot_eps}, we quantify the effect of unequal group sizes on the power of the test. In the public setting, power reduction only occurs when the sample sizes are meaningfully different. In the private setting, however, when random noise is added to $n_1$ and $n_2$ for the simulation of the reference distribution (and thus critical value computation), smaller total sample sizes often result in substantial differences in size of $m^*$ and $n - m^*$.

\begin{figure} [!htb] 
    \centering
    \includegraphics[width=\linewidth, keepaspectratio]{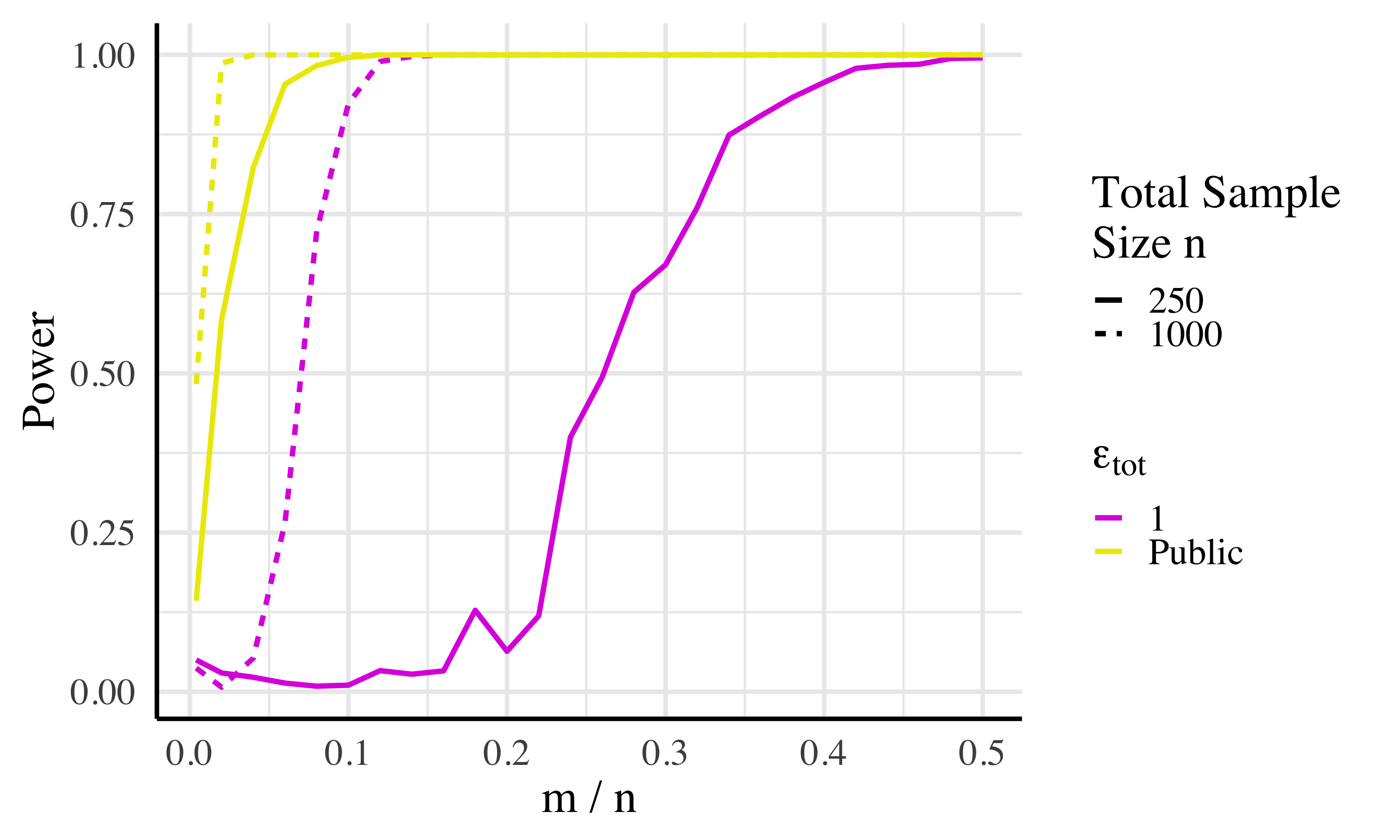}
    \caption{Power of \privMWP at various values of $\epsilon_{tot}$ and sample size $m$. (Effect size: $\mu_1 - \mu_2 = 1 \sigma$; proportion of $\epsilon_{tot}$ to $\epsilon_{m} = .65$; $\alpha=.05$; $n = 500$; normally distributed sample data)}\label{fig:var_n1_tot_eps}
\end{figure}

\paragraph{Comparative effects of unequal group sizes} We continue the discussion given in Section \ref{sec:pwr} on comparing \privMWP to \privKWPA in the 2-group use case, this time quantifying the comparative losses in power to the two tests when group sizes are unequal. The \privMWP test does not assume equal sizes, and must rather allot some of the total privacy budget to estimating one of the group sizes $m$ since the Type I error rate increases when equal group sizes are assumed when generating the reference distribution. However, as we have previously shown, it is acceptable for Type I error for \privKWPA to assume equal group sizes, though this means increases in the Type II error rate (i.e. losses in power) when group sizes are highly unequal. Thus, it is natural to ask whether the \privMWP test becomes more powerful when group sizes are very different since some of the privacy budget is allotted to estimating $m$, and thus, the reference distribution is more precisely calibrated to reduce Type II error rate. Our findings in this setting are summarized in Figures \ref{fig:kw_vs_mw_uneq_n250} and \ref{fig:kw_vs_mw_uneq_n1000}, where we find that \privKWPA persists in achieving higher power than \privMWP even under such parameterizations.

\begin{figure} [!htb] 
    \centering
    \includegraphics[width=\linewidth, keepaspectratio]{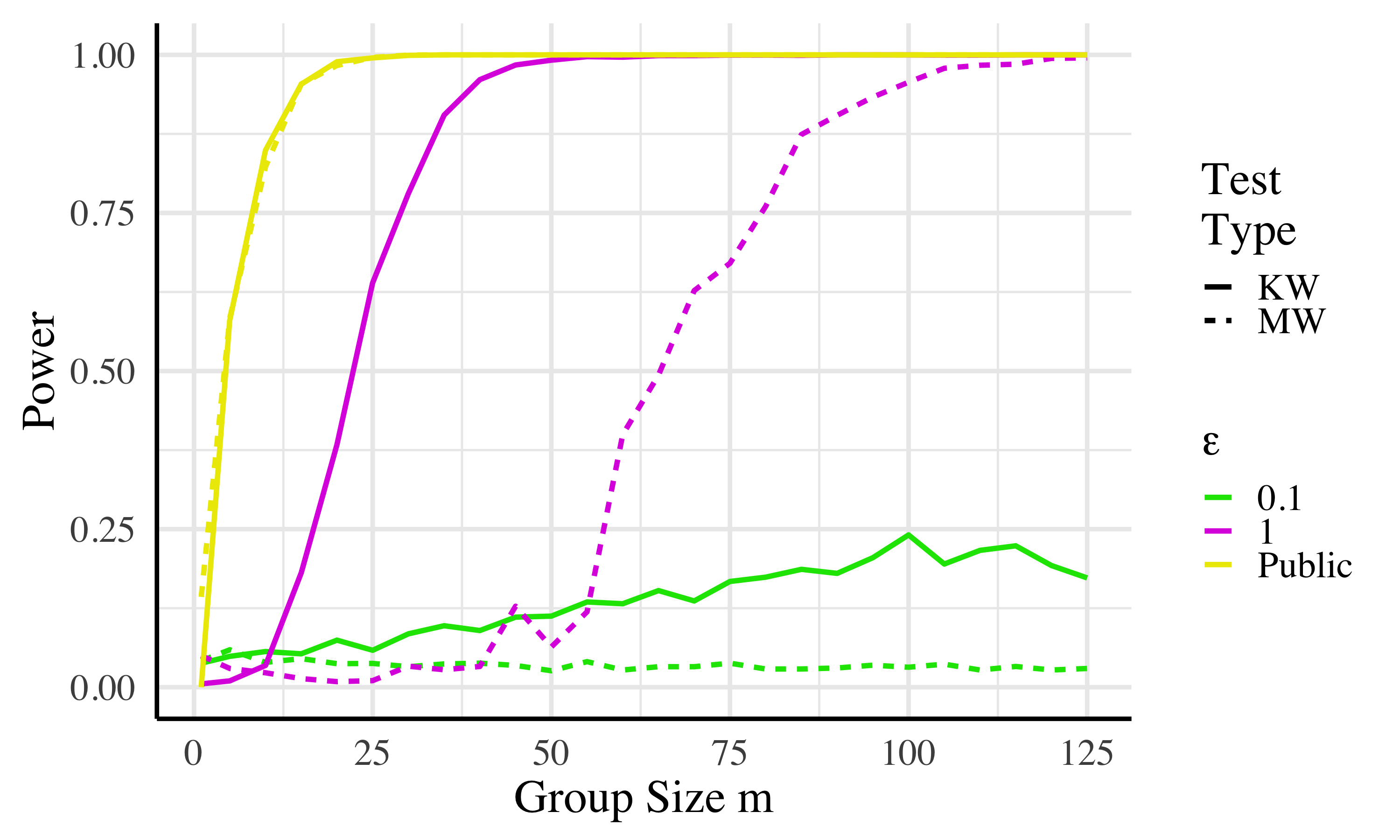}
    \caption{Power of \privMWP and \privKWPA at various values of $\epsilon_{tot}$ and group sample size $m$. (Effect size: $\mu_1 - \mu_2 = 1 \sigma$; proportion of $\epsilon_{tot}$ to $\epsilon_{m} = .65$; $\alpha=.05$; $n = 250$; $g = 2$; normally distributed sample data)}\label{fig:kw_vs_mw_uneq_n250}
\end{figure}

\begin{figure} [!htb] 
    \centering
    \includegraphics[width=\linewidth, keepaspectratio]{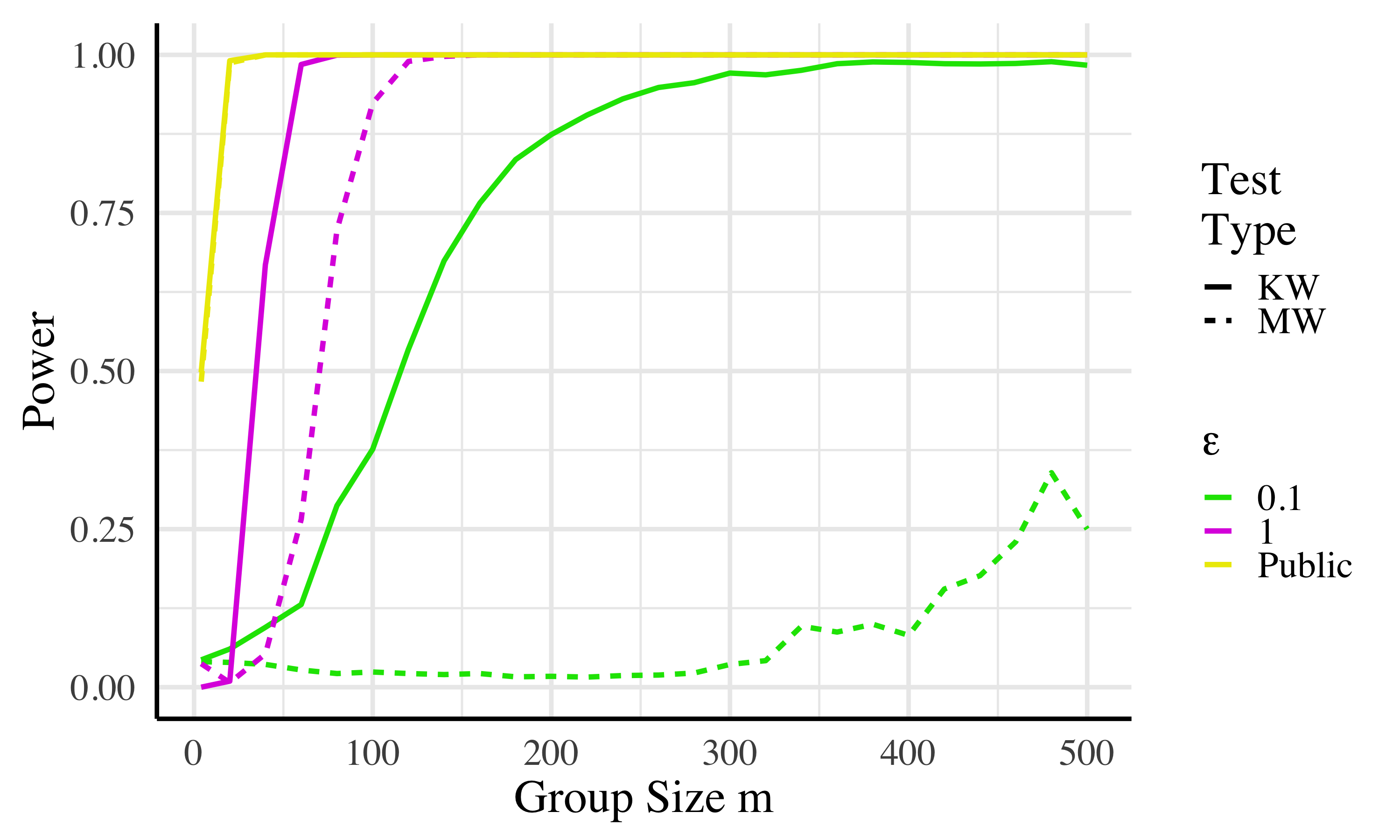}
    \caption{Power of \privMWP and \privKWPA at various values of $\epsilon_{tot}$ and group sample size $m$. (Effect size: $\mu_1 - \mu_2 = 1 \sigma$; proportion of $\epsilon_{tot}$ to $\epsilon_{m} = .65$; $\alpha=.05$; $n = 1000$; $g = 2$; normally distributed sample data)}\label{fig:kw_vs_mw_uneq_n1000}
\end{figure}

\paragraph{Known group size special case} As noted in \ref{sec:mw}, when group sizes are known publicly, the entirety of the privacy budget can be allotted to estimating the test statistic $U$ in the \privMWP algorithm since $m$ need not be estimated. This special use case increases the precision of \privMWP enough to cause the test to reach higher powers than \privKWPA at equivalent parameterizations, as shown in Figure \ref{fig:kw_vs_mw_public_n}.

\begin{figure} [!htb] 
    \centering
    \includegraphics[width=\linewidth, keepaspectratio]{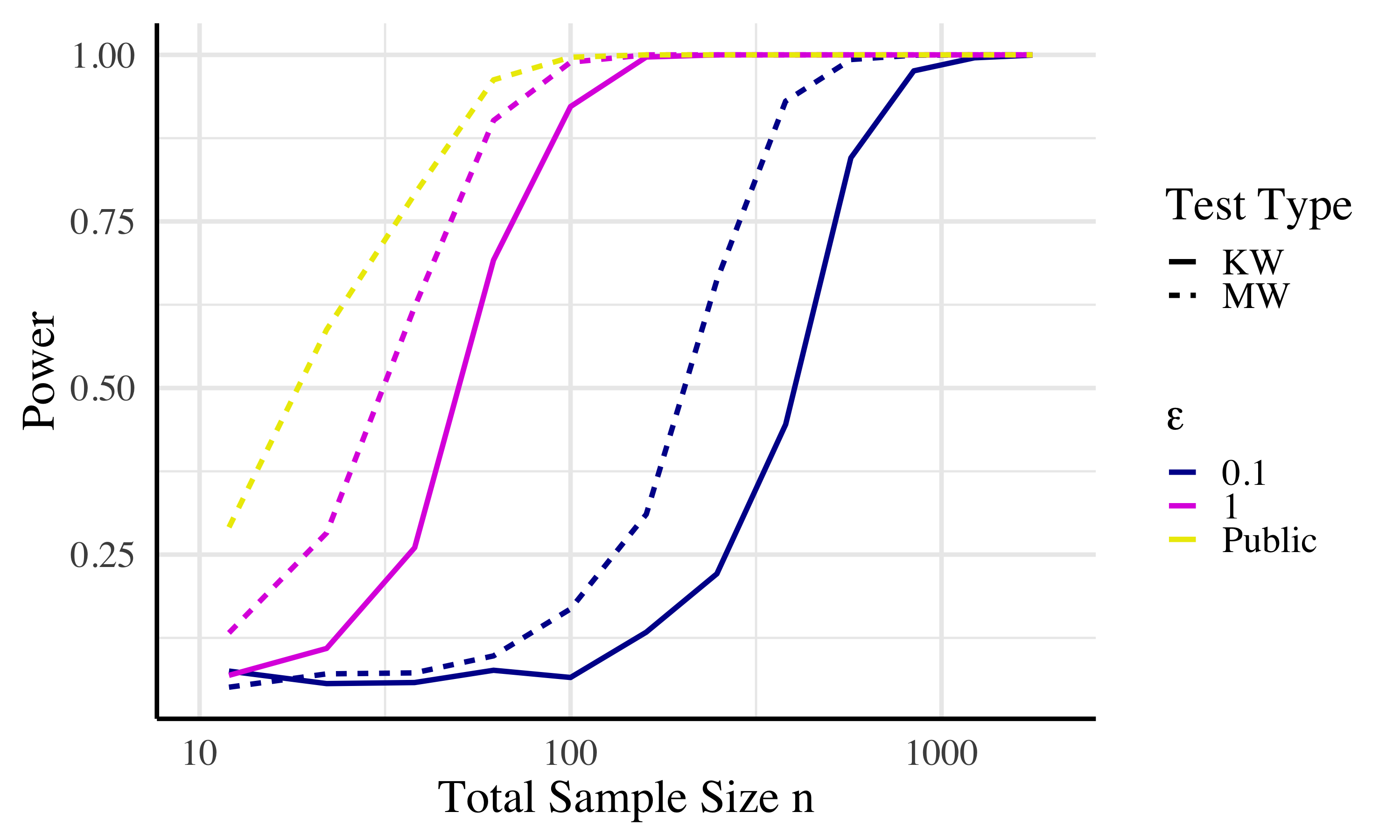}
    \caption{Power of \privMWP and \privKWPA at various values of $\epsilon_{tot}$ and sample size $m$. (Effect size: $\mu_1 - \mu_2 = 1 \sigma$; proportion of $\epsilon_{tot}$ to $\epsilon_{m} = .65$; $\alpha=.05$; $n = 1000$; $g = 2$; normally distributed sample data)}\label{fig:kw_vs_mw_public_n}
\end{figure}

\paragraph{Allotment of Privacy Budget} This section serves to justify our choice of the optimum proportion of $\epsilon_{tot}$ to allot to estimating $m$ and the test statistic $\MW$.

Initially, we fix all parameters other than the proportion of $\epsilon$ to $m$ and total sample size $n$ and calculate power across many values.

\begin{figure} [!htb] 
    \centering
    \includegraphics[width=\linewidth, keepaspectratio]{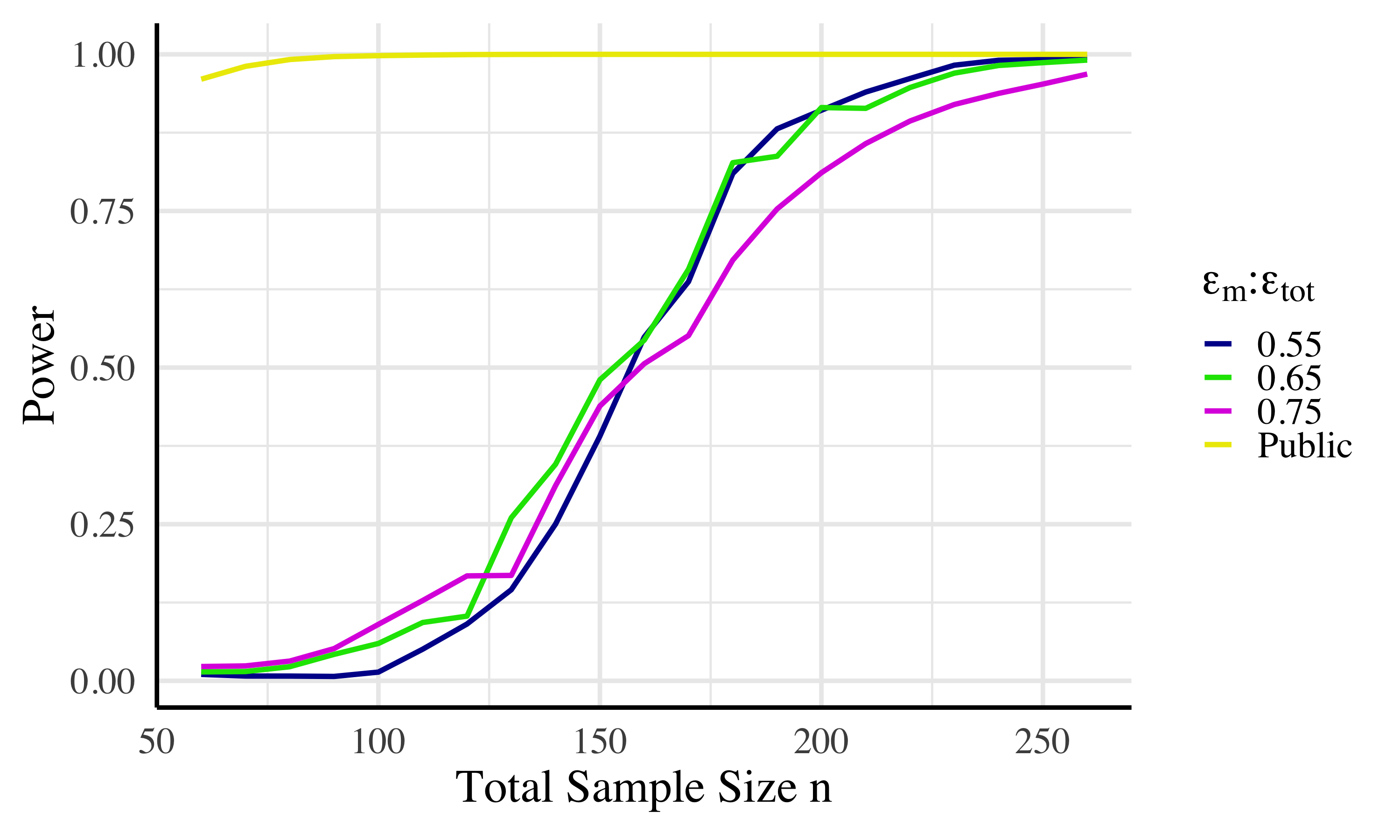}
    \caption{Power of \privMWP at various values of proportion of $\epsilon_{tot}$ to $\epsilon_{m}$ and $n$. (Effect size: $\mu_{1} - \mu_{2} = 1 \sigma$; $\epsilon_{tot} = 1$; $\alpha=.05$; $m$:($n-m$) $= 1$; normally distributed sample data)}\label{fig:var_N_prop_eps}
\end{figure}

As shown in Figure \ref{fig:var_N_prop_eps}, to maximize power, the optimum proportion of epsilon to allot to estimating $m$ is near $.65$.

\clearpage

\subsection{Wilcoxon Signed-Rank}\label{sec:wc_appendix}
\paragraph{Dealing with Ties} In Figure \ref{fig:30_pct_ties}, 30\% of rows have $d_i=0$ and the other 70\% are distributed as before, with the two data points sampled from normal distributions with means one standard deviation apart.

\begin{figure} [!htb] 
    \centering
    \includegraphics[width=\linewidth]{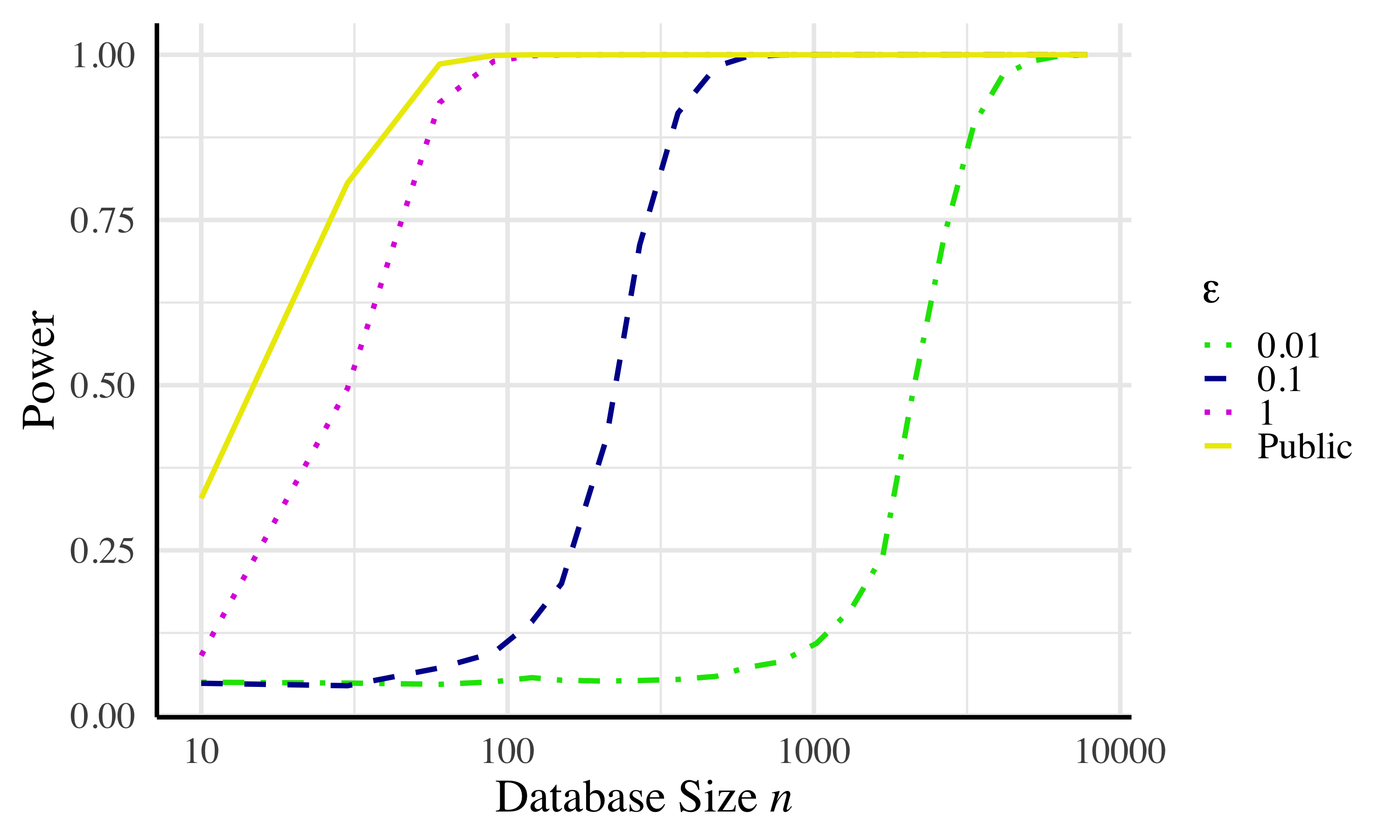}
    \caption{Power of \privWPP at various values of $\epsilon$ and $n$. (Effect size: $\mu_u - \mu_v = 1 \sigma$; $\alpha = .05$; proportion of zeroes in $d_i$ is $.3$; normally distributed sample data)}
    \label{fig:30_pct_ties}
\end{figure}

Figure \ref{fig:comp_no_ties} uses a continuous distribution for the testing data. To show that our test handles ties more effectively than the TC test, we again introduce ties into the data, this time at varying proportion, and measure power. In particular, Figures \ref{fig:ties-utility} and \ref{fig:ties-privacy} compare our algorithm to the high utility and high privacy variants of the TC test, respectively.  Here we first choose the number of rows with $d_i=0$ and then sample the remaining data points as before.  Of course, as the number of rows showing no difference increases, the power of all tests decreases, but we see that in all cases our test retains power longer.

\begin{figure} [!htb] 
    \centering
    \includegraphics[width=\linewidth]{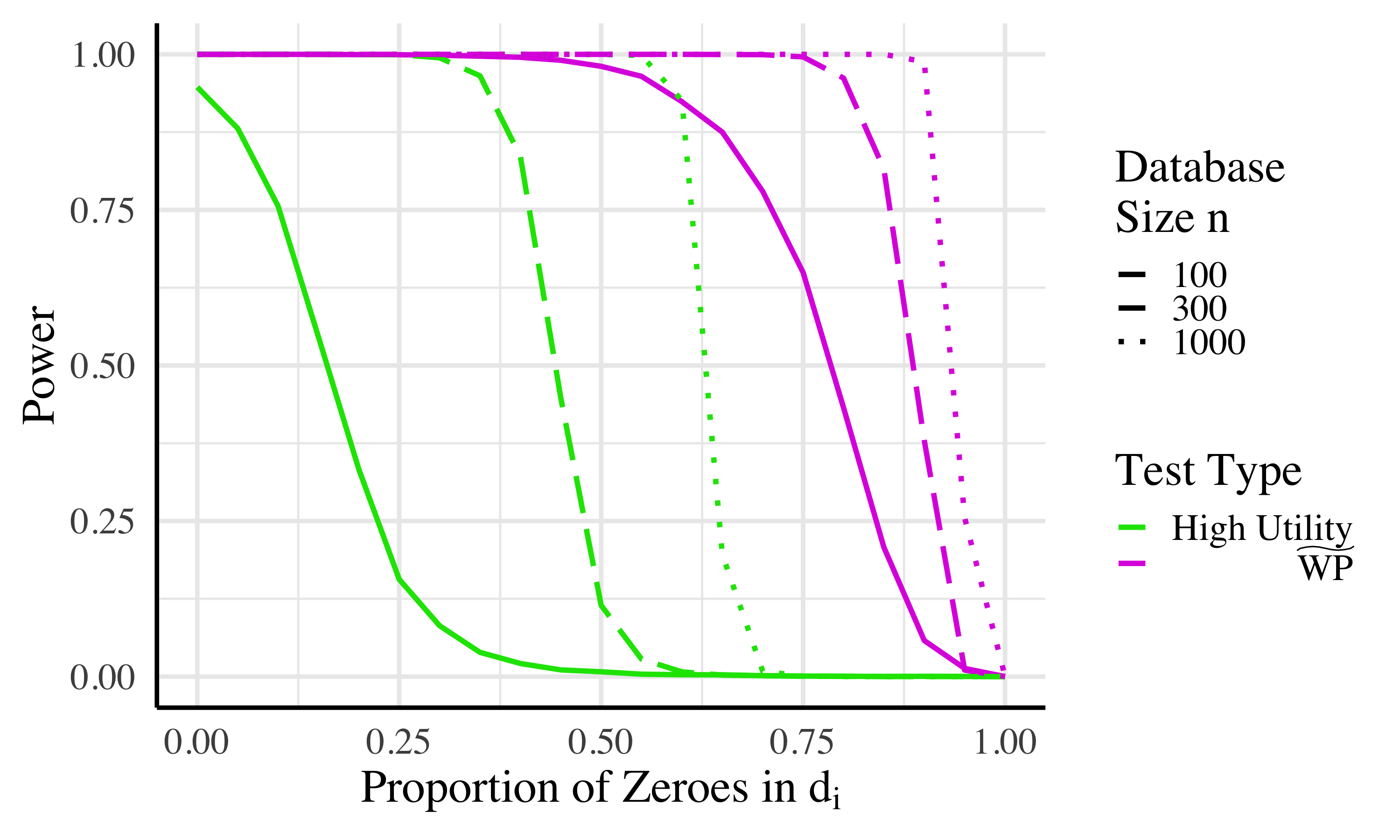}
    \caption{Power comparison of the TC \cite{DPWilcoxon} \textit{High Utility} algorithm and \privWPP at various proportions of tied values and $n$. (Effect size: $\mu_u - \mu_v = 1 \sigma$; $\epsilon = 1$; $\alpha = .05$; normally distributed sample data)}
    \label{fig:ties-utility}
\end{figure}

\begin{figure} [!htb] 
    \centering
    \includegraphics[width=\linewidth]{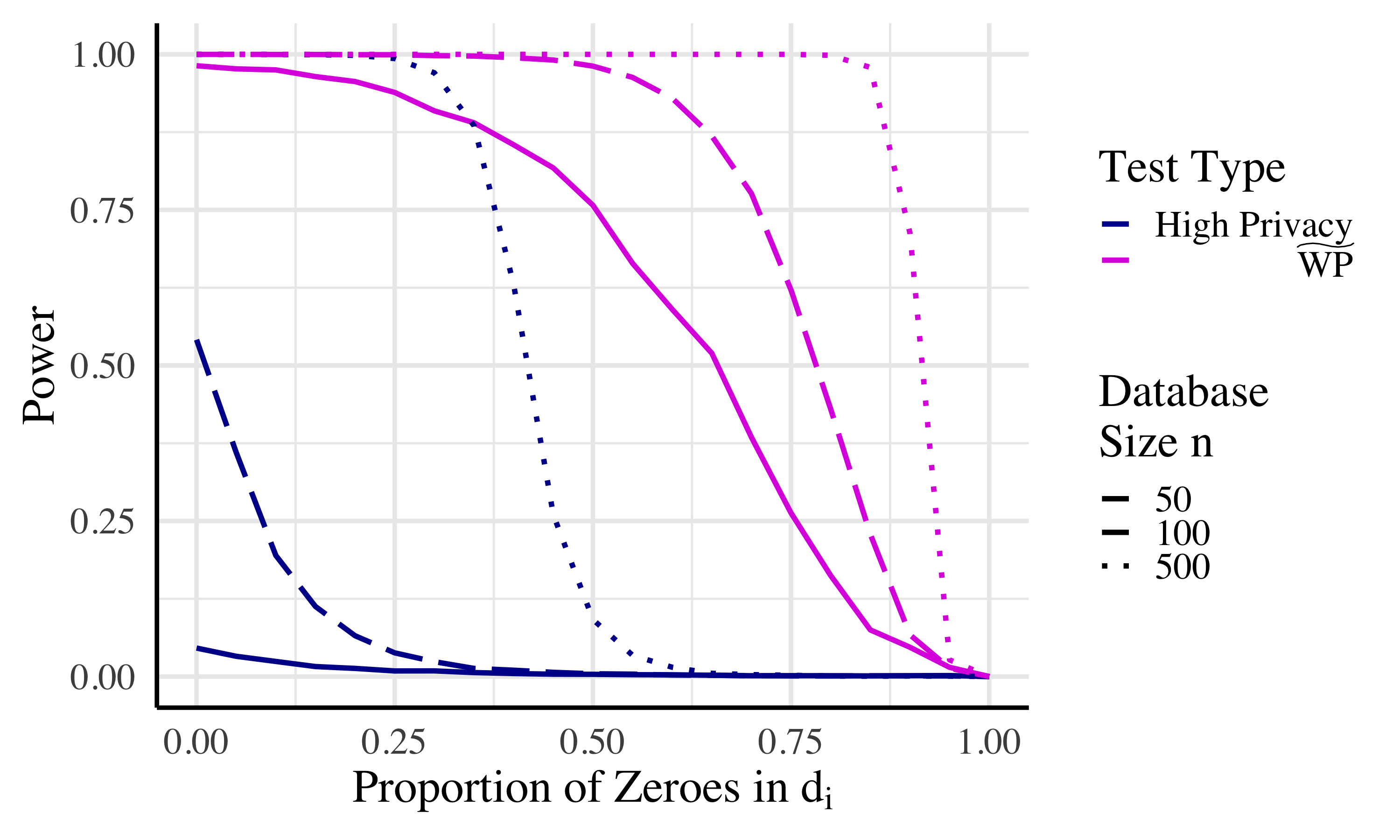}
    \caption{Power comparison of the TC  \textit{High Privacy} algorithm ($k = 15$) and \privWPP at various proportions of tied values and $n$. (Effect size: $\mu_u - \mu_v = 1 \sigma$; $\epsilon = 1$; $\alpha = .05$; normally distributed sample data)}
    \label{fig:ties-privacy}
\end{figure}

\paragraph{Varying Effect Size}\label{sec:wc_vary_effect}

We also examine the power of our test at varying effect sizes. As shown in Figure \ref{fig:varying_effect_size_wc}, at large sample sizes $n$, in this case $2500$, there is essentially no difference between the minimum effect detectable in the private and public setting.  This is because the random variation in the sample overwhelms the relatively small random noise being added for privacy.  We also note, while not shown in the figure, that for small enough choices of $n$ and $\epsilon$ no effect, no matter how large, can be detected.  (Once all $d_i$ values are positive, increasing the effect size further has no effect.)

\begin{figure} [!htb] 
    \centering
    \includegraphics[width=\linewidth]{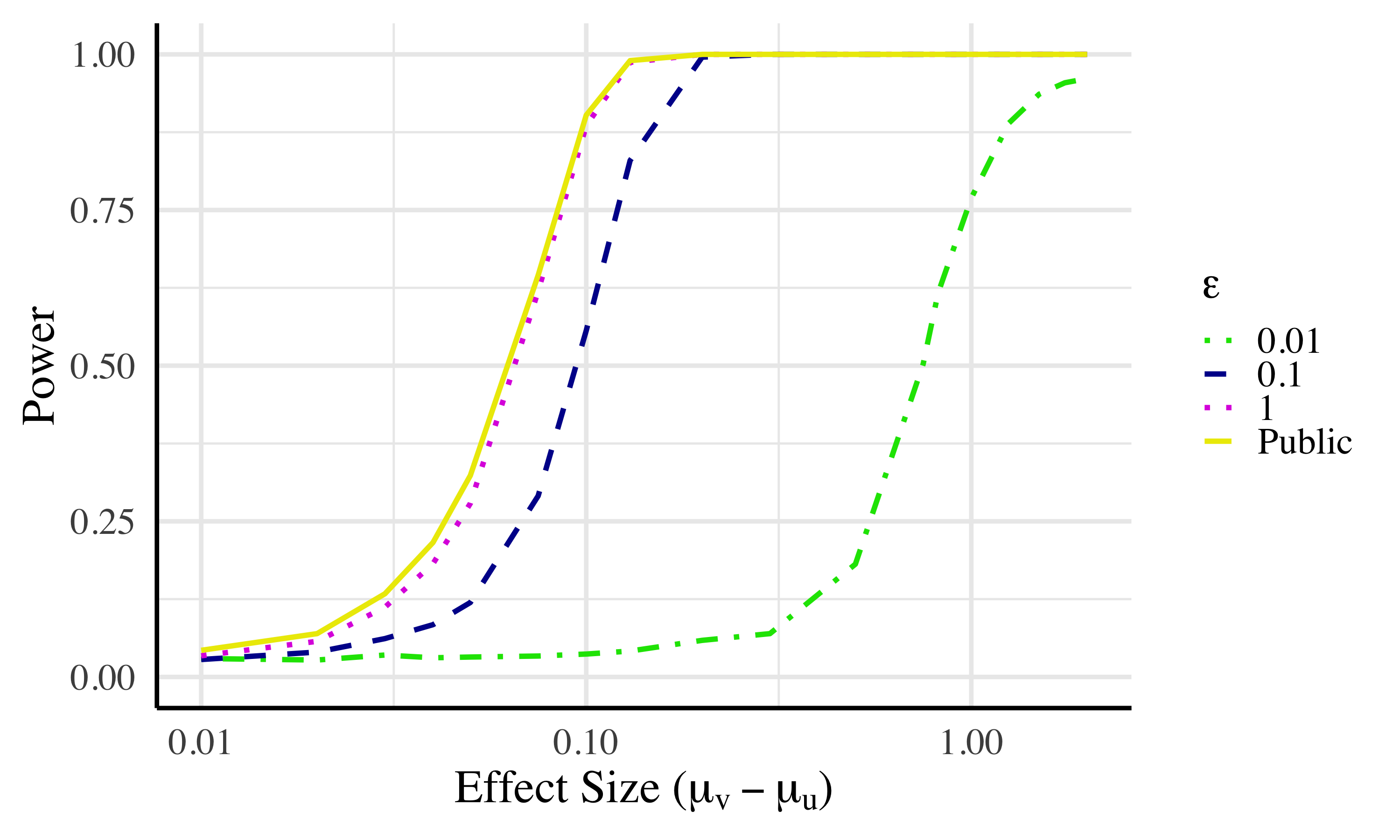}
    \caption{Power of \privWPP at various values of $\epsilon$ and effect size. ($\alpha = .05$; $n = 2500$; normally distributed sample data)}
    \label{fig:varying_effect_size_wc}
\end{figure}

\paragraph{Application to Real-World Data}\label{sec:realdata} We now demonstrate the use of our Wilcoxon algorithm on real-world data.  We use a database of New York City tax ride information released in 2014.  The database contains information on every Yellow Taxi ride in New York City in 2013, and its release resulted in high-profile de-anonymization attacks \cite{NYCData}.  Our main finding is that (at least for some natural analyses that we attempt) a differentially private query interface would have been sufficient, and the release of the data set was unnecessary.  We also again compare our statistical power to that of the TC tests and find it superior.\footnote{This data set is the same one Task and Clifton originally use for evaluating their test.}

This dataset contains several variables of interest for every Yellow Taxi ride in New York City in 2013, of which the following will be most useful:

\begin{itemize}
  \item Hack License: a unique identifier for every taxi driver in the city
  \item Number of Passengers: how many people rode together in the taxi
  \item Trip Time: the duration of the ride in seconds
  \item Trip Distance: the total distance traveled during the ride
\end{itemize}

We subsetted this dataset, initially, to include rides occuring on January 1st and 2nd of 2013. Then, we only kept rides given by a driver who drove on both the 1st and the 2nd.  In one data set, $u_i$ and $v_i$ were the average trip time on the 1st and 2nd, respectively.  The Wilcoxon test can then be used to test whether trip time varied between the two days. We calculated two other data sets similarly for trip distance and number of passengers.

The resulting data sets have 17,066 entries.  We sampled (with replacement) $10^5$ different smaller datasets of size $n=400$.  Finally, we ran each algorithm on each of these resulting data sets and report the proportion in which a significant result (at $\alpha = .05$) was found. The results are summarized in Figure \ref{fig:App}. 

\begin{figure} [!htb] 
    \centering
    \includegraphics[width=\linewidth]{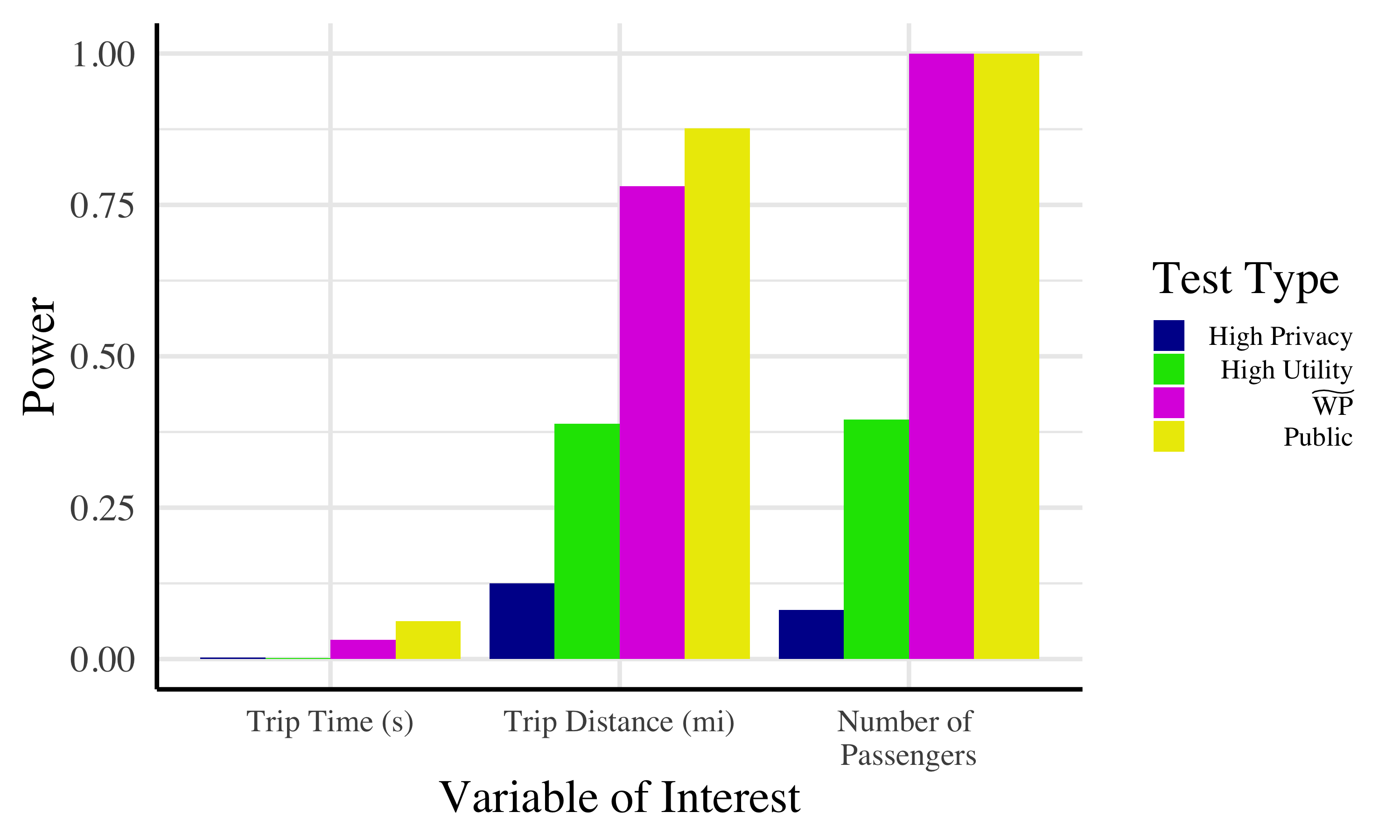}
    \caption{Power of the TC tests, \privWPP, and the public test at sample sizes of 400, sampled with replacement from the NYC taxi data. ($\epsilon = 1$; $\alpha = .05$)}
    \label{fig:App}
\end{figure}

The power of our test follows the public test closely, and we see that for a data set as small as 400, we can achieve results with a private query that are almost as useful as a full public release of the data.  The TC algorithms do not achieve this goal, though the high utility variant is still meaningfully useful.  Of course, there will be choices of $n$ and $\epsilon$ for which the gap between our power and the power of the public test is quite large, but our point here is to argue that for even reasonably small data sets a private query interface can be sufficient for many important tasks.

\paragraph{Power Comparison at $\epsilon=.1$} We compare the power of our next algorithm to the previous best at a smaller choice of $\epsilon$.

\begin{figure} [!htb] 
    \centering
    \includegraphics[width=\linewidth]{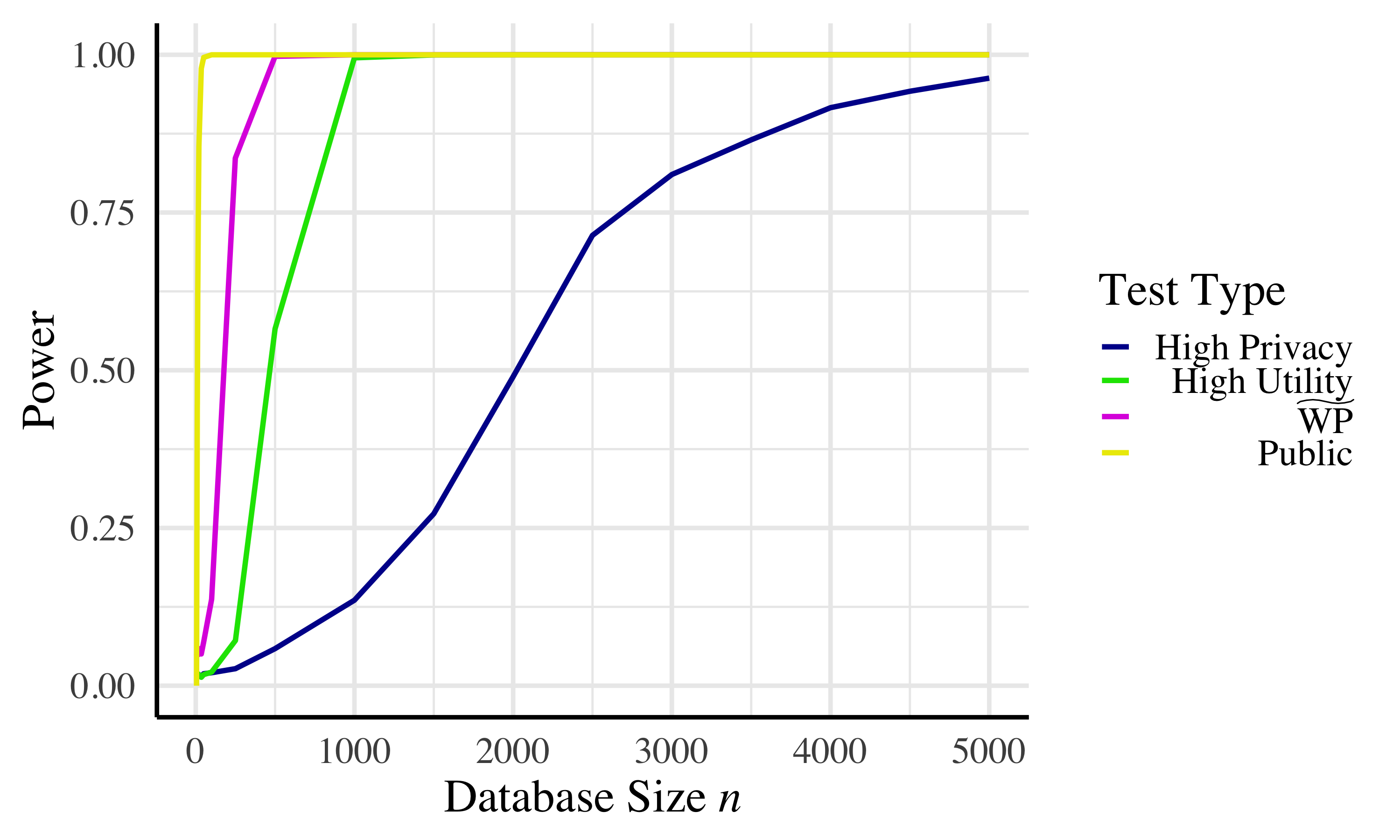}
    \caption{Power comparison of Task and Clifton's algorithms, \privWPP, and the public algorithm at various database sizes $n$. (Effect size: $\mu_u - \mu_v = 1 \sigma$; $\epsilon=.1$; $\alpha=.05$; normally distributed sample data)}
    \label{fig:comp_low_eps}
\end{figure}

\paragraph{Critical Value Tables} We present critical values for several choices of $\epsilon$ and sample size.

\begin{table}[ht]
\centering
\caption{Critical Value Comparison for $n=1000$}\label{tab:n_1000_comp}
\begin{tabular}{llrrr}
  \hline
$\epsilon$ & $\alpha$ & Public & New & TC \\ 
  \hline
1 & 0.1 & 1.282 & 1.296 & 1.763 \\ 
   & 0.05 & 1.645 & 1.665 & 2.174 \\ 
   & 0.025 & 1.960 & 1.984 & 2.594 \\ 
   \hline
0.1 & 0.1 & 1.282 & 2.203 & 5.617 \\ 
   & 0.05 & 1.645 & 2.975 & 6.028 \\ 
   & 0.025 & 1.960 & 3.740 & 6.448 \\ 
   \hline
0.01 & 0.1 & 1.282 & 17.681 & 44.157 \\ 
   & 0.05 & 1.645 & 25.234 & 44.568 \\ 
   & 0.025 & 1.960 & 32.844 & 44.988 \\ 
   \hline
\end{tabular}
\\[10pt]
\caption*{Critical values for \textit{n} = 1000 and several values of $\epsilon$ and $\alpha$.  To allow easy comparison, these values are for a normalized $W$ statistic, i.e., $W$ has been divided by the relevant constant so that it is (before the addition of Laplacian noise) distributed according to a standard normal.}
\end{table}

\begin{table}[ht]
\centering
\caption{New Critical Value Table for $\epsilon$ = $1.0$}
\label{tab:crit_val_eps_1.0}
\begin{tabular}{rrrrr}
  \hline
$n$ & 0.05 & 0.025 & 0.01 & 0.005 \\ 
  \hline
10 & 70 & 83 & 102 & 116 \\ 
  20 & 155 & 183 & 220 & 248 \\ 
  30 & 256 & 299 & 355 & 397 \\ 
  40 & 369 & 429 & 506 & 562 \\ 
  50 & 494 & 572 & 670 & 742 \\ 
  75 & 854 & 984 & 1143 & 1257 \\ 
  100 & 1271 & 1460 & 1690 & 1853 \\ 
  200 & 3402 & 3895 & 4486 & 4900 \\ 
  300 & 6127 & 7012 & 8069 & 8798 \\ 
  400 & 9335 & 10679 & 12276 & 13382 \\ 
  500 & 12978 & 14845 & 17061 & 18592 \\ 
  1000 & 36235 & 41443 & 47637 & 51906 \\ 
   \hline
\end{tabular}
\\[10pt]
\caption*{Critical values at several sample sizes $n$ and two-sided significance levels $\alpha$. To calcuate these values, we run 10 million simulations for each parameter combination and compute the $1$ $-$ $\alpha$th percentile of the absolute value of the distribution.}
\end{table}

\begin{table}[ht]
\centering
\caption{New Critical Value Table for $\epsilon$ = $0.1$}
\label{tab:crit_val_eps_0.1}
\begin{tabular}{rrrrr}
  \hline
$n$ & 0.05 & 0.025 & 0.01 & 0.005 \\ 
  \hline
10 & 600 & 739 & 922 & 1061 \\ 
  20 & 1202 & 1479 & 1846 & 2123 \\ 
  30 & 1806 & 2220 & 2770 & 3185 \\ 
  40 & 2413 & 2968 & 3704 & 4261 \\ 
  50 & 3018 & 3713 & 4628 & 5324 \\ 
  75 & 4541 & 5577 & 6954 & 7989 \\ 
  100 & 6073 & 7461 & 9294 & 10677 \\ 
  200 & 12328 & 15098 & 18767 & 21531 \\ 
  300 & 18733 & 22892 & 28391 & 32519 \\ 
  400 & 25296 & 30837 & 38193 & 43736 \\ 
  500 & 32054 & 38979 & 48128 & 55083 \\ 
  1000 & 68258 & 82120 & 100408 & 114230 \\ 
   \hline
\end{tabular}
\\[10pt]
\caption*{Critical values at several sample sizes $n$ and two-sided significance levels $\alpha$. To calcuate these values, we run 10 million simulations for each parameter combination and compute the $1$ $-$ $\alpha$th percentile of the absolute value of the distribution.}
\end{table}

\begin{table}[ht]
\centering
\caption{New Critical Value Table for $\epsilon$ = $0.01$}
\label{tab:crit_val_eps_.01}
\begin{tabular}{rrrrr}
  \hline
$n$ & 0.05 & 0.025 & 0.01 & 0.005 \\ 
  \hline
10 & 5992 & 7377 & 9209 & 10596 \\ 
  20 & 11971 & 14742 & 18416 & 21196 \\ 
  30 & 17976 & 22137 & 27644 & 31774 \\ 
  40 & 23974 & 29516 & 36877 & 42425 \\ 
  50 & 29964 & 36905 & 46081 & 53034 \\ 
  75 & 44933 & 55371 & 69105 & 79513 \\ 
  100 & 59921 & 73792 & 92066 & 106005 \\ 
  200 & 119902 & 147619 & 184222 & 212010 \\ 
  300 & 179942 & 221477 & 276678 & 317895 \\ 
  400 & 239695 & 295106 & 368374 & 423528 \\ 
  500 & 299627 & 368763 & 460256 & 529522 \\ 
  1000 & 600096 & 738071 & 921529 & 1061150 \\ 
   \hline
\end{tabular}
\\[10pt]
\caption*{Critical values at several sample sizes $n$ and two-sided significance levels $\alpha$. To calcuate these values, we run 10 million simulations for each parameter combination and compute the $1$ $-$ $\alpha$th percentile of the absolute value of the distribution.}
\end{table}

\clearpage

\subsection{T-Test}\label{sec:t_appendix}
\paragraph{Allotment of Privacy Budget} This section serves to briefly justify our choice of the optimum proportion of the privacy budget to allot to estimating $\widehat{\bar{x}}$ and $\widehat{s^2}$.

\begin{figure} [!htb] 
    \centering
    \includegraphics[width=\linewidth, keepaspectratio]{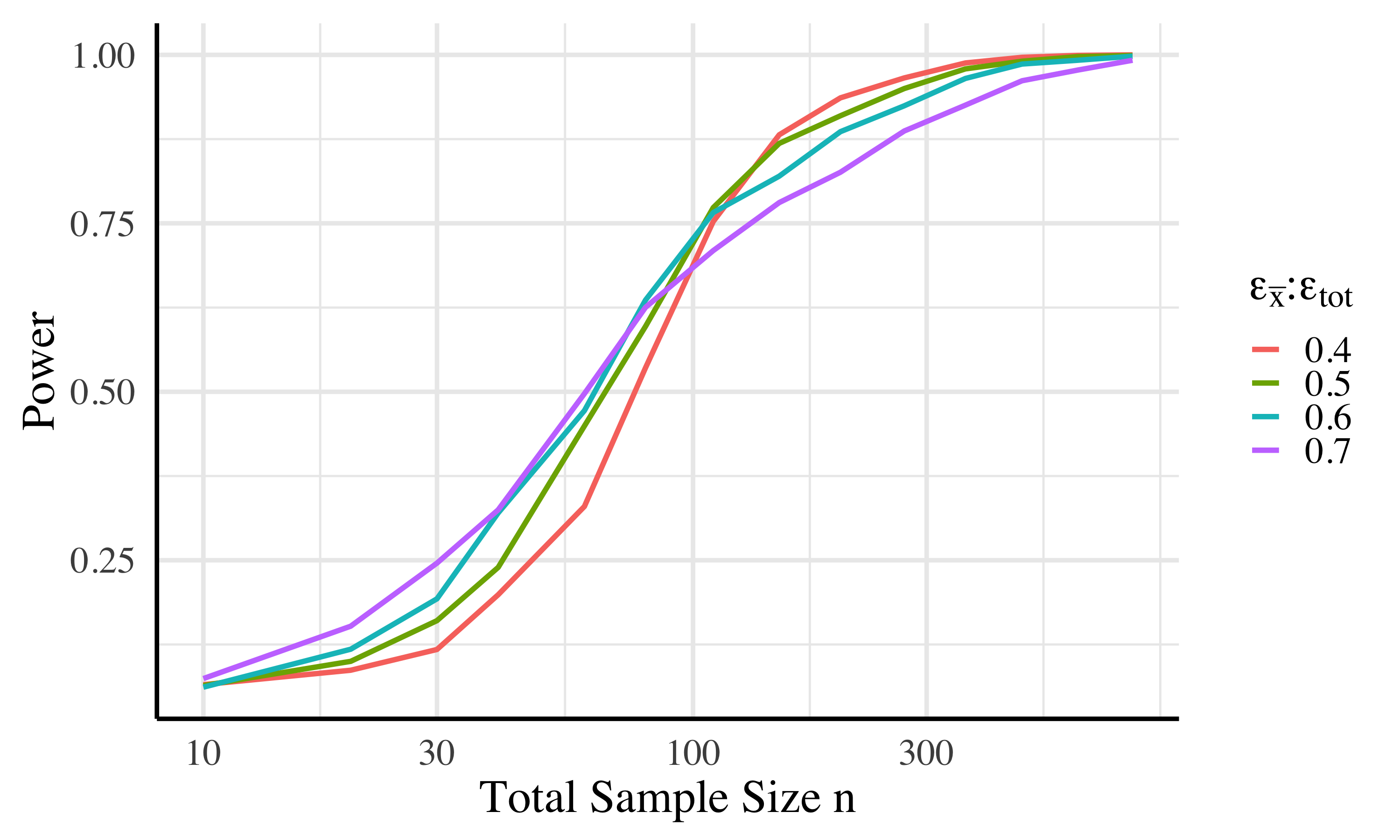}
    \caption{Power at various values of proportion of the total privacy budget allotted to $\widehat{\bar{x}}$ and $\widehat{s^2}$. (Effect size: $\mu_{1} - \mu_{2} = 1 \sigma$; $\epsilon = 1$; $\alpha=.05$}\label{fig:new_t_eps_allotment}
\end{figure}

As evidenced in Figure \ref{fig:new_t_eps_allotment}, unlike \privMWP, the optimum proportion of $\epsilon$ to allot to either statistic is not obvious; certain proportions offer distinct advantages depending on parameterizations. To most objectively choose our allotment of the privacy budget, we refer to the common 80\% power rule used in determining sample size requirements in biological studies and consider the $\epsilon$ allotment that reaches 80\% power most quickly to be optimum. Resultantly, we choose to allot equal proportions of the total privacy budget to estimating $\widehat{\bar{x}}$ and $\widehat{s^2}$.

\end{document}